\documentclass{article}

\usepackage[vmargin=1.18in,hmargin=1.1in]{geometry}
\usepackage{amsmath,amssymb,amsfonts,graphicx,amsthm,nicefrac,mathtools,bm}
\usepackage{hyperref}
\usepackage[numbers]{natbib}
\usepackage[english]{babel}
\usepackage[T1]{fontenc}
\usepackage{bbm,mdframed}
\usepackage{xcolor}
\usepackage{xspace}
\usepackage{rotating,multirow}
\usepackage{enumerate,graphics,enumitem}
\usepackage[algo2e]{algorithm2e}

\usepackage{gensymb}

\setlist[itemize]{noitemsep, topsep=0pt}
\setlist[enumerate]{itemsep=5pt, topsep=5pt, leftmargin=25pt}

\newtheorem{theorem}{Theorem}

\definecolor{verylightblue}{rgb}{0.7,0.8,1}
  {\begin{mdframed}[backgroundcolor=verylightblue]\begin{theorem}}%
  {\end{theorem}\end{mdframed}}

\definecolor{verylightgray}{gray}{0.95}
  {\begin{mdframed}[backgroundcolor=verylightgray]\begin{proof}}%
  {\end{proof}\end{mdframed}}

\newtheorem{lemma}{Lemma}
\definecolor{verylightred}{rgb}{1,0.8,0.8}
  {\begin{mdframed}[backgroundcolor=verylightred]\begin{lemma}}%
  {\end{lemma}\end{mdframed}}

\newtheorem{proposition}{Proposition}
  {\begin{mdframed}[backgroundcolor=verylightblue]\begin{proposition}}%
  {\end{proposition}\end{mdframed}}

\newtheorem{corollary}{Corollary}

\newtheorem{assumption}{Assumption}
\theoremstyle{definition}

\theoremstyle{remark}

\newtheorem{example}{Example}

\makeatletter

\newtheorem*{rep@theorem}{\rep@title}
\newcommand{\newreptheorem}[2]
{\newenvironment{rep#1}[1]
{\def\rep@title{#2 \ref{##1}} \begin{rep@theorem}}%
 {\end{rep@theorem}}}
\makeatother
\newreptheorem{theorem}{Theorem}
\newreptheorem{lemma}{Lemma}
\newreptheorem{corollary}{Corollary}
\newreptheorem{proposition}{Proposition}

\definecolor{DarkRed}{rgb}{0.5,0.1,0.1}
\definecolor{DarkBlue}{rgb}{0.1,0.1,0.5}

\DeclareMathOperator*{\argmax}{arg\,max}
\DeclareMathOperator*{\argmin}{arg\,min}

\newcommand{\E}{\mathbb{E}}

\def\R{\mathbb{R}}

\newcommand{\ignore}[1]{}

\let\emptyset\varnothing

\usepackage{tikz}
\usetikzlibrary{arrows}
\usetikzlibrary{shapes}

\newcommand{\B}{\mathcal{B}}
\newcommand{\cN}{\mathcal{N}}

\newcommand{\thedate}{\today}
\newcommand{\theauthor}{Tijana Zrnic \quad \quad William Fithian\\\texttt{\{tijana.zrnic,wfithian\}@berkeley.edu} \\ \\ University of California, Berkeley}
\newcommand{\thetitle}{Locally Simultaneous Inference}

\date{\thedate}
\author{\theauthor}
\title{\thetitle}

\usepackage[mmddyy]{datetime}

\def\P{\mathcal{P}}
\def\V{\mathcal{V}}

\def\M{\mathcal{M}}

\def\F{\mathcal{F}}
\def\D{\mathcal{D}}

\def\ci{C}

\makeatletter
\newcommand{\dotfrac}[2]{
\mathchoice
{\ooalign{$\genfrac{}{}{0pt}{0}{#1}{#2}$\cr\leavevmode\cleaders\hb@xt@ .22em{\hss $\displaystyle\cdot$\hss}\hfill\kern\z@\cr}}
{\ooalign{$\genfrac{}{}{0pt}{1}{#1}{#2}$\cr\leavevmode\cleaders\hb@xt@ .22em{\hss $\textstyle\cdot$\hss}\hfill\kern\z@\cr}}
{\ooalign{$\genfrac{}{}{0pt}{2}{#1}{#2}$\cr\leavevmode\cleaders\hb@xt@ .22em{\hss $\scriptstyle\cdot$\hss}\hfill\kern\z@\cr}}
{\ooalign{$\genfrac{}{}{0pt}{3}{#1}{#2}$\cr\leavevmode\cleaders\hb@xt@ .22em{\hss $\scriptscriptstyle\cdot$\hss}\hfill\kern\z@\cr}}
}
\makeatother

\usepackage{algorithm,algorithmic}

\makeatletter
\long\def\@makecaption#1#2{
        \vskip 0.8ex
        \setbox\@tempboxa\hbox{\small {\bf #1:} #2}
        \parindent 1.5em  
        \dimen0=\hsize
        \advance\dimen0 by -3em
        \ifdim \wd\@tempboxa >\dimen0
                \hbox to \hsize{
                        \parindent 0em
                        \hfil 
                        \parbox{\dimen0}{\def\baselinestretch{0.96}\small
                                {\bf #1.} #2
                                } 
                        \hfil}
        \else \hbox to \hsize{\hfil \box\@tempboxa \hfil}
        \fi
        }
\makeatother

\usepackage[algo2e]{algorithm2e}

\begin{document}

\maketitle

\begin{abstract}
Selective inference is the problem of giving valid answers to statistical questions chosen in a data-driven manner. A standard solution to selective inference is \emph{simultaneous inference}, which delivers valid answers to the set of all questions that could possibly have been asked. However, simultaneous inference can be unnecessarily conservative if this set includes many questions that were unlikely to be asked in the first place.
We introduce a less conservative solution to selective inference that we call {\em locally simultaneous inference}, which only answers those questions that could {\em plausibly} have been asked in light of the observed data, all the while preserving rigorous type I error guarantees. For example, if the objective is to construct a confidence interval for the ``winning'' treatment effect in a clinical trial with multiple treatments, and it is obvious in hindsight that only one treatment had a chance to win, then our approach will return an interval that is nearly the same as the uncorrected, standard interval. Locally simultaneous inference is implemented by refining any method for simultaneous inference of interest. Under mild conditions satisfied by common confidence intervals, locally simultaneous inference \emph{strictly dominates} its underlying simultaneous inference method, meaning it can never yield less statistical power but only more. Compared to conditional selective inference, which demands stronger guarantees, locally simultaneous inference is more easily applicable in nonparametric settings and is more numerically stable.
\end{abstract}

\section{Introduction}

Modern scientific investigations increasingly involve choosing inferential questions of interest only \emph{after} seeing the data. While this practice offers more freedom to the scientist than the traditional paradigm of specifying the relevant questions up front, it is by now well understood that it also creates undesirable selection bias, thereby invalidating type I error guarantees of classical statistical methods. The area of \emph{selective inference} formally describes this problem and offers rigorous solutions across a variety of settings.

One standard and broadly applicable solution is to perform \emph{simultaneous inference}, i.e., deliver valid answers to all questions that could possibly be asked.
However, simultaneous inference can be unnecessarily conservative when many questions, although possible, are unlikely to be of interest in the first place. For example, suppose that a clinical trial estimates the effectiveness of multiple treatments and, after observing the data, it is clear that there are many ineffective treatments and only a handful of effective ones. Even if we are only interested in constructing a confidence interval for the effectiveness of the best-performing treatment, simultaneous inference would still widen our intervals enough to cover all possible treatments, including the clearly ineffective ones that never stood a chance of being selected. Similarly, if we use a method like the LASSO to select a sparse subset of variables for a linear model, it may be clear in hindsight that most variables had no chance of being selected. In both of these cases, simultaneous inference can be overly conservative.

In this work we introduce \emph{locally simultaneous inference}, a method for valid selective inference that only answers those questions that could \emph{plausibly} have been selected in light of the observed data. Locally simultaneous inference comes with rigorous type I error guarantees like simultaneous inference but is less conservative; in particular, it reduces to standard simultaneous inference in an extreme case. In the clinical trial example, locally simultaneous inference would require taking a correction only over reasonably effective treatments, while testing the ineffective ones comes virtually for free.

To sketch our main idea, suppose that we have a family of estimands $\{\theta_\gamma: \gamma\in \Gamma\}$, where $\Gamma$ indexes all possible targets of inference. In the running example, $\Gamma$ would be the set $\{1,\dots,m\}$, where $m$ is the total number of treatments in the clinical trial, and $\theta_\gamma$ is the mean effect of the indexed treatment. Given data $y$, we are interested in doing inference on $\theta_{\hat\gamma}$, where $\hat\gamma$ is a data-dependent target chosen from $\Gamma$; in the running example, $\hat \gamma$ indexes the treatment that seems most effective according to $y$. It would be invalid to reuse the same data $y$ to perform an uncorrected inference on $\theta_{\hat \gamma}$, since the winning treatment is likely to have been overestimated by the trial data. However, if it is clear in hindsight that only a small number $k \ll m$ of the treatments were even in the running to win, it would be wasteful to make the full multiplicity correction for all $m$ treatments.

The main idea behind our framework is to find a \emph{data-dependent} set of targets $\widehat \Gamma^+$, which is nested between the selected target and all possible targets, $\hat \gamma\in \widehat \Gamma^+ \subseteq \Gamma$, such that taking a standard simultaneous correction over $\widehat\Gamma^+$ ensures valid selective inferences. Perhaps surprisingly, this strategy is valid despite the dependence between $\widehat\Gamma^+$ and the data. Moreover, if the selection $\hat\gamma$ is ``obvious enough'' in hindsight, $\widehat\Gamma^+$ only contains $\hat\gamma$ and our approach nearly reduces to classical, uncorrected inference.

Unlike standard simultaneous inference methods, our approach adapts to the specifics of the selection criterion; in this sense, locally simultaneous inference
resembles \emph{conditional} selective inference, which delivers valid inference after conditioning on the event that a specific target was selected. However, since our approach builds on the robust and broadly applicable principle of simultaneous inference, it comes with several advantages over conditional inference, including numerical stability and robustness to parametric assumptions.

\begin{figure}[t]
    \centering
    \includegraphics[width=0.5\textwidth]{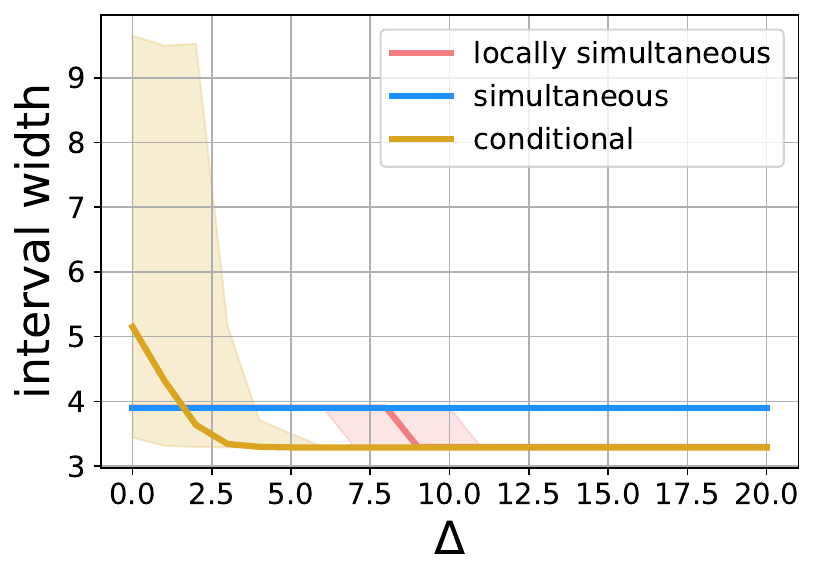}
    \caption{Interval width achieved by locally simultaneous inference, fully simultaneous inference, and conditional inference. The data is $(y_1,y_2)\sim \mathcal{N}((\mu_1,\mu_2),I_2)$ and the goal is to do inference on the mean of observation $\hat \gamma = \argmax_{\gamma\in\{1,2\}}y_\gamma$. We vary $\Delta = \mu_2 - \mu_1$. Details are in Appendix~\ref{app:figure1}.}
    \label{fig:intro_comparison}
\end{figure}

We provide a detailed comparison of locally simultaneous inference to relevant baselines in the next section. For now, we give a glimpse of the comparison using a simple illustrative example. Suppose that $y_1\sim\mathcal{N}(\mu_1,1), y_2\sim \mathcal{N}(\mu_2,1)$ are independent and we wish to do inference on the mean of observation $\hat \gamma = \argmax_{\gamma\in\{1,2\}} y_\gamma$. When the gap $\Delta = \mu_2 - \mu_1$ is near zero, the inferential question of interest is most uncertain, while large $\Delta$ corresponds to the case where the inferential question of interest is ``obvious''.
In Figure \ref{fig:intro_comparison} we plot the median, together with the $5\%$ and $95\%$ percentile, of the width of selective confidence intervals constructed via locally simultaneous inference, standard simultaneous inference, and conditional inference. We observe that for small $\Delta$ conditional inference can lead to large intervals, and as $\Delta$ grows conditional intervals approach nominal, unadjusted intervals. Simultaneous inference is insensitive to the value of $\Delta$ and delivers constant-width intervals, which are smaller than conditional intervals for small $\Delta$ due to their unconditional nature. Locally simultaneous inference adapts to the certainty of the selection like conditional inference, but is never worse than simultaneous inference.
Formally, by relying on our general theory of locally simultaneous inference, in Appendix \ref{app:figure1} we justify the following approach. Fix $\alpha = 0.1$.
Let $q^{\delta}(k)$ be the $1-\delta$ quantile of $\max_{i\in[k]}|Z_i|, Z_i\stackrel{\mathrm{i.i.d.}}{\sim}\mathcal{N}(0,1)$. Then, we have
$$P\left\{\mu_{\hat \gamma} \in \left(y_{\hat\gamma} \pm \min\{q^{0.95\alpha}(|\widehat\Gamma^+|), q^\alpha(2)\} \right)\right\}\geq 1- \alpha,$$
where $\widehat\Gamma^+ =\{1, 2\}$ when $|y_2 - y_1|\leq 2\sqrt{2}q^{0.05\alpha}(1)$ and $\widehat\Gamma^+ =\{\hat\gamma\}$ otherwise. Therefore, when $|y_2 - y_1|$ is small, locally simultaneous intervals are equal to simultaneous intervals; when $|y_2 - y_1|$ is large, only $\hat\gamma$ is deemed to be a plausible selection in hindsight, making the intervals essentially uncorrected.

\subsection{Related work}

Most existing solutions to the problem of selective inference fall under one of two categories: \emph{simultaneous} approaches and \emph{conditional} approaches. In what follows we compare locally simultaneous inference to these two families of solutions and briefly discuss the relationship to other related threads in the literature.

\paragraph{Simultaneous approaches.}
The basic principle of simultaneous approaches is to ensure simultaneous inference---the criterion of ensuring valid inferences for \emph{all} questions that could possibly be asked. More formally, if we denote by $C_{\gamma}$ a confidence region for target $\gamma\in\Gamma$, then this basic principle is captured by the inequality:
$$P\{\theta_{\hat\gamma}\not\in C_{\hat\gamma}\} \leq P\{\exists \gamma\in\Gamma:\theta_{\gamma}\not\in C_{\gamma}\}.$$
Simultaneous inference asks for $C_{\gamma}$ so that the right-hand side is bounded at a pre-specified level $\alpha\in(0,1)$; we call any method that meets this requirement a simultaneous inference method. Notice that the right-hand side has no dependence on $\hat\gamma$. Indeed, simultaneous approaches ensure valid selective inference in a selection-agnostic manner and as such are broadly applicable. Canonical examples of simultaneous inference methods include the Bonferroni correction, Holm's procedure \cite{holm1979simple}, and other related extensions \cite{hochberg1988sharper, hommel1988stagewise, strassburger2008compatible}. In the context of multivariate normal observations, simultaneous inference typically relies on estimating quantiles of the maximal z- or t-statistic~\cite{hothorn2008simultaneous, berk2013valid, genz1992numerical, genz1999numerical, bretz2001numerical}. See \cite{dickhaus2014simultaneous} for an overview of simultaneous inference methods.

Locally simultaneous inference is a principle for refining any simultaneous inference method by asking for a correction only over a smaller, \emph{data-dependent} set of targets $\widehat \Gamma^+$, rather than the whole set $\Gamma$. As such, locally simultaneous inference can itself be seen as a simultaneous inference method, one that adaptively focuses all statistical power on the data-dependent set $\widehat \Gamma^+$ and returns infinite confidence sets for all other targets $\gamma'\in \Gamma\setminus \widehat \Gamma^+$. Throughout the manuscript, when we compare locally simultaneous inference with ``simultaneous inference,'' we are implicitly referring to the simultaneous inference method that underlies the locally simultaneous construction. As we will see, locally simultaneous inference is compatible with most common simultaneous inference methods. Moreover, when we compare the power of the two approaches, we are implicitly only comparing power over the selected targets.

In its simplest form, locally simultaneous inference achieves the power gain by splitting the error budget $\alpha$ into two parts: the first is used to reduce $\Gamma$ to a smaller set of plausible targets $\widehat \Gamma^+$ and the remainder is used for simultaneous inference over $\widehat \Gamma^+$. When the whole error budget is saved for the second step, locally simultaneous inference reduces to standard simultaneous inference. If there are many candidate selections that are sufficiently implausible, then $\widehat \Gamma^+$ can be much smaller than $\Gamma$ and, consequently, locally simultaneous inference can be a significantly more powerful alternative to simultaneous inference.

We present the version with error splitting as our main proposal because it is easy to implement for a broad range of selection problems and the loss due to the split is negligible for practical choices of parameters. However, we also provide a refined version of locally simultaneous inference that comes at no cost. Formally, we show that it \emph{strictly dominates} simultaneous inference: its confidence regions are never larger and can be strictly smaller.

\begin{example}
\label{ex:timeseries}
In Figure \ref{fig:sim_vs_local} we compare standard simultaneous inference and locally simultaneous inference in a time series problem. We seek to construct a confidence band around the mean of a Gaussian process whenever its value exceeds a fixed critical threshold $T$. The set of possible targets is the set of all time step indices, $\Gamma = \{1,\dots,1000\}$, and we are interested in doing inference on $\widehat\Gamma = \{\gamma\in\Gamma: y_\gamma \geq T\}$, where $y$ is the observed time series. Simultaneous inference constructs a confidence band that contains the whole series with high probability (region between blue dashed lines) while locally simultaneous inference creates a confidence band valid only for the selected target (pink shaded region). The width of the locally simultaneous band is computed by taking a simultaneous correction only over sections of the time series that are close to exceeding $T$ and is thus smaller than the width of the simultaneous band (approximately by a factor of $1.5$).
\end{example}

\begin{example}
\label{ex:chi}
Suppose that we observe $y\sim \mathcal{N}(\mu,I_d)$ and consider estimands of the form $\theta_\gamma = \gamma^\top \mu$ for $\gamma$ chosen from the unit sphere, $\Gamma = \mathcal{S}^{d-1}$. We wish to construct a confidence interval for the data-dependent estimand $\theta_{\hat\gamma}$, where $\hat\gamma = \frac{y}{\|y\|_2}$. 
Fully simultaneous inference is required to cover all of $\{\theta_\gamma:\gamma\in\mathcal S^{d-1}\}$; this is achieved by Scheff\'e's method. Specifically, Scheff\'e's method delivers intervals of the form $C_\gamma = (\gamma^\top y \pm w)$, where $w$ is of the order $\sim \sqrt{d}$.
Locally simultaneous inference exploits the intuition that, when $y$ is far from the origin, $\hat \gamma$ seems virtually deterministic in hindsight; in particular, one would need to perturb $y$ significantly in order to incur a noticeable change in $\hat \gamma$.
More formally, we show that it in fact suffices to take a correction only over $\widehat\Gamma^+ = \{\gamma\in\mathcal {S}^{d-1}: \angle (\gamma,y) \leq \delta_d(\|y\|_2)\}$, where $\delta_d(\|y\|_2)$ is an angle determined by $\|y\|_2$ and $d$. Importantly, as $\|y\|_2\rightarrow\infty$, we get $\delta_d(\|y\|_2)\rightarrow 0$ and inference reduces to standard, uncorrected inference. In Figure \ref{fig:selective_chi}, we plot $\delta_d(\|y\|_2)$ and the resulting interval width divided by the width implied by a fully simultaneous correction for different $d$. We also illustrate the comparison between $\widehat \Gamma^+$ and $\mathcal S^{d-1}$. We include a formal expression for $\delta_d(\|y\|_2)$ with a proof in Appendix~\ref{app:chi_section}.

\begin{figure}[t]
    \centering
    \includegraphics[width=0.9\textwidth]{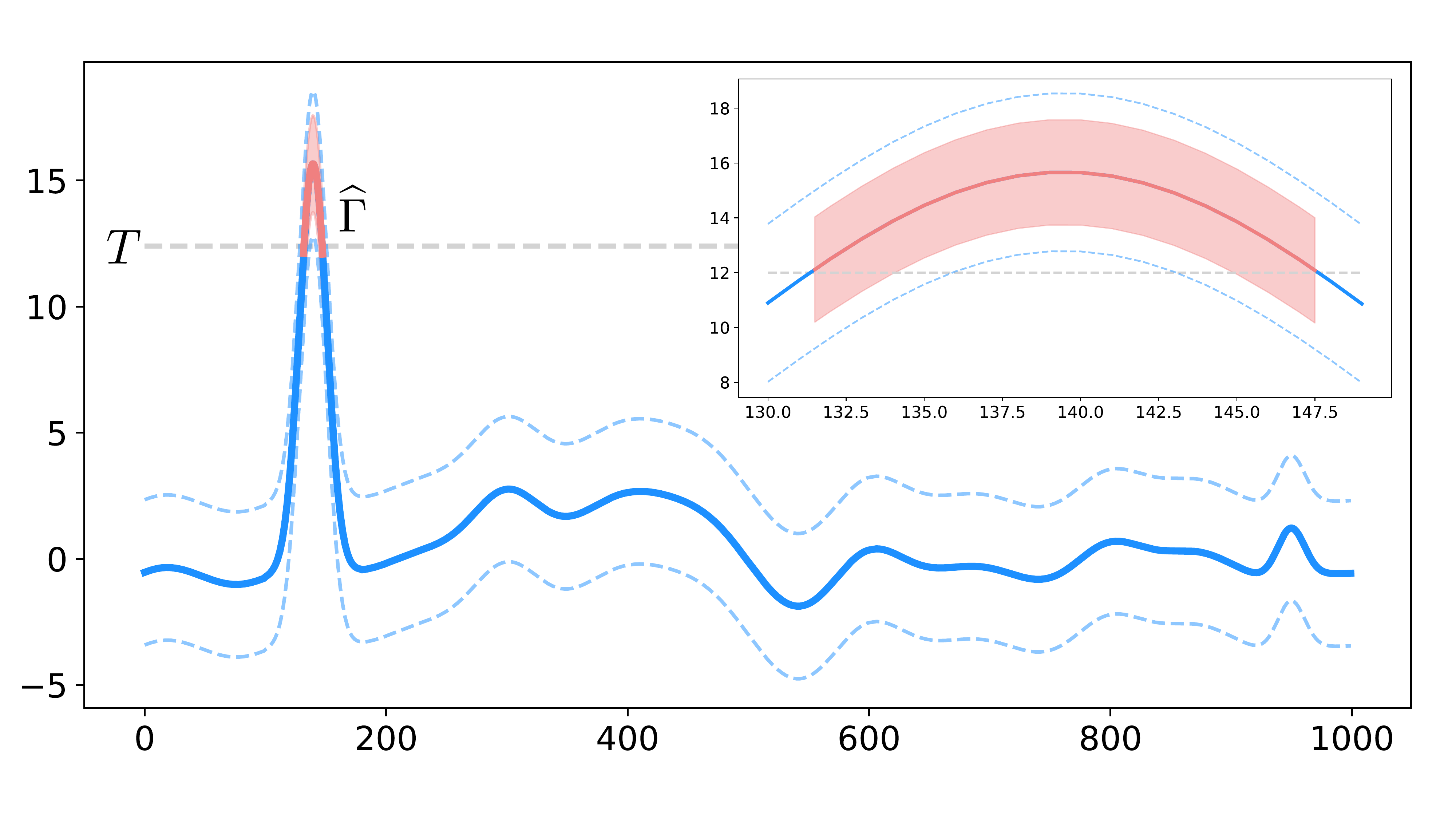}
    \caption{Comparison of fully simultaneous inference (region between blue dashed lines) and locally simultaneous inference (pink shaded region). The panel in the upper right corner zooms in around the selected region. The width of the locally simultaneous band is approximately $1.5$ times smaller.}
    \label{fig:sim_vs_local}
\end{figure}

\begin{figure}[t]
    \centering
    \includegraphics[width=0.3\textwidth]{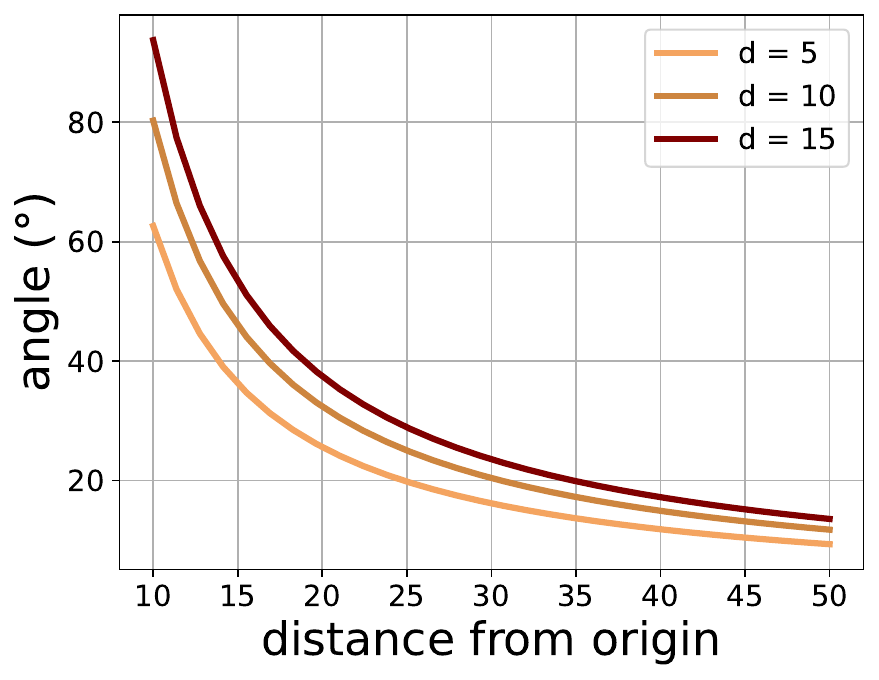}\hspace{0.5cm} \includegraphics[width=0.3\textwidth]{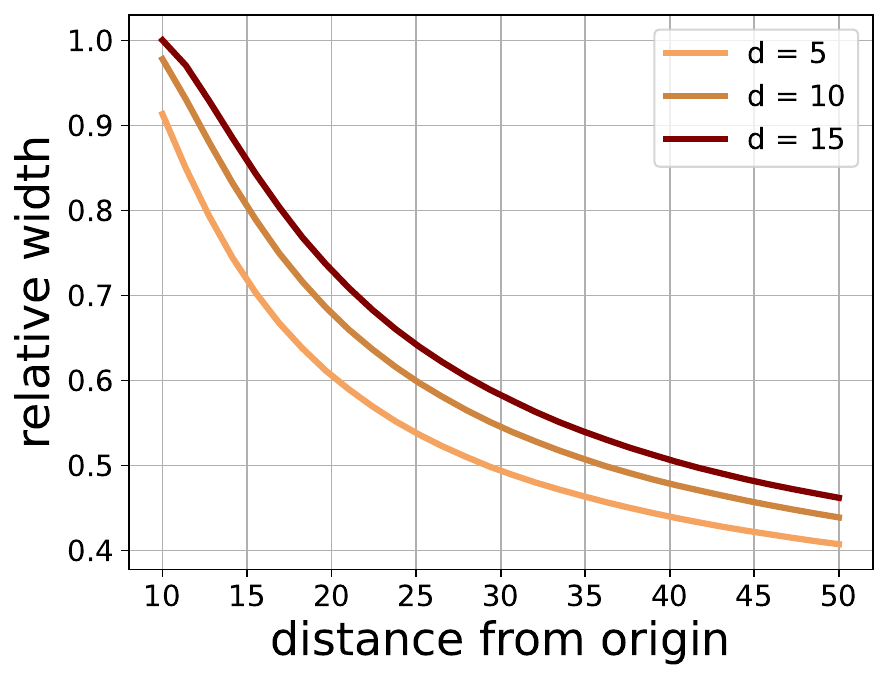}\hspace{0.5cm}
    \includegraphics[width=0.3\textwidth]{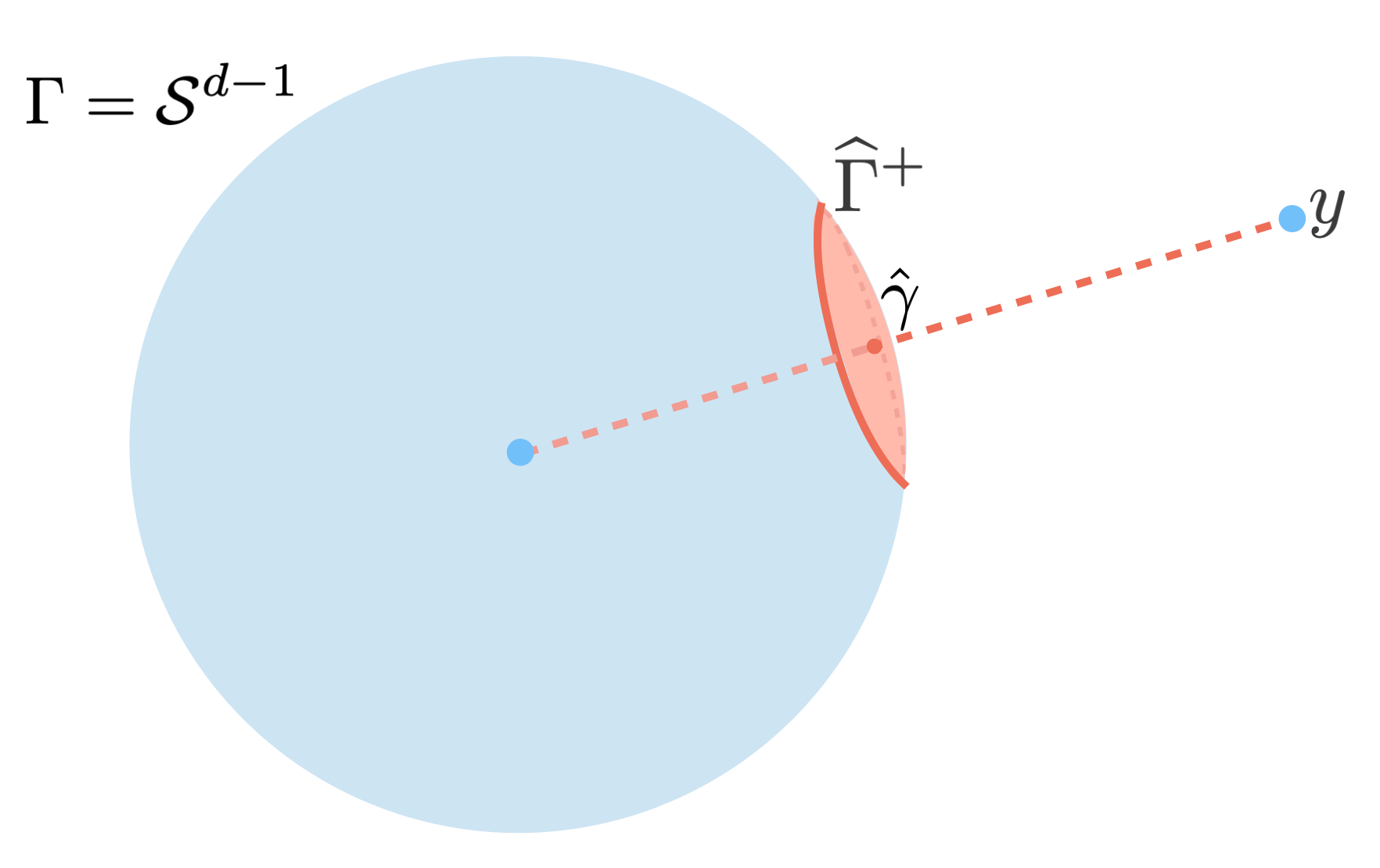}
    \caption{(Left) Angle $\delta_d(\|y\|_2)$ (in $^\circ$) for different values of $d$. Locally simultaneous inference takes a correction over $\widehat \Gamma^+ = \{\gamma\in\mathcal S^{d-1}: \angle (\gamma, y) \leq \delta_d(\|y\|_2)\}$. (Middle) Interval width after taking a correction over $\widehat \Gamma^+$, relative to a fully simultaneous correction. (Right) Illustrative comparison of $\widehat \Gamma^+$ and $\mathcal S^{d-1}$.}
    \label{fig:selective_chi}
\end{figure}
\end{example}

\paragraph{Conditional approaches.}
Conditional approaches bound the probability of error conditional on selecting a specific target:
$$P\{\theta_\gamma \not\in C_\gamma~|~\hat \gamma = \gamma\}.$$
While simultaneous methods ensure validity for arbitrary $\hat \gamma$ chosen from the set $\Gamma$, regardless of any further properties of the selection, conditional approaches adapt to the selection at hand. In particular, a crucial step in implementing a conditional correction is tractably characterizing the selection event $\{\hat\gamma = \gamma\}$. Prior work has provided such characterizations for a variety of model selection methods, such as the LASSO, forward stepwise, LARS, etc \cite{lee2016exact, tibshirani2016exact, fithian2014optimal}. This adaptivity of conditional approaches often allows them to outperform simultaneous approaches; for example, if a specific target is selected with overwhelming probability, then conditional methods yield confidence intervals that are nearly the same as uncorrected intervals. On the other hand, characterizing the selection event is difficult in general and conditional corrections are available only in certain restrictive problem settings, usually parametric exponential families. Furthermore, since the final guarantees are conditional rather than unconditional, conditional methods can yield large intervals; notably, Goeman and Solari~\cite{goeman2022conditional} showed that for every conditional method there exists a simultaneous inference method that dominates it in terms of power. Kivaranovic and Leeb~\cite{kivaranovic2021length} showed that conditionally valid intervals, even if tight, have infinite expected length for common selection problems.
To fix this issue of enlarged intervals due to conditioning, Andrews et al.~\cite{andrews2019inference} introduced a refinement of conditional inference for the problem of inference on the ``winner'' called the \emph{hybrid} method, which was subsequently generalized to other problems by McCloskey \cite{mccloskey2020hybrid}. We will use the term ``hybrid method'' to refer to the general procedure of McCloskey. The hybrid method begins by constructing simultaneous intervals for all candidate targets of inference. Then, it implements a correction conditional on both the selected target \emph{and} the event that the intervals constructed in the first step cover the target. This strategy offers unconditional guarantees only, but can lead to significant power gains over standard conditional inference. An alternative approach to fixing large conditional intervals is to randomize the selection method \cite{tian2018selective}, but this solution comes at the cost of deteriorating selection quality.

Locally simultaneous inference is adaptive to the selection strategy at hand like conditional approaches, but its technical underpinnings in the robust framework of simultaneous inference allow for broader applicability. First, compared with conditional inference, locally simultaneous inference is more easily applicable in nonparametric problem settings and in  ``continuous selection'' problems, where the probability of any specific selection event is equal to zero. Further, seeing that it aims to provide unconditional guarantees, locally simultaneous inference can be more powerful than conditional methods, especially the standard, nonhybridized approach. For example, as mentioned earlier, locally simultaneous inference is at least as powerful as simultaneous inference, which is not true of the conditional and hybrid methods. In general, however, either of locally simultaneous inference and conditional inference can be more powerful depending on the nature of the inference problem. We give further intuition as to when one approach outperforms the other in numerical simulations in Section~\ref{sec:exps}, where we compare locally simultaneous inference to the conditional and hybrid approaches.
Another important advantage of locally simultaneous inference over conditional inference is that it is easy to implement when inference is sought after a \emph{set} of targets $\widehat\Gamma \subseteq \Gamma$; e.g., one might be interested in constructing intervals for all treatments whose estimated effect exceeds a fixed threshold. It is unclear how to construct tight conditionally valid intervals that simultaneously cover a set of targets, since conditioning on the selected set could severely constrain the sample space. Finally, a convenient feature of locally simultaneous inference is that it is amenable to stable numerical implementation, as it does not require amplifying low-probability events like conditional methods.

\begin{example}
\label{ex:conditional}
To illustrate the brittleness of the conditional approach, we revisit the selection problem from Example \ref{ex:chi}: given $y\in\R^d$, we select $\hat \gamma = \frac{y}{\|y\|_2}$. We let $y\sim \mathcal N(\mu,\Sigma)$; the estimand of interest is $\theta_{\hat\gamma} = \hat\gamma^\top \mu$. Even if $\Sigma$ is known, conditional inference breaks down unless $\Sigma$ is a multiple of the identity matrix. As we show in Appendix \ref{app:example_conditional}, the conditional distribution of $y$ given $\{\hat \gamma = \gamma\}$ depends on $\mu$ only through $\gamma^\top \Sigma^{-1} \mu$. Unless $\gamma$ is an eigenvector of $\Sigma$, there will exist some $v\in\R^d$ such that $\gamma^\top \Sigma^{-1}v = 0$  but $\gamma^\top v \neq 0$, in which case $\mu$ and $\mu + tv$ give the same conditional distribution for any $t \in \mathbb{R}$, but $\gamma^\top \mu$ and $\gamma^\top (\mu + tv)$ can be arbitrarily different. In other words, $\gamma^\top \mu$ is almost surely unidentifiable in the conditional model---thus leading to infinite intervals---unless $\gamma$ happens to be an eigenvector of $\Sigma$, a measure-zero event. In contrast, simultaneous inference, as well as its local refinement, remain both powerful and tractable. Moreover, as in Example~\ref{ex:chi}, locally simultaneous intervals approach nominal intervals as $\|y\|_2\rightarrow\infty$.

\end{example}

\paragraph{Other related work.}
There are several works that study inference on the winning effect (and, more generally, the top $k$ out of $m$ observed effects) \cite{venter1988confidence, fuentes2018confidence, benjamini2019confidence}. These fall under neither of the two discussed frameworks; rather, the guarantees are ``simultaneous over the selected'' \cite{benjamini2019confidence}. In contrast to locally simultaneous inference, these analyses require independent observations and the confidence interval widths are not data-adaptive; that is, they do not depend on how ``obvious'' the selection is. Rather, the improvement over the usual simultaneous correction lies in a careful analysis of the selection event. Interestingly, these works show that an asymmetric interval construction is more appropriate when selecting the largest observations: the winning observation only has upward bias and no downward bias. Hence, only the lower endpoint of the confidence interval requires correction. Understanding how to combine such an asymmetric analysis with locally simultaneous inference is a valuable direction going forward.
It is worth noting that, for the illustrative problem in Figure \ref{fig:intro_comparison}, Benjamini et al.~\cite{benjamini2019confidence} show that no multiplicity correction is in fact necessary, but this is only true for the special case when $m=2,k=1$ and under additional regularity assumptions.

In Section \ref{sec:erm} we show that locally simultaneous inference directly implies a refinement of complexity-based generalization arguments from empirical process theory. The refinement relies on bounding the complexity of hypotheses with low empirical risk, which closely resembles the developments in the literature on \emph{localization} \cite{koltchinskii2000rademacher,bartlett2005local, koltchinskii2006local, koltchinskii2011oracle}. While conceptually similar, to the best of our knowledge a localization result in the vein of Theorem \ref{thm:risk_minimizer}, which is agnostic to the considered complexity notion, was not previously known. While this direction is not the main focus of the current work, we believe that locally simultaneous inference could have other useful implications to problems in learning theory.

Another conceptually similar existing idea is relying on the \emph{stability} of selection to correct for selective inferences \cite{zrnic2020post}. Indeed, locally simultaneous inference is most powerful when the selection is ``stable'', meaning it does not vary a lot as the data changes, since such a form of stability makes the set of plausible selections $\widehat\Gamma^+$ small. However, at a technical level our framework is entirely different, the main difference being that it does not rely on introducing randomization to the selection procedure.

Finally, the high-level idea of splitting the error budget into two parts, one for constructing an initial confidence region and the remainder for conducting simultaneous inference within the confines of the confidence region, was explored by Leeb and P{\"o}tscher~\cite{leeb2017testing}, Romano et al.~\cite{romano2014practical}, Zhao et al.~\cite{zhao2019multiple}, and others in the context of multiple testing.

\section{Locally simultaneous inference: general theory}

In this section, we state our main results in full generality. Then, in the following sections, we instantiate the results for a variety of selection problems.

We begin by introducing notation and technical preliminaries.

\subsection{Preliminaries}

We consider a possibly nonparametric family of distributions $\P$. For every distribution $P\in\P$, we have a family of possible target estimands indexed by $\gamma\in\Gamma$, $\{\theta_\gamma(P)\}_{\gamma\in \Gamma}$. For example, $\P = \{\mathcal N(\mu, I_m):\mu\in\R^m\}$ could be a location family, $\Gamma = \{1,\dots,m\}$ the set of possible target indices, and $\theta_\gamma(\mathcal{N}(\mu, I_m)) = \mu_\gamma$ asks for the coordinate of $\mu$ indexed by $\gamma$. The relevant distribution $P$ will usually be clear from the context, in which case we will simplify notation and write $\theta_\gamma\equiv \theta_\gamma(P)$.


Selective inference studies the problem of doing inference on $\{\theta_\gamma:\gamma\in \widehat{\Gamma}(y)\}$ given data $y\sim P$, where $\widehat{\Gamma}(y)$ determines a \emph{data-dependent} set of inferential targets. We will adopt the convention that $\widehat\Gamma \equiv \widehat\Gamma(y)$ when the argument $y$ is clear from the context. When there is a single selected target, we will denote it by $\hat \gamma \equiv \hat \gamma(y)$; in that case $\widehat\Gamma = \{\hat \gamma\}$.

Simultaneous inference controls type I error simultaneously for \emph{all} possible targets $\gamma\in\Gamma$: in particular, given $y\sim P$, it delivers confidence regions $\{\ci_{\gamma\cdot \Gamma}^{\alpha}\}_{\gamma\in\Gamma}$ such that
\[P\left\{\theta_\gamma \in \ci_{\gamma\cdot \Gamma}^\alpha, \ \forall \gamma\in \Gamma\right\} \geq 1-\alpha,\]
for a pre-specified error level $\alpha\in(0,1)$.
The subscript $\gamma\cdot\Gamma$ indicates that the confidence region is provided for target $\gamma$ while adjusting for the requirement of simultaneous coverage over $\Gamma$. The superscript $\alpha$ indicates the tolerated error probability. 
Sometimes the error probability will be irrelevant or clear from the context, in which case we will drop it and write $C_{\gamma\cdot\Gamma} \equiv C_{\gamma\cdot\Gamma}^\alpha$.

We give examples of two simultaneous inference methods that will be our default choices for nonparametric problems and parametric problems, respectively. We focus on these methods since they are among the simplest and most common simultaneous inference methods (though certainly more sophisticated choices exist).

\begin{example}[Bonferroni correction]
The Bonferroni correction achieves simultaneous control by using nominal (i.e., unadjusted) intervals at the corrected error level $\alpha / |\Gamma|$:
\[
C_{\gamma\cdot\Gamma}^{\text{Bonf}(\alpha)} = C_\gamma^{\text{nom}(\alpha/|\Gamma|)}.
\]
Here, $C_\gamma^{\text{nom}(\alpha)}$ are any intervals that satisfy
$$P\left\{\theta_\gamma \in C_\gamma^{\text{nom}(\alpha)}\right\} \geq 1-\alpha,$$
for a specified level $\alpha\in(0,1)$. Bonferroni-corrected intervals can be applied nonparametrically and are valid regardless of any dependencies between the different estimation problems included in $\Gamma$.
\end{example}

\begin{example}[Maximal z- or t-statistic]
\label{ex:maximal_z_or_t}
If we have prior knowledge about the dependence structure of the different estimation problems included in $\Gamma$, there are approaches that outperform the Bonferroni correction. Suppose that for each $\gamma\in\Gamma$ we observe $\hat \theta_\gamma \sim \mathcal{N}(\theta_\gamma, \sigma_\gamma^2)$ and that jointly these observations make a multivariate Gaussian vector with a known covariance matrix. Denote the known covariance matrix of $(\hat \theta_\gamma)_{\gamma\in\Gamma}$ by $\Sigma$. Then, standard simultaneous confidence intervals are obtained by simulating the $1-\alpha$ quantile of the maximal z-statistic given by:
$$\max_{\gamma\in\Gamma} \frac{|Z_\gamma|}{\sigma_\gamma},$$
where $(Z_\gamma)_{\gamma\in\Gamma}\sim \mathcal{N}(0,\Sigma)$. Denote this quantile by $q_\Gamma^\alpha$. We construct the confidence intervals as
$$C_{\gamma\cdot \Gamma}^\alpha = \left(\hat\theta_\gamma \pm q_\Gamma^\alpha \sigma_\gamma\right).$$
The validity of the intervals follows immediately from the definition of $q_\Gamma^\alpha$. When the covariance matrix of the estimates is not known exactly but can be estimated, one can similarly construct intervals by computing the $1-\alpha$ quantile of the maximal t-statistic.
\end{example}

With simultaneously valid confidence regions in hand, as long as the selected targets $\widehat\Gamma$ are guaranteed to fall within the fixed admissible set $\Gamma$, meaning $\widehat{\Gamma}\subseteq \Gamma$ almost surely, an immediate implication is that
\[P\left\{\theta_\gamma \in \ci_{\gamma\cdot \Gamma}, \ \forall \gamma\in \widehat{\Gamma}\right\} \geq P\left\{\theta_\gamma \in \ci_{\gamma\cdot \Gamma}, \ \forall \gamma\in \Gamma\right\} \geq 1-\alpha,\]
which resolves the coupling between the data and the selected targets $\widehat \Gamma$.

The basic principle of our correction is to find a \emph{data-dependent}  set of targets nested between $\widehat\Gamma$ and $\Gamma$, such that taking a simultaneous correction over the set ensures type I error control. To implement this idea, we assume that we can construct simultaneous confidence regions $\ci_{\gamma\cdot\Gamma'}$ for any desired subset $\Gamma'\subseteq \Gamma$, at any target error level $\alpha$. Formally, we have access to a family of confidence regions $\{\ci_{\gamma\cdot\Gamma'}\}_{\Gamma'\subseteq \Gamma}$ such that
\[P\left\{\theta_\gamma \in \ci_{\gamma\cdot \Gamma'}, \ \forall \gamma\in \Gamma'\right\} \geq 1-\alpha,\]
for all $\Gamma'\subseteq \Gamma$.


To ensure validity of our construction, we make a mild and natural monotonicity assumption, requiring that the confidence regions can only increase as the set of inferential targets increases.

\begin{assumption}
\label{ass:monotone}
We say that the confidence regions $\{\ci_{\gamma\cdot\Gamma'}\}_{\Gamma'\subseteq \Gamma}$ are \emph{monotone} if for all $\Gamma_1\subseteq \Gamma_2\subseteq \Gamma$ and $\gamma\in\Gamma_1$,
\[\ci_{\gamma\cdot\Gamma_1} \subseteq \ci_{\gamma\cdot\Gamma_2}.\]
\end{assumption}

We are now ready to outline the general solution based on locally simultaneous inference.

\subsection{General construction}

The core idea of our method is to narrow down the set of inference targets to a subset $\widehat\Gamma^+$ that includes only those targets that have a reasonable chance of being selected. Specifically, for some initial tolerance $\nu \leq \alpha$, we first construct a $1-\nu$ level prediction set $\Gamma_\nu$ for $\widehat\Gamma$, followed by a $1-\nu$ level confidence set $\widehat\Gamma_\nu^+$ for $\Gamma_{\nu}$. By construction, we will ensure that we have $\widehat\Gamma \subseteq \Gamma_\nu \subseteq \widehat\Gamma_\nu^+$ with probability at least $1-\nu$.

The key observation to justify our method is that, whenever $\widehat\Gamma \subseteq \Gamma_\nu \subseteq \widehat\Gamma_\nu^+$, the \emph{random} set of intervals $\{\ci_{\gamma\cdot\widehat\Gamma_\nu^+}:\; \gamma \in \widehat\Gamma\}$ enjoys greater simultaneous coverage than the \emph{fixed} set of intervals $\{\ci_{\gamma\cdot \Gamma_\nu}:\; \gamma \in \Gamma_\nu\}$. Informally, in the random set, there are fewer chances to make a mistake because $\widehat\Gamma \subseteq \Gamma_\nu$, and the intervals are also wider because $\Gamma_\nu\subseteq \widehat\Gamma_\nu^+$. This observation will allow us to bound the simultaneous coverage of the random set in terms of the simultaneous coverage of the fixed set, at least on the event where $\widehat\Gamma \subseteq \Gamma_\nu \subseteq \widehat\Gamma_\nu^+$. The remainder of this section develops these ideas in detail.

For every $P\in\P$, suppose that we can construct a set $A_\nu(P)$ that satisfies
\[P\{y \in A_\nu(P)\} \geq 1-\nu,\]
for any pre-specified $\nu\in(0,1)$. In other words, $A_\nu(P)$ is the acceptance region of a valid test for the null hypothesis $H_P: y\sim P$ at level $1-\nu$.
Intuitively, $A_\nu(P)$ can be thought of as the set of all plausible observations according to distribution $P$. When $\P = \{P_\mu\}_{\mu\in\mathcal{M}}$ is a parametric family, we will simply write $A_\nu(P_\mu) \equiv A_\nu(\mu)$.

We define the set of \emph{plausible targets} under distribution $P$ to be:
\begin{align*}
\Gamma_\nu(P) &:= \underset{y'\in A_\nu(P)}{\cup} \widehat{\Gamma}(y').
\end{align*}
Note that, unlike the realized selection $\widehat\Gamma(y)$, $\Gamma_\nu(P)$ is a \emph{fixed} set of targets. Again, when $\P = \{P_\mu\}_{\mu\in\mathcal{M}}$ is a parametric family, we will write $\Gamma_\nu(P_\mu) \equiv \Gamma_\nu(\mu)$.

Finally, we define the inversion of $A_\nu(P)$, which gives a confidence region for the true distribution $P$:
$$B_\nu(y) = \{P\in\P: y\in A_\nu(P)\}.$$

Before proving our main result, which asserts validity of locally simultaneous inference, we prove a key technical lemma that makes the core of the argument.

\begin{lemma}
\label{lemma:local_posi_tighter}
Fix $\alpha\in(0,1)$ and $\nu\in(0,\alpha)$. Let $\{\tilde \ci_{\gamma\cdot\Gamma'}\}_{\Gamma'\subseteq \Gamma}$ be a family of confidence regions such that
\begin{equation}
\label{eq:C-til}
P\left\{\theta_\gamma \in \tilde \ci_{\gamma\cdot\Gamma'},\forall\gamma\in\Gamma', y \in A_\nu(P) \right\}\geq 1-\alpha,
\end{equation}
for all $\Gamma'\subseteq\Gamma$.
Moreover, suppose that the regions are monotone (Ass.~\ref{ass:monotone}). Consider the set of targets
\begin{equation*}
\widehat{\Gamma}_\nu^+ = \underset{P'\in B_\nu(y)}{\cup} \Gamma_\nu(P') = \underset{P'\in B_\nu(y)}{\cup} \underset{y'\in A_\nu(P')}{\cup} \widehat \Gamma(y').
\end{equation*}
Then, it holds that
\[P\left\{\theta_\gamma \in \tilde \ci_{\gamma\cdot \widehat\Gamma_{\nu}^+}, \forall \gamma \in \widehat\Gamma\right\} \geq 1-\alpha.\]
\end{lemma}

\begin{proof}
First, we can write
\[
P\left\{\theta_\gamma \in \tilde \ci_{\gamma\cdot \widehat\Gamma_{\nu}^+}, \forall \gamma \in \widehat\Gamma\right\} 
\geq
P\left\{\theta_\gamma \in \tilde \ci_{\gamma\cdot \widehat\Gamma_{\nu}^+}, \forall \gamma \in \widehat\Gamma, y \in A_\nu\right\}.
\]
Now notice that, on the event $\{y\in A_\nu(P)\} = \{P \in B_\nu(y)\}$, it almost surely holds that $\widehat \Gamma \subseteq \Gamma_\nu(P) \subseteq \widehat \Gamma_\nu^+$. Using this fact, we have
\begin{align*}
P\left\{\theta_\gamma \in \tilde \ci_{\gamma\cdot \widehat\Gamma_{\nu}^+}, \forall \gamma \in \widehat\Gamma, y \in A_\nu\right\} &\geq P\left\{\theta_\gamma \in \tilde \ci_{\gamma\cdot \widehat\Gamma_{\nu}^+}, \forall \gamma \in \Gamma_\nu(P), y \in A_\nu\right\}\\
&\geq P\left\{\theta_\gamma \in \tilde \ci_{\gamma\cdot \Gamma_\nu(P)}, \forall \gamma \in \Gamma_\nu(P), y \in A_\nu\right\},
\end{align*}
where the second inequality follows by the monotonicity of the confidence regions. Since the right-hand side is at least $1-\alpha$ by the definition of $\tilde \ci_{\gamma\cdot \Gamma_\nu(P)}$, we have shown
$$P\left\{\theta_\gamma \in \tilde \ci_{\gamma\cdot \widehat\Gamma_{\nu}^+}, \forall \gamma \in \widehat\Gamma\right\} \geq 1- \alpha,$$
as desired.
\end{proof}

Therefore, Lemma \ref{lemma:local_posi_tighter} reduces the problem of constructing selective confidence regions to the problem of constructing the regions $\tilde C_{\gamma\cdot\Gamma'}$ satisfying Eq.~\eqref{eq:C-til} for every \emph{fixed} $\Gamma'$. The following theorem, providing such regions, states our main result.

\begin{theorem}
\label{thm:local_posi_main}
Fix $\alpha\in(0,1)$ and $\nu\in(0,\alpha)$. Suppose that the simultaneous confidence regions $\{\ci_{\gamma\cdot\Gamma'}\}_{\Gamma'\subseteq \Gamma}$ are monotone (Ass.~\ref{ass:monotone}). Consider the set of targets
\begin{equation}
\label{eq:local_targets}
\widehat{\Gamma}_\nu^+ = \underset{P'\in B_\nu(y)}{\cup} \Gamma_\nu(P') = \underset{P'\in B_\nu(y)}{\cup} \underset{y'\in A_\nu(P')}{\cup} \widehat \Gamma(y').
\end{equation}
Then, it holds that
\[P\left\{\theta_\gamma \in \ci_{\gamma\cdot \widehat\Gamma_{\nu}^+}^{(\alpha-\nu)}, \forall \gamma \in \widehat\Gamma\right\} \geq 1-\alpha.\]
\end{theorem}

\begin{proof}
The proof follows by an application of Lemma \ref{lemma:local_posi_tighter}. In particular, by a union bound it follows that $C^{(\alpha-\nu)}_{\gamma\cdot \Gamma'}$ is a valid choice of $\tilde C_{\gamma\cdot \Gamma'}$ in Lemma \ref{lemma:local_posi_tighter}:
\begin{align*}
P\left\{\theta_\gamma \in C^{(\alpha-\nu)}_{\gamma\cdot \Gamma'}, \forall \gamma \in \Gamma', y\in A_\nu(P)\right\} &\geq 1 - P\{\exists \gamma\in\Gamma':\theta_\gamma \not \in C^{(\alpha-\nu)}_{\gamma\cdot \Gamma'}\} - P\{y \not\in A_\nu(P)\}\\
&\geq 1-(\alpha-\nu) - \nu = 1-\alpha.
\end{align*}
\end{proof}

Intuitively, Theorem \ref{thm:local_posi_main} justifies the following refinement of simultaneous inference. Given data $y$, first construct a set of all distributions under which the observed data is plausible. Then, consider all plausible observations under those distributions; this essentially gives a collection of datasets $y'$ in a neighborhood around $y$. Finally, perform simultaneous inference over all inferential targets $\widehat\Gamma(y')$ that could be selected in this neighborhood. Despite the fact that the set of targets is constructed as a function of $y$, a simultaneous correction over this set nevertheless ensures valid selective inferences for $\widehat \Gamma(y)$.

Next, we provide a slightly different correction from that of Theorem \ref{thm:local_posi_main} that \emph{strictly dominates} simultaneous inference at error level $\alpha$, for any choice of $\nu$ (note that the correction in  Theorem~\ref{thm:local_posi_main} dominates simultaneous inference at level $\alpha-\nu$). This is achieved by carefully choosing $A_\nu(P)$. The refined correction is often easy to apply, however we find that the strategy from Theorem~\ref{thm:local_posi_main} is usually more practical as it allows choosing $A_\nu(P)$ freely. For the next result, we assume \emph{centered} confidence intervals.

\begin{assumption}
\label{ass:centered_ints}
We say that $\{C_{\gamma\cdot\Gamma'}^\alpha\}_{\gamma'\subseteq\Gamma}$ are \emph{centered} confidence intervals if
$$C_{\gamma \cdot \Gamma'}^\alpha = \left(\hat\theta_\gamma \pm q_{\Gamma'}^\alpha \cdot \hat \sigma_\gamma\right),$$
for some estimator $\hat\theta_\gamma$ and standard error $\hat \sigma_\gamma$, where $q_{\Gamma'}^\alpha$ is chosen such that $P\{\theta_\gamma \in C^\alpha_{\gamma\cdot \Gamma'}, \forall \gamma\in\Gamma'\} \geq 1- \alpha$.
\end{assumption}
Confidence intervals are often centered; for example, this is true of intervals based on the maximal z- or t-statistic, as in Example \ref{ex:maximal_z_or_t}. We denote $C_{\gamma}(q) := \left(\hat\theta_\gamma \pm q \hat \sigma_\gamma\right)$; then, $C_{\gamma\cdot \Gamma'} = C_\gamma(q_{\Gamma'}^\alpha)$ are intervals valid simultaneously over $\Gamma'$ at level $1-\alpha$. Without loss of generality we assume that $q_{\Gamma'}^\alpha$ is nonincreasing in~$\alpha$.

\begin{theorem}
\label{thm:refined}
Fix $\alpha\in(0,1)$ and $\nu\in(0,\alpha)$. Suppose that the confidence intervals are monotone (Ass.~\ref{ass:monotone}), i.e., $q_{\Gamma_1}^\alpha\leq q_{\Gamma_2}^\alpha$ for all $\Gamma_1\subseteq \Gamma_2$, and centered (Ass.~\ref{ass:centered_ints}). Let $A_\nu(P) = \{y : \theta_\gamma \in C_{\gamma}(q_\Gamma^\nu), \forall\gamma\in\Gamma\}$, and let $\widehat\Gamma_\nu^+$ denote the set of targets from Theorem \ref{thm:local_posi_main} (Eq.~\eqref{eq:local_targets}). 
Let
$$\hat q = \min\left\{q_{\widehat\Gamma_\nu^+}^{(\alpha-\nu)},q_\Gamma^\alpha\right\}.$$
Then, it holds that
$$P\left\{ \theta_\gamma \in C_{\gamma}(\hat q), \forall \gamma\in\widehat\Gamma \right\}\geq 1-\alpha.$$
\end{theorem}

\begin{proof}
Analogously to Theorem \ref{thm:local_posi_main}, we show that $C_{\gamma}\left(\min\{q^{(\alpha-\nu)}_{\Gamma'}, q^\alpha_{\Gamma}\}\right)$ is a valid choice of $\tilde C_{\gamma\cdot\Gamma'}$ in Lemma~\ref{lemma:local_posi_tighter}. Invoking Lemma \ref{lemma:local_posi_tighter} then completes the proof.

We split the analysis into two cases, depending on which term achieves the minimum.

First, suppose that $\Gamma'$ is such that $q^{(\alpha-\nu)}_{\Gamma'} \leq q^\alpha_{\Gamma}$. Then, by a union bound, we have
\begin{align*}
P\{\theta_\gamma \in \tilde C_{\gamma\cdot \Gamma'}, \forall \gamma\in \Gamma', y\in A_\nu(P)\} &\geq 1 - P\{y\not \in A_\nu(P)\} - P\{\exists \gamma\in\Gamma': \theta_\gamma \not\in \tilde C_{\gamma\cdot \Gamma'}\}\\
&= 1 - P\{y\not \in A_\nu(P)\} - P\{\exists \gamma\in\Gamma': \theta_\gamma \not\in  C_{\gamma}(q_{\Gamma'}^{(\alpha-\nu)})\}\\
&\geq 1-\nu - (\alpha-\nu)\\
&= 1-\alpha.
\end{align*}
Next, suppose that $\Gamma'$ is such that $q^\alpha_{\Gamma}\leq q^{(\alpha-\nu)}_{\Gamma'}$. Then,
\begin{align*}
P\{\theta_\gamma \in \tilde C_{\gamma\cdot \Gamma'}, \forall \gamma\in \Gamma', y\in A_\nu(P)\} &= P\{\theta_\gamma \in C_{\gamma}(q^\alpha_{\Gamma}), \forall \gamma\in \Gamma', y\in A_\nu(P)\}\\
&\geq  P\{\theta_\gamma \in C_{\gamma}(q^\alpha_{\Gamma}), \forall \gamma\in \Gamma, y\in A_\nu(P)\}\\
&= P\{\theta_\gamma \in C_{\gamma}(q^\alpha_{\Gamma}), \forall \gamma\in \Gamma\}\\
&\geq 1-\alpha,
\end{align*}
where the third step follows because, by definition, $\{\theta_\gamma \in C_{\gamma}(q^\alpha_{\Gamma}), \forall \gamma\in \Gamma\}\Rightarrow \{y\in A_\nu(P)\}$. 

Therefore, $\tilde C_{\gamma\cdot\Gamma'} = C_{\gamma}\left(\min\{q^{(\alpha-\nu)}_{\Gamma'}, q^\alpha_{\Gamma}\}\right)$ is a valid choice of $\tilde C_{\gamma\cdot\Gamma'}$ in Lemma \ref{lemma:local_posi_tighter}, as desired.
\end{proof}

Under these mild conditions, locally simultaneous inference therefore comes at \emph{no cost} in terms of power: the intervals are at least as tight as fully simultaneous intervals. Moreover, whenever $\widehat \Gamma_\nu^+$ is a strict subset of all admissible selections, they will be strictly tighter.

\section{Inference on the ``most promising'' effects}
\label{sec:winner}

We first study the problem of constructing confidence intervals for the ``most promising'' effects. We consider two instantiations of the problem: inference on the winner and the file-drawer problem.

Given data $y = (y_1,\dots,y_m)\in\R^m$, the problem of inference on the winner asks for a confidence interval for the mean of the largest entry of $y$. Formally, if we let $\theta_\gamma = \E y_\gamma$ for all $\gamma\in[m]$, the goal is to do inference on $\theta_{\hat \gamma}$, where
\begin{equation}
\label{eq:winner}
 \hat \gamma = \argmax_{\gamma\in[m]} y_\gamma.
\end{equation}
The file-drawer problem asks for a confidence region that simultaneously covers the means of all observations that exceed a critical threshold $T$. Formally, the region is required to cover $\{\theta_\gamma: \gamma\in \widehat\Gamma\}$, where
\begin{equation}
\label{eq:filedrawer}
\widehat\Gamma = \{\gamma\in[m]: y_\gamma \geq T\}.
\end{equation}

The $m$ coordinates of $y$ can correspond to, for example, the effectiveness of $m$ different treatments, in which case the selection corresponds to focusing on the single seemingly best treatment or multiple treatments that are deemed sufficiently promising. The $m$ outcomes can also correspond to measurements of a time series over $m$ time steps (e.g., blood pressure in a specified interval), in which case selection focuses on the time steps at which the series achieves extreme values. Finally, $y$ can capture an estimate of the effectiveness of a treatment on $m$ different subgroups (e.g., $m$ age groups); the selection would then ask for the effectiveness within the single subgroup or several subgroups for which the treatment seems most promising.

We consider a parametric version and a nonparametric version of the two problems. Importantly, conditional selective inference is not directly applicable in the latter setting.

\subsection{Parametric case}

We begin with the case where $\P$ is a parametric family; in particular, we take $\P =\{P_\mu\}_\mu$ to be a location family with location parameter $\mu\in \R^m$. In other words, $y= (y_1,\dots,y_m)\sim P_\mu$ can be written as $y = \mu + Z$, where $Z= (Z_1,\dots,Z_m)\sim P_0$. For simplicity of exposition we assume that the errors $Z_i$ have the same marginal symmetric zero-mean distribution $P_0^{(1)}$ (e.g., $P_0^{(1)} = \mathcal{N}(0,\sigma^2)$), however generalizing beyond this setting is straightforward. We do not assume that the errors $Z_i$ are necessarily independent, i.e. that $P_0$ is a product distribution.

For an index set $\mathcal I\subseteq [m]$, we define
$$q^{\alpha}(\mathcal I) = \inf \left\{ q : P_0\left\{\max_{i\in\mathcal I} |Z_i| \leq q\right\} \geq 1-\alpha \right\}.$$
In words, $q^\alpha(\mathcal I)$ is the $1-\alpha$ quantile of the maximum absolute error over indices in $\mathcal I$. This would be the usual interval half-width if a simultaneous correction is required over the observations in $\mathcal I$. 
This value can be loosely upper bounded by taking a Bonferroni correction over $\mathcal I$. We note that exact knowledge of $P_0$ is not necessary; being able to compute an upper bound on $q^{\alpha}(\mathcal I)$ suffices.

Note that $\theta_i = \E y_i \equiv \mu_i$ in this setting; that is, the possible estimands $\theta_i$ are coordinates of the location parameter. Therefore, for $\hat \gamma$ as in Eq.~\eqref{eq:winner}, we want to construct a confidence interval for $\mu_{\hat \gamma}$; for $\widehat \Gamma$ as in Eq.~\eqref{eq:filedrawer}, we want to construct a confidence region for $\{\mu_{\gamma}:\gamma\in \widehat \Gamma\}$.

We now apply our general result about locally simultaneous inference.

\begin{theorem}
\label{thm:gauss_seq_winner}
Fix $\alpha\in(0,1)$ and $\nu\in(0,\alpha)$.
\begin{itemize}
    \item
For the problem of inference on the winner (Eq.~\eqref{eq:winner}), let the set of plausible indices be
$$\widehat \Gamma_\nu^+ = \left\{\gamma\in[m]: y_\gamma \geq y_{\hat \gamma} - 4q^{\nu}([m])\right\}.$$
Then,
$$P_\mu\left\{\mu_{\hat \gamma} \in \left(y_{\hat \gamma} \pm \min\left\{q^{(\alpha-\nu)}(\widehat \Gamma_\nu^+), q^{\alpha}([m])\right\} \right)\right\} \geq 1-\alpha.$$
\item For the file-drawer problem (Eq.~\eqref{eq:filedrawer}), let the set of plausible indices be
$$\widehat \Gamma_\nu^+ = \left\{\gamma\in[m]: y_\gamma \geq T - 2q^{\nu}([m])\right\}.$$
Then,
$$P_\mu\left\{\mu_{\gamma} \in \left(y_{\gamma} \pm \min\left\{q^{(\alpha-\nu)}(\widehat \Gamma_\nu^+), q^{\alpha}( [m])\right\}\right),~\forall \gamma\in\widehat\Gamma \right\} \geq 1-\alpha.$$
\end{itemize}
\end{theorem}

Theorem \ref{thm:gauss_seq_winner} formalizes the intuition that one should only have to add the ``nearly selected'' observations to the simultaneous correction, if the goal is to construct a valid confidence region around the selected ones. When there are many observations that are far from promising, then $\widehat \Gamma_\nu^+$ can be much smaller than~$[m]$.

We note the file-drawer problem asks for a confidence region around \emph{all} parameters that exceed the selection threshold; in contrast, the conditional approach of \citet{lee2016exact} provides inference for one real-valued parameter at a time. It is unclear how to generalize it to the problem of inference on multiple parameters without resorting to a trivial solution such as a Bonferroni correction over all estimands, which ignores the dependencies between the different estimation problems. The same observation applies to the hybrid method~\cite{andrews2019inference, mccloskey2020hybrid}. In contrast, since locally simultaneous inference builds on standard simultaneous inference, it is able to adapt to the dependencies at hand.

\subsection{Nonparametric case}

We show that essentially the same reasoning as in the parametric case applies to nonparametric settings.

For each of the $m$ candidates, we assume that we have $n$ i.i.d. observations that are bounded in $[0,1]$. More formally, we observe $n$ i.i.d. samples $y^{(1)},\dots,y^{(n)}$ drawn from a distribution $P$ with $\text{supp}(P)\subseteq [0,1]^m$. As before, we denote the $m$-dimensional vector of means by $\theta = \E y^{(1)}$.

In the problem of inference on the winner, we would like to do inference on $\theta_{\hat \gamma}$, where
\begin{equation}
\label{eq:winner_np}
\hat \gamma = \argmax_{\gamma\in[m]} y_\gamma := \argmax_{\gamma\in[m]}  \frac{1}{n}\sum_{j=1}^n y_\gamma^{(j)}.
\end{equation}
In the file-drawer problem, we would like to do inference on $\{\theta_{\gamma}:\gamma\in\widehat\Gamma\}$, where
\begin{equation}
\label{eq:filedrawer_np}
\widehat\Gamma = \{\gamma\in[m]: y_\gamma \geq T\} := \left\{\gamma\in[m] :  \frac{1}{n}\sum_{j=1}^n y_\gamma^{(j)} \geq T\right\}.
\end{equation}

Let $w_n^\alpha$ be any valid bound on the deviation of the empirical average of $n$ i.i.d. random variables $X_1,\dots,X_n\in [0,1]$ from their mean. Formally, $w_n^\alpha$ satisfies
$$P\left\{ \E X_1 \in \left(\frac{1}{n} \sum_{i=1}^n X_i \pm w_n^\alpha \right)\right\}\geq 1- \alpha.$$
For example, a standard choice of $w_n^\alpha$ is obtained from Hoeffding's inequality:
$$w_n^\alpha = \sqrt{\frac{\log(2/\alpha)}{2n}}.$$
Tighter choices of $w_n^\alpha$ are generally possible, e.g. by applying Bentkus' \cite{bentkus2004hoeffding}, Bernstein's \cite{bernstein1924modification}, or Bennett's inequality \cite{bennett1962probability}. 
Furthermore, for every $\gamma\in[m]$, we let $C_{\gamma}^\alpha$ be a confidence region for $\theta_\gamma$ valid at level $1-\alpha$. In our nonparametric experiments, we will take $C_{\gamma}^\alpha$ to be the betting-based confidence intervals due to \citet{waudby2020estimating}.

Similarly as in the parametric setting, we ensure valid selective inference by only requiring simultaneous control---here achieved by taking a Bonferroni correction---over the selected and nearly selected observations.

\begin{theorem}
\label{thm:nonparametric_winner}
Fix $\alpha\in(0,1)$ and $\nu\in(0,\alpha)$. Assume that $C_{\gamma}^{\alpha_1} \supseteq C_{\gamma}^{\alpha_2}$ for all $\alpha_1,\alpha_2\in(0,1)$ such that $\alpha_1 \leq \alpha_2$.
\begin{itemize}
    \item For the problem of inference on the winner (Eq.~\eqref{eq:winner_np}), let the set of plausible indices be
$$\widehat \Gamma_\nu^+ = \left\{\gamma\in[m]: y_\gamma \geq  y_{\hat \gamma} - 4w_n^{\nu/m} \right\}.$$
Then,
$$P\left\{\theta_{\hat \gamma} \in C_{\hat\gamma}^{(\alpha-\nu)/|\widehat\Gamma_\nu^+|} \right\} \geq 1-\alpha.$$
\item For the file-drawer problem (Eq.~\eqref{eq:filedrawer_np}), let the set of plausible indices be
$$\widehat \Gamma_\nu^+ = \left\{\gamma\in[m]: y_\gamma \geq T - 2w_n^{\nu/m}\right\}.$$
Then,
$$P\left\{\theta_{\gamma} \in C_{\gamma}^{(\alpha-\nu)/|\widehat \Gamma_\nu^+|},~\forall \gamma\in\widehat\Gamma \right\} \geq 1-\alpha.$$
\end{itemize}
\end{theorem}

\section{Inference after model selection}

We next consider the problem of inference after data-driven model selection. We first state a general implication of locally simultaneous inference in this context and then specialize this result to selection via the LASSO.

To set up the problem, suppose that we have a fixed design matrix $X\in\R^{n\times d}$ and a corresponding vector of outcomes $y\in\R^n$, where $y\sim P_\mu$. 
We assume $P_\mu$ is a location family, that is, $y\sim P_\mu \Leftrightarrow y \stackrel{d}{=} \mu + Z$, where $Z\sim P_0$ has mean zero.

We want to select a \emph{model} $\hat M \equiv \hat M(y)$ corresponding to a subset of the $d$ features, and then regress the outcome onto the selected features. Following the proposal of \citet{berk2013valid}, for a fixed model $M\subseteq[d]$ we take the so-called \emph{projection parameter} to be the target of inference. This parameter is obtained by approximating the outcome using the columns in $X$ indexed by $M$:
$$\theta_{M} := \argmin_\theta \E \left\|y - X_{M} \theta\right\|^2_2 = X_M^+ \mu,$$
where $X_M^+$ denotes the pseudoinverse of $X_M$. We denote the empirical counterpart of $\theta_M$ by $\hat \theta_M := X_M^+ y$. For a data-driven choice of model $\hat M$, the inferential target is therefore $\theta_{\hat M}$. We use $\theta_{j\cdot M}$ to denote the entry of $\theta_M$ corresponding to feature $j$. Note that, in general, $\theta_{j\cdot M}\neq \theta_{j\cdot M'}$ for two different models $M,M'$.

The set of possible estimands in this context is all possible values of $\theta_{j\cdot M}$. The natural index set $\Gamma$ for these estimands is given by all feature--model pairs: $\Gamma = \{(j,M):j\in M, M\in \M\}$, where $\M$ corresponds to all admissible feature selections, which is often $2^{[d]}$. The selected targets are the regression coefficients in the selected model, i.e. $\widehat \Gamma = \{(j,\hat M): j\in \hat M\}$.

To apply a locally simultaneous correction, we need to compute the augmented set of targets $\widehat \Gamma_\nu^+$. The key step in doing so is to find \emph{all plausible models}, which we will denote by $\widehat \M_\nu^+$; after we have $\widehat \M_\nu^+$, we apply a simultaneous correction in the vein of \citet{berk2013valid}, called the \emph{PoSI} correction, over $\widehat \M_\nu^+$.
Intuitively, $\widehat \M_\nu^+$ is the set of all models that could be selected on outcome vectors similar to $y$. Again, we note that this set is data-dependent; Berk et al., on the other hand, consider a deterministic set of possible models.


We will focus on methods that use $X^\top y$ as a sufficient statistic, which includes most common selection methods such as the LASSO, forward stepwise, etc. For such methods, a natural choice for the plausible set $A_\nu(\mu)$ is outcome vectors for which $X^\top y \approx X^\top \mu$. We formalize this in Corollary \ref{cor:linear_reg} below.

To state the result, for a set of contrasts $\V$, we define
$$q^{\alpha}(\V) = \inf\left\{q: P_0\left\{\sup_{v\in\V} |v^\top Z|\leq q\right\} \geq 1-\alpha \right\},$$
where $Z\sim P_0$. Let also $e_{j\cdot M}\in\R^{|M|}$ be the canonical vector with entry $1$ corresponding to feature $j\in M$ and $\hat\sigma_{j\cdot M} = \sqrt{e_{j\cdot M}^\top (X_M^\top X_M)^{-1} e_{j\cdot M}}$.

\begin{corollary}
\label{cor:linear_reg}
Fix $\alpha\in(0,1)$ and $\nu\in(0,\alpha)$. Let
$$\widehat\M_\nu^+ = \left\{\hat M(y'): \left\|X^\top y - X^\top y' \right\|_\infty \leq 2q^{\nu}\left(\{X_j\}_{j=1}^d\right) \right\}$$
and
$$\widehat \V_\nu^+ = \left\{\frac{e_{j\cdot M}^\top X_M^+}{\hat \sigma_{j\cdot M}}:M\in \widehat \M_\nu^+, j\in M\right\}.$$
Then,
$$P_\mu\left\{\theta_{j\cdot \hat M} \in \left(\hat\theta_{j\cdot \hat M} \pm q^{(\alpha - \nu)} \left(\widehat \V_\nu^+\right)\hat \sigma_{j\cdot M}  \right),~\forall j \in \hat M \right\} \geq 1 - \alpha.$$
\end{corollary}

Therefore, unlike the PoSI method \cite{berk2013valid}, we take a simultaneous correction only over models that seem plausible in hindsight. The correction of Corollary \ref{cor:linear_reg} is at least as tight as the PoSI correction at error level $\alpha-\nu$; note that we can in principle obtain a correction that is at least as tight as the PoSI correction at error level $\alpha$, by invoking the refined analysis of Theorem \ref{thm:refined}. However, this refined correction would require computing both the full PoSI correction and the local correction $q^{(\alpha - \nu)} \left(\widehat \V_\nu^+\right)$, which can be far more computationally demanding than computing only the local correction. As a result, we feel that the correction of Corollary \ref{cor:linear_reg} is more practical.

As a warmup, we instantiate Corollary \ref{cor:linear_reg} for marginal screening, which admits a simple, explicit characterization of $\widehat\M_\nu^+$. Then we study selection via the LASSO.

\begin{example}[Marginal screening] Marginal screening is a simple feature selection method that selects $\hat M = \{\hat i_1,\dots,\hat i_k\}$, where $\hat i_j$ is the $j$-th largest inner product $|X_i^\top y|$, for a pre-specified $k\in [d]$.

Let $c_{(j)}$ denote the $j$-th largest inner product $|X_i^\top y|$. Then, it is not difficult to see that $\widehat\M_\nu^+$ consists of all subsets of size $k$ of the set
$$\left\{i\in[d]:|X_i^\top y| \geq c_{(k)} - 4q^{\nu}(\{X_j\}_{j=1}^d)\right\}.$$

In words, all variables $i$ with $|X_i^\top y|$ within a $4q^{\nu}(\{X_j\}_{j=1}^d)$ margin of $c_{(k)}$ have a plausible chance of being selected. As a result, taking a simultaneous correction over them suffices to get valid inference, while all other variables can be searched through ``for free.''
\end{example}



\subsection{Model selection via the LASSO}

We discuss a method for locally simultaneous inference after model selection via the LASSO. While we focus on the LASSO, the method can be applied to any selection procedure where the selection event admits a polyhedral representation, such as forward stepwise \cite{tibshirani2016exact}. We will elaborate on this point later in the section.

Recall that the LASSO solves the following penalized regression problem:
$$\hat\beta(y) = \argmin_\beta \frac{1}{2} \|y - X\beta\|_2^2 + \lambda \|\beta\|_1,$$
and selects $\hat M = \{i\in[d]: \hat\beta(y)_i\neq 0 \}$. We will write $\hat\beta(y) \equiv \hat \beta$ when the argument is clear from the context.

The key step in applying Corollary \ref{cor:linear_reg} is to find the set of plausible models $\widehat\M_\nu^+$. More precisely, denoting $\B_\nu^\infty = \{y':\|X^\top y - X^\top y'\|_\infty \leq 2q^{\nu}(\{X_j\}_{j=1}^d)\}$ the relevant neighboring outcome vectors, the set of plausible models is $\widehat\M_\nu^+ = \{\hat M(y'):y'\in\B_\nu^\infty\}$. To simplify notation we will denote by $s_\nu = 2q^{\nu}(\{X_j\}_{j=1}^d)$ the radius of $\B_\nu^\infty$.

To find all possible models in $\B_\nu^\infty$, we apply the polyhedral characterization of the LASSO selection event due to \citet{lee2016exact}. Denoting by $\hat s = \mathrm{sign}\left(\hat\beta_{\hat M}\right)$ the signs of the selected variables in the LASSO solution, Lee et al. show that
$$\{\hat M = M, \hat s = s\} = \left\{
\begin{pmatrix} A_0^+(M,s)\\
A_0^-(M,s)\\
A_1(M,s)
\end{pmatrix} y < 
\begin{pmatrix} b_0^+(M,s)\\
b_0^-(M,s)\\
b_1(M,s)
\end{pmatrix}
\right\},$$
for any fixed model-sign pair $(M,s)$, where
\begin{align*}
&A_0^+(M,s) = \frac{1}{\lambda} X^\top_{M^c}(I - \Pi_M), \quad b_0^+(M,s) = \mathbf{1} - X_{M^c}^\top (X_M^\top)^+ s;\\
&A_0^-(M,s) = - \frac{1}{\lambda} X^\top_{M^c}(I - \Pi_M), \quad b_0^-(M,s) = \mathbf{1} + X_{M^c}^\top (X_M^\top)^+ s;\\
&A_1(M,s) = - \mathrm{diag}(s)(X_M^\top X_M)^{-1}X_M^\top, \quad b_1(M,s) = - \lambda \mathrm{diag}(s) (X_M^\top X_M)^{-1} s.
\end{align*}
Here, $\Pi_M := X_M (X_M^\top X_M)^{-1} X_M^\top$.
We will denote the polyhedron above by $P(M,s)$.

At a high level, our approach to finding $\widehat \M_\nu^+$ is the following. The Lee et al. characterization shows that the set of outcome vectors for which a model $M$ and sign vector $s$ are realized is a polyhedron. Moreover, for each active constraint of the polyhedron, meaning that the constraint is not redundant in defining the polyhedron, we know exactly which model-sign pair is on the other side of the face (depending on the constraint corresponding to the active face). The basic idea of our procedure is to compute the model-sign pair (and the corresponding polyhedron) at the data $y$, and then recursively move to neighboring polyhedra until the whole box $\B^\infty_\nu$ is tiled by the visited polyhedra. The set $\widehat \M_\nu^+$ is then simply all the models recorded in the visited polyhedra.

The described principle is agnostic to the fact that the polyhedron characterizes the LASSO selection event specifically. In particular, it works for \emph{any} selection procedure that admits a polyhedral representation. Just like in the case of the LASSO, the goal is to enumerate all polyhedra contained in $\B_\nu^\infty$, which encode the different plausible selection events, and this is precisely what our method accomplishes.

In what follows we discuss rules for determining the set of neighboring model-sign polyhedra given the current model-sign polyhedron, which make the core of our procedure. We rely on two types of rules: \emph{exact} screening rules and \emph{safe} screening rules. Exact screening rules are necessary and sufficient to screen out ``irrelevant'' variables, i.e. those whose inclusion/exclusion does not change when going from the current model-sign region to any neighboring model-sign region: they either remain in the model with the same sign in all neighboring polyhedra or they never enter the model. Safe screening rules are not exact but provide sufficient conditions for screening; we combine them with exact rules to improve computational efficiency. Our safe rules resemble prior work on variable elimination for the LASSO \cite{ghaoui2010safe, tibshirani2012strong}, but are fundamentally different as they rely on properties of $\B_\nu^\infty$. It is worth mentioning that the safe rules are LASSO-specific; the exact rules work for general selection strategies with a polyhedral characterization. 

We use $\B(M,s)$ to denote the set of model-sign pairs whose corresponding polyhedra neighbor, i.e. share a face with, $P(M,s)$.

\paragraph{Exact screening rules.} Exact screening rules proceed by checking for each variable $i\in[d]$ if it can change its inclusion/exclusion status when going from the current model-sign polyhedron $P(M,s)$ to any neighboring polyhedron. In other words, for each variable $i\in M$, they check if there exists a pair $(M',s')\in \B(M,s)$ such that $i\not\in M'$; similarly, for each $i\in M^c$, they check if there exists a pair $(M',s')\in \B(M,s)$ such that $i\in M'$, and they additionally identify the corresponding sign of variable $i$ if such a pair exists.


The core idea of exact screening rules is to find the minimal representation of $P(M,s)\cap \B_\nu^\infty$. That is, the goal is to prune all redundant constraints coming from $P(M,s)$; the inequalities that remain are ``active'' and indicate that the variables corresponding to those constraints can enter or leave the model in one of the neighboring polyhedra. In Algorithm \ref{alg:exact_screening} in the Appendix we use a standard solution to finding a minimal polyhedral representation, which relies on solving one linear program for each constraint whose redundancy is being checked.

\paragraph{Safe screening rules.}
Safe rules serve to speed up the search for a minimal representation of a polyhedron corresponding to a model-sign pair.
For all $y'\in P(M,s)$, the LASSO optimality conditions imply that the LASSO solution is locally linear, namely
$$\hat \beta(y') = \beta_{(M,s)}(y') := (X_M^\top X_M)^{-1}(X_M^\top y' - \lambda s).$$
Note that, while $\beta_{(M,s)}(y')$ is equal to the LASSO solution for $y'\in P(M,s)$, it can be computed for $y'\not\in P(M,s)$. We use this characterization to design the safe screening rules.

\begin{lemma}[Safe exclusion]
\label{lemma:safe_Mc}
Fix a model-sign pair $(M,s)$. Let
$$\mathcal{I}_{\mathrm{safe}}^-(M,s) := \left\{j\in M^c: |X_j^\top (y - X_M\beta_{(M,s)}(y))| < \lambda - s_\nu \left(1 + \|X_j^\top X_M(X_M^\top X_M)^{-1}\|_1\right)\right\}.$$
Then, for any $j\in \mathcal{I}_{\mathrm{safe}}^-(M,s)$, variable $j$ cannot enter the model in any of the neighboring polyhedra:
$$\forall (M',s')\in \B(M,s),~ j\not\in M'.$$
\end{lemma}

\begin{lemma}[Safe inclusion]
\label{lemma:safe_M}
Fix a model-sign pair $(M,s)$. Let
$$\mathcal{I}_{\mathrm{safe}}^+(M,s) = \left\{j\in M: |\beta_{j\cdot (M,s)}(y)| > s_\nu \left\|e_{j\cdot (M,s)}^\top (X_M^\top X_M)^{-1}\right\|_1 \right\}.$$
Then, for any $j\in \mathcal{I}_{\mathrm{safe}}^+(M,s)$, variable $j$ cannot exit the model in any of the neighboring polyhedra: 
$$\forall (M',s')\in \B(M,s),~j\in M'.$$
\end{lemma}

\begin{algorithm}[t]
\SetAlgoLined
\SetKwInOut{Input}{input}
\Input{design matrix $X$, outcome vector $y$, penalty $\lambda$, error level $\alpha$, parameter $\nu\in(0,\alpha)$}
\textbf{output:} set of plausible models $\widehat \M_\nu^+$\newline
Compute width of $\B^\infty_\nu$: $s_\nu = 2q^{\nu}(\{X_j\}_{j=1}^d)$\newline
Compute LASSO solution: $\hat \beta = \argmin_\beta \frac{1}{2}\|y - X\beta\|_2^2 + \lambda \|\beta\|_1$\newline
Let $\hat M = \mathrm{supp}(\hat \beta)$, $\hat s = \mathrm{sign}(\hat \beta_{\hat M})$\newline
Initialize $\P_\mathrm{todo} \leftarrow \{(\hat M,\hat s)\}, \P_\nu^+ \leftarrow \emptyset$\newline
\While{$\P_{\mathrm{todo}}\neq \emptyset$}
{\ Take any pair $(M,s)\in \P_{\mathrm{todo}}$\newline
Update $(M,s)$ as visited: $\P_{\mathrm{todo}} \leftarrow \P_{\mathrm{todo}} \setminus \{(M,s)\}$, $\P_\nu^+ \leftarrow \P_\nu^+ \cup \{(M,s)\}$\newline
$\mathcal{I}_{\mathrm{safe}}(M,s) \leftarrow \mathrm{SafeScreening}(X,y,(M,s))$ \quad (Alg.~\ref{alg:safe_screening})\newline
$\B(M,s) \leftarrow \mathrm{ExactScreening}(X,y,(M,s), \mathcal{I}_{\mathrm{safe}}(M,s))$ \quad (Alg.~\ref{alg:exact_screening})\newline
$\P_{\mathrm{todo}} \leftarrow \P_{\mathrm{todo}} \cup (\B(M,s) \setminus \P_\nu^+)$
}
Return $\widehat \M_\nu^+ = \{M: \exists s \text{ s.t. } (M,s)\in \P_\nu^+\}$
\caption{Locally simultaneous inference for the LASSO}
\label{alg:local_lasso}
\end{algorithm}

Lemma \ref{lemma:safe_Mc} and Lemma \ref{lemma:safe_M} show that the safe screening subroutine (stated explicitly in Algorithm~\ref{alg:safe_screening} in the Appendix) is valid. Putting everything together, we formalize the guarantees of Algorithm \ref{alg:local_lasso} in Theorem \ref{thm:local_lasso}.


\begin{theorem}
\label{thm:local_lasso}
Algorithm \ref{alg:local_lasso} returns exactly the set of plausible models, i.e.
$$\widehat \M_\nu^+ = \left\{\hat M(y'): \|X^\top y - X^\top y'\|_\infty \leq 2q^{\nu}(\{X_j\}_{j=1}^d)\right\}.$$
\end{theorem}

Putting together Theorem \ref{thm:local_lasso} and Corollary \ref{cor:linear_reg}, we conclude that it suffices to take a simultaneous correction in the sense of \citet{berk2013valid} at error level $\alpha-\nu$, \emph{only} over the local model set $\widehat\M_\nu^+$, to get a valid confidence region for $\theta_{\hat M}$.

\section{Population risk of empirical risk minimizer}
\label{sec:erm}

Finally, we consider an application to empirical process theory. In standard theory, the population risk of an empirical risk minimizer is typically bounded by relying on some notion of complexity, such as VC-dimension or Rademacher complexity, of the hypothesis class over which the empirical risk is optimized. We show that locally simultaneous inference implies a bound on the population risk that relies only on the complexity of those hypotheses that could \emph{plausibly} be empirical risk minimizers in light of the data, while all clearly suboptimal hypotheses, no matter how complex, come at virtually no cost. This application serves as another example of how locally simultaneous inference is applicable in nonparametric settings.

To set the problem up formally, suppose we have a dataset $\D = \{z_i\}_{i=1}^n \sim P^n$, which should typically be thought of as consisting of feature--outcome pairs. We are interested in the \emph{empirical risk minimizer} on this dataset within a hypothesis class $\F$:
$$\hat f = \argmin_{f\in\F} R_n(f,\D) := \argmin_{f\in\F} \frac{1}{n} \sum_{i=1}^n \ell(f,z_i),$$
where $\ell(f,z)$ measures the loss incurred by predicting on point $z$ using hypothesis $f$ (for example, if $z = (x,y)$ is a feature--outcome pair, we could have $\ell(z,f) = \mathbf{1}\{y \neq f(x)\}$).

We would like to bound the \emph{population risk} of $\hat f$. For a fixed hypothesis $f\in\F$, we define its population risk as:
\begin{equation*}
R(f,P) = \E_{z\sim P} \ell(f,z).
\end{equation*}
A standard way of bounding $R(\hat f, P)$ is to apply a \emph{simultaneous correction} over $\F$:
$$R(\hat f, P) = R_n(\hat f, \D) + (R(\hat f, P) -  R_n(\hat f, \D)) \leq R_n(\hat f, \D) + \sup_{f\in\F} |R(f,P) -  R_n(f,\D)|.$$
Therefore, it suffices to bound the so-called \emph{generalization gap}, $ \sup_{f\in\F} |R(f,P) -  R_n(f,\D)|$, to get a valid upper bound on $R(\hat f, P)$. Such a bound follows by relying on the concentration of $\sup_{f\in\F} |R(f,P) - R_n(f,\D)|$ around its mean, where the latter is controlled by relying on a complexity measure of $\F$, such as VC-dimension or Rademacher complexity.

The general result on locally simultaneous inference implies that it suffices to bound the generalization gap over \emph{near empirical risk minimizers} on $\D$. Notably, this is a data-dependent hypothesis class. If there are many hypotheses whose empirical risk on $\D$ is far worse than $R_n(\hat f,\D)$, then this set could be significantly smaller than $\F$.

In the following we let $\texttt{Gap}_n(\F)$ denote any valid upper bound on $\E  \sup_{f\in\F} |R(f,P) - R_n(f,\D)|$. Such a bound typically follows from a complexity argument. For example, the classical symmetrization argument proves that
$\texttt{Gap}_n(\F) = 2\mathcal R_n(\F)$
is a valid upper bound on $\E  \sup_{f\in\F} |R(f,P) - R_n(f,\D)|$, where $\mathcal R_n(\F)$ denotes the Rademacher complexity of $\{\ell(f,\cdot)\}_{f\in\F}$.

\begin{theorem}
\label{thm:risk_minimizer}
Fix $\alpha\in(0,1)$ and $\nu\in(0,\alpha)$. Assume that $|\ell(f,z)|\leq 1$ for all $z$ and $f\in\F$. Consider the data-dependent hypothesis class:
$$\widehat \F_\nu^+ = \left\{f\in\F:  R_n(f,\D) \leq  R_n(\hat f, \D) + 4 \emph{\texttt{Gap}}_n(\F) + 4 \sqrt{\frac{2}{n} \log\left(\frac{1}{\nu}\right)}\right\}.$$
Then,
$$
P\left\{R(\hat f,P) \leq R_n(\hat f,\D) + \emph{\texttt{Gap}}_n (\widehat \F_\nu^+ ) + \sqrt{\frac{2}{n} \log\left(\frac{1}{\alpha - \nu}\right)}\right\} \geq 1 - \alpha.
$$
\end{theorem}

Theorem \ref{thm:risk_minimizer} relies on the fact that the generalization gap of the empirical risk minimizer, $|R_n(\hat f,\D) - R(\hat f,P)|$, can be bounded using a data-dependent complexity term $\texttt{Gap}_n(\widehat \F_\nu^+)$.

\section{Numerical evaluation}
\label{sec:exps}

We compare locally simultaneous inference to its underlying simultaneous inference method, the conditional method due to Lee et al.~\cite{lee2016exact}, and the hybrid refinement of conditional inference due to Andrews et al.~\cite{andrews2019inference} and McCloskey~\cite{mccloskey2020hybrid}. Code for reproducing the experiments is available in the repository at \href{https://github.com/tijana-zrnic/locally-simultaneous-inference}{this link}.

Throughout we apply the version of locally simultaneous inference from Theorem \ref{thm:local_posi_main} with $\nu=0.1\alpha$.
In all figures comparing interval widths we plot the median width over $100$ trials, together with the $5\%$ and $95\%$ percentile, plotted as error bars around the median. The target error level is $\alpha = 0.1$ throughout.

\begin{figure}[t]
    \centering
    \includegraphics[width=0.24\textwidth]{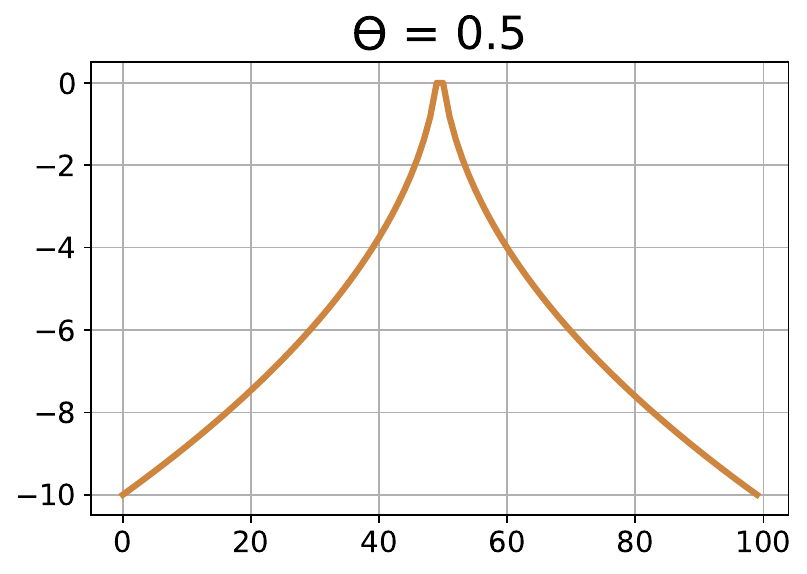}
    \includegraphics[width=0.24\textwidth]{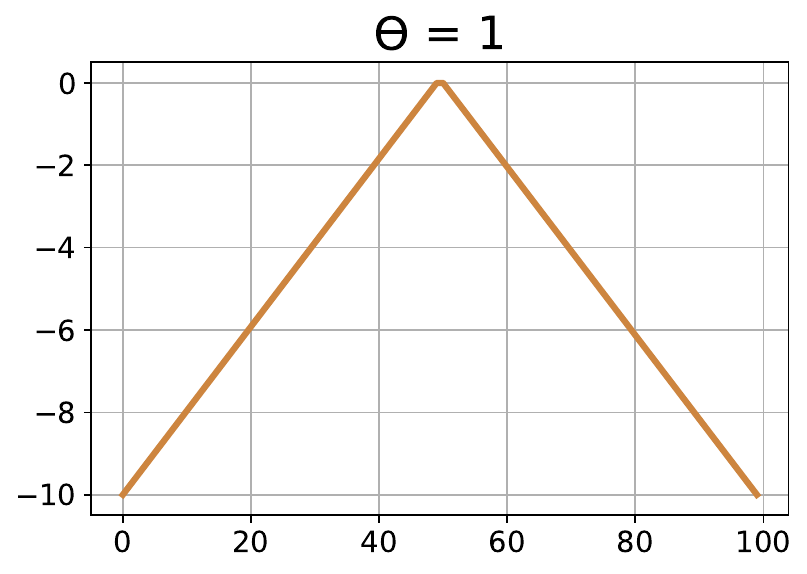}
    \includegraphics[width=0.24\textwidth]{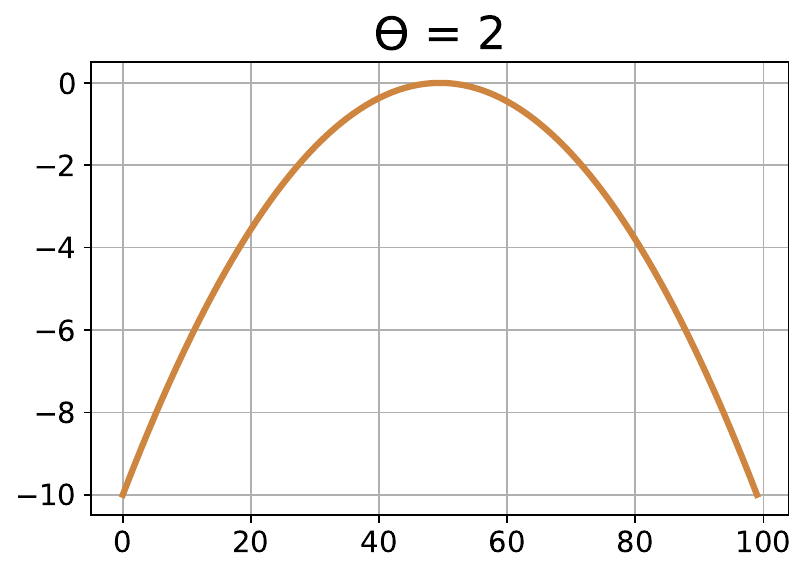}
    \includegraphics[width=0.24\textwidth]{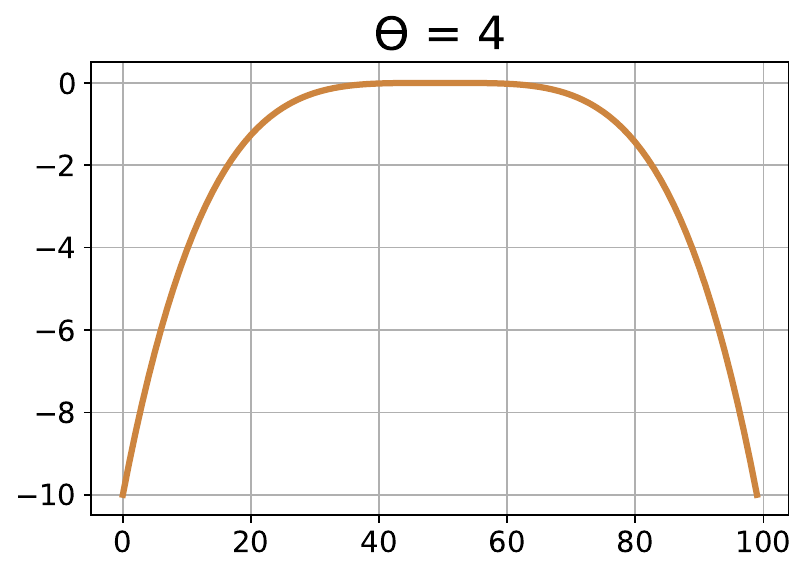}
    \caption{Mean outcome $\mu$ for different problem parameters $\theta$, at scale level $C = 10$.}
\label{fig:mu_examples}
\end{figure}

\subsection{Inference on the winner}
\label{sec:iow_exp}

We begin by studying the problem of inference on the winner from Section \ref{sec:winner}. We generate the mean outcome $\mu$ as a smooth curve; this simulates a setting where nearby entries of $\mu$ are similar, such as when the data is a time series or when neighboring entries of $\mu$ correspond to outcomes in neighboring subgroups (e.g., neighboring age groups). We vary the shape of the mean outcome vector $\mu$, thereby making inference more or less challenging for the different methods. We let $\mu_i \propto -|i - 0.5(m+1)|^\theta$ for $i\in[m]$, where $\theta>0$ varies the sharpness of $\mu$. Small $\theta$ corresponds to the case where the winner stands out, while large $\theta$ makes the mean outcome flat, implying that many observations have a plausible chance of being selected as the winner. The other tuning parameter is $C>0$: we rescale the mean $\mu$ so that the difference between the minimum and maximum entry of $\mu$ is equal to $C$. When $C$ is large, $\mu$ gets ``stretched out'' and, as a result, there are fewer candidates that can plausibly be selected. In Figure \ref{fig:mu_examples} we plot the shape of $\mu$ for different values of $\theta$, at $C = 10$.

\begin{figure}[t]
    \centering
    \includegraphics[width=0.24\textwidth]{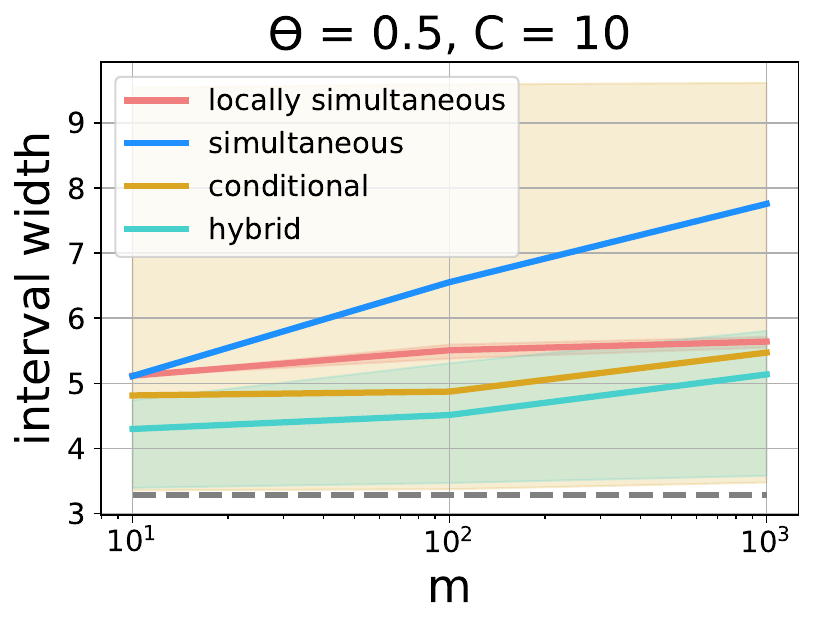}
    \includegraphics[width=0.24\textwidth]{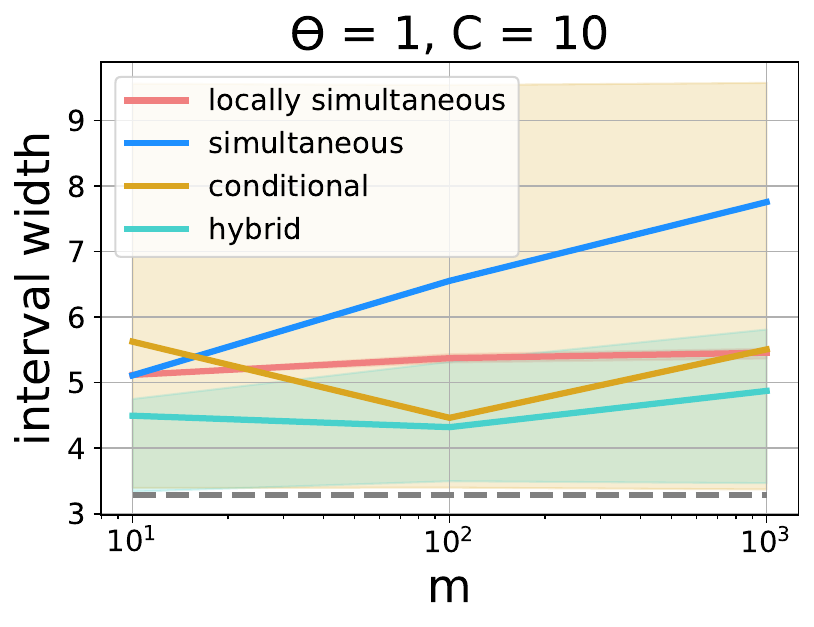}
    \includegraphics[width=0.24\textwidth]{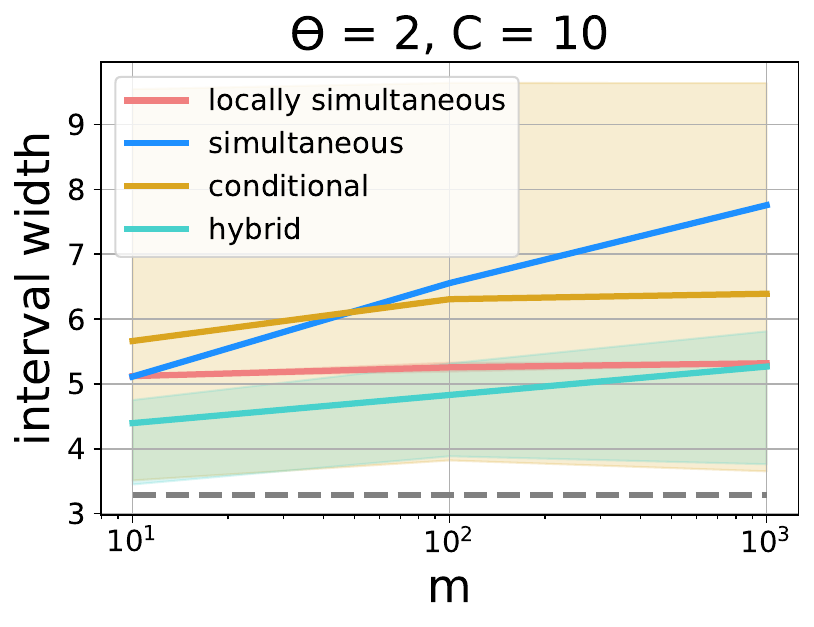}
    \includegraphics[width=0.24\textwidth]{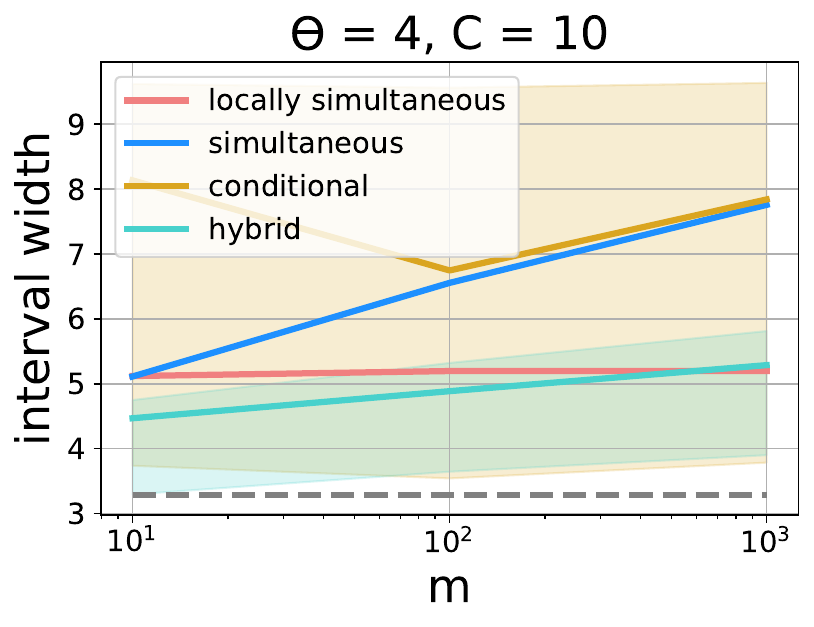}
    \includegraphics[width=0.24\textwidth]{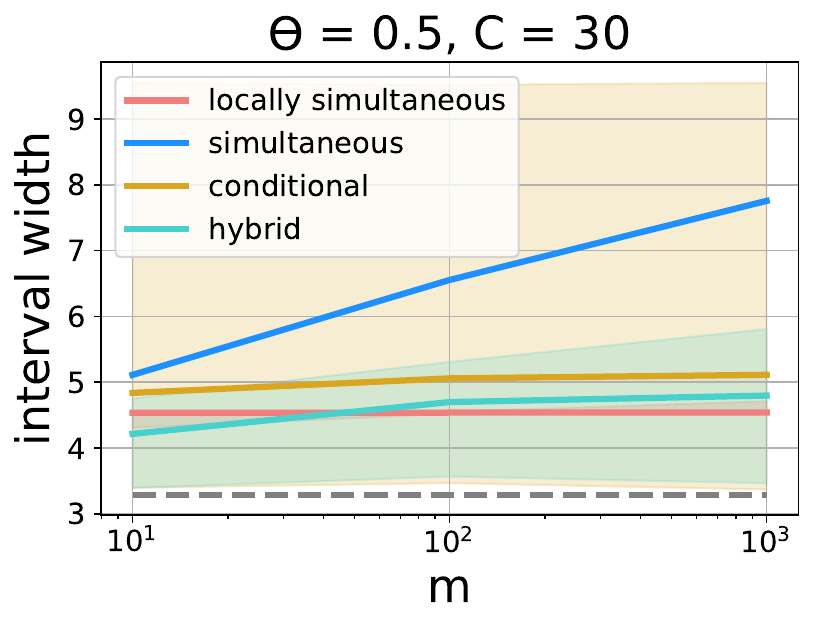}
    \includegraphics[width=0.24\textwidth]{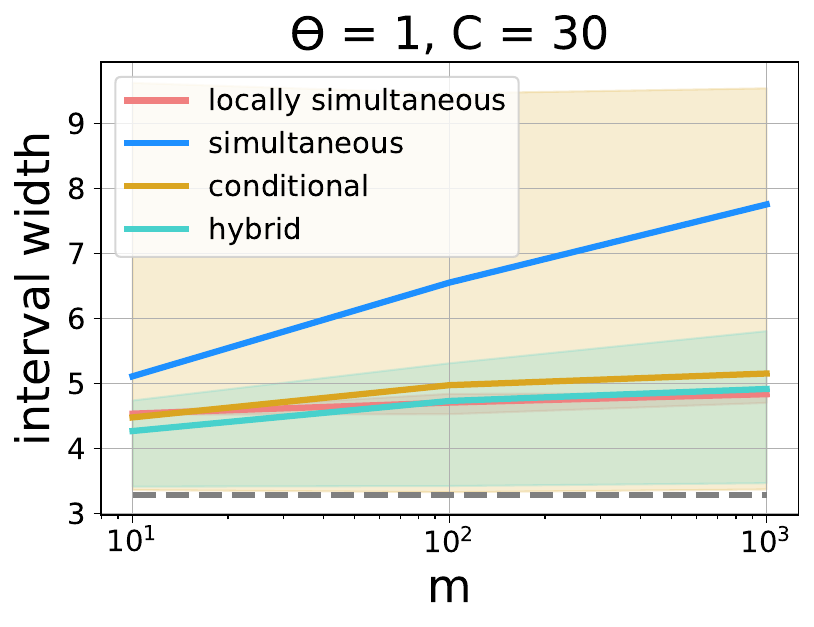}
    \includegraphics[width=0.24\textwidth]{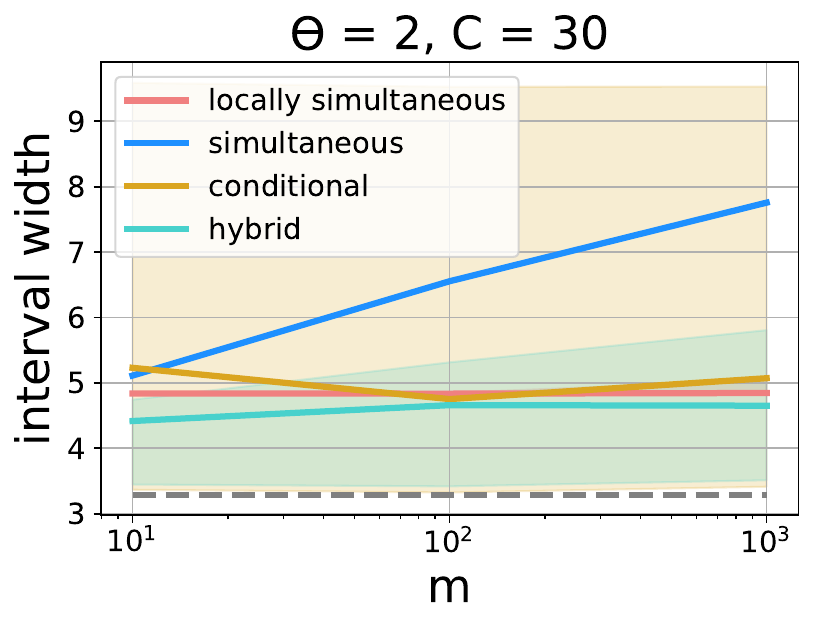}
    \includegraphics[width=0.24\textwidth]{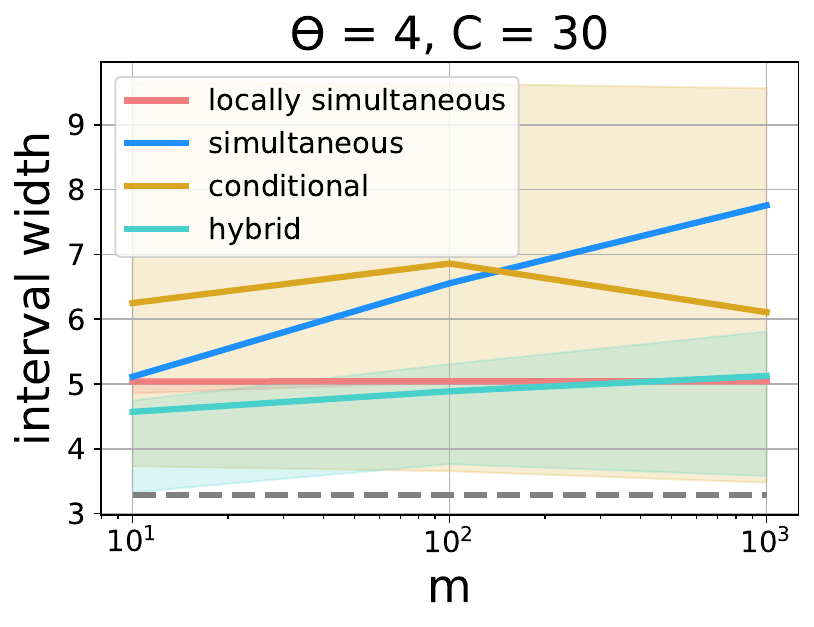}
    \includegraphics[width=0.24\textwidth]{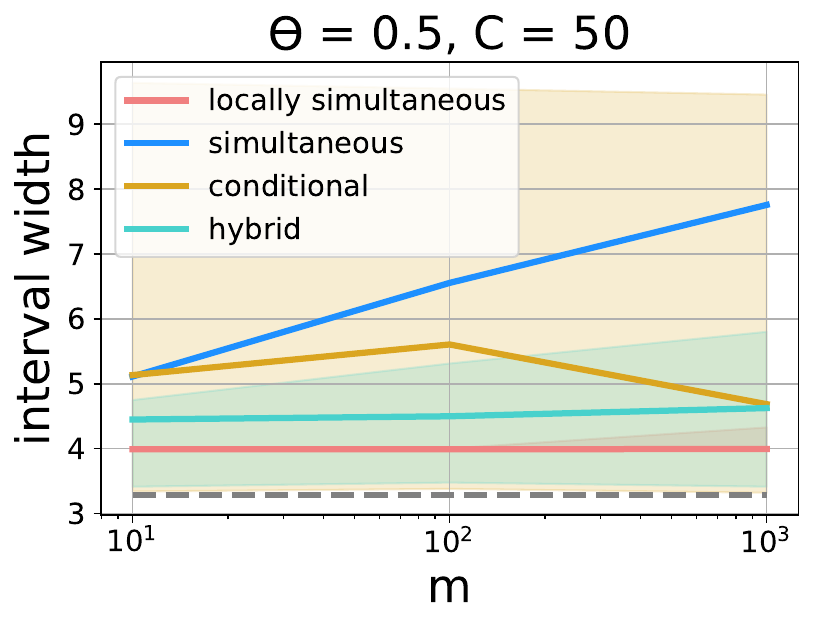}
    \includegraphics[width=0.24\textwidth]{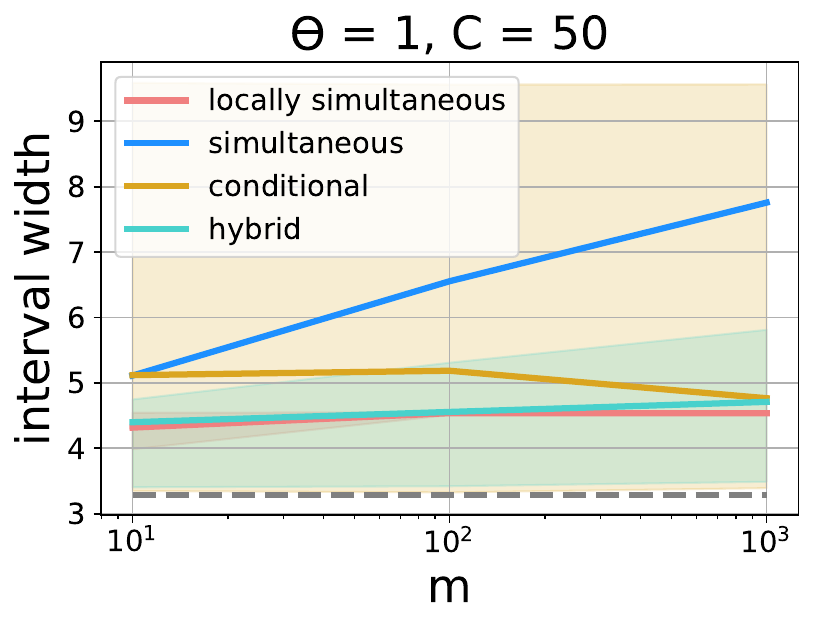}
    \includegraphics[width=0.24\textwidth]{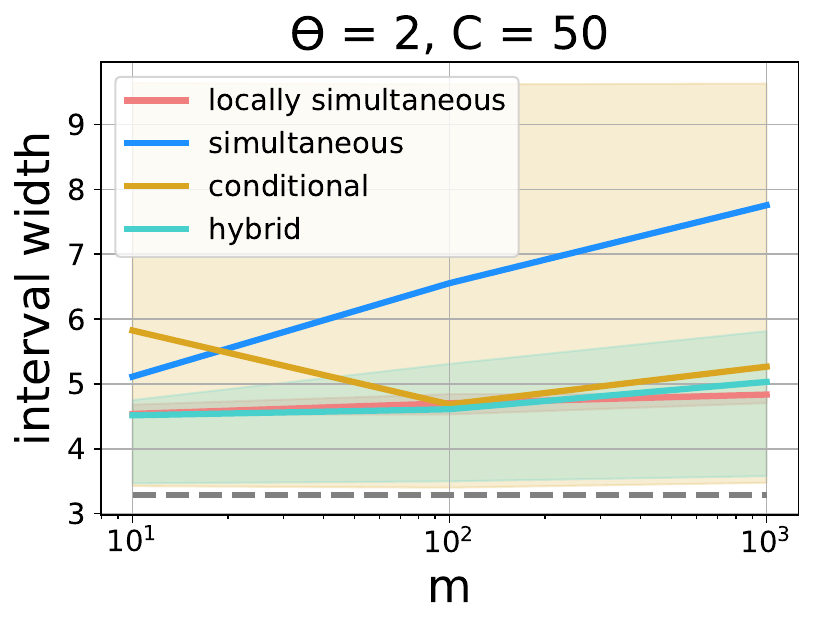}
    \includegraphics[width=0.24\textwidth]{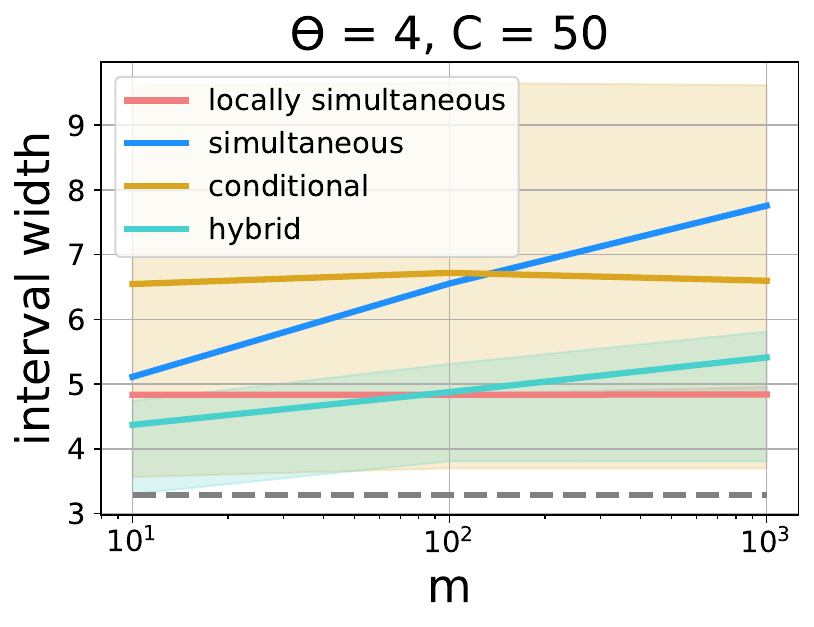}
    \caption{Interval width achieved by locally simultaneous, fully simultaneous, conditional, and hybrid inference in the problem of inference on the winner. The dashed line corresponds to nominal interval widths.}
\label{fig:smooth_mu}
\end{figure}

\paragraph{Parametric case.} In the first setting, we generate the vector of observations as $y = \mu + \xi$, where $\xi\sim \cN(0,I_m)$. In Figure \ref{fig:smooth_mu}, we plot the interval width resulting from locally simultaneous, simultaneous, conditional, and hybrid inference for varying $\theta\in\{0.5,1,2,4\}$, $C\in\{10,30,50\}$, and $m\in\{10,100, 1000\}$. The mean $\mu$ has range $C$ at $m=10$ and for higher $m$ it is not renormalized to range $C$; the purpose of increasing $m$ is to demonstrate the behavior of the different methods when the number of irrelevant observations (i.e., those far from the winning observation) increases. We observe that conditional inference exhibits high variability for all problem parameters, and as $\theta$ grows---meaning $\mu$ becomes flat---the median intervals become large. Simultaneous inference is by construction only sensitive to changes in $m$, and its intervals grow with $m$ despite the fact that only the number of irrelevant observations grows. Locally simultaneous inference is most sensitive to changes in $C$: as $\mu$ is stretched over a larger range, the method finds fewer plausible candidates and thus leads to smaller intervals. Moreover, it is virtually insensitive to increasing $m$. The hybrid approach exhibits high variability like the conditional approach (albeit to a more moderate extent) and its intervals grow with $m$ because, as $m\rightarrow \infty$, the hybrid method reduces to standard conditional inference.

\begin{figure}[t]
    \centering
    \includegraphics[width=0.24\textwidth]{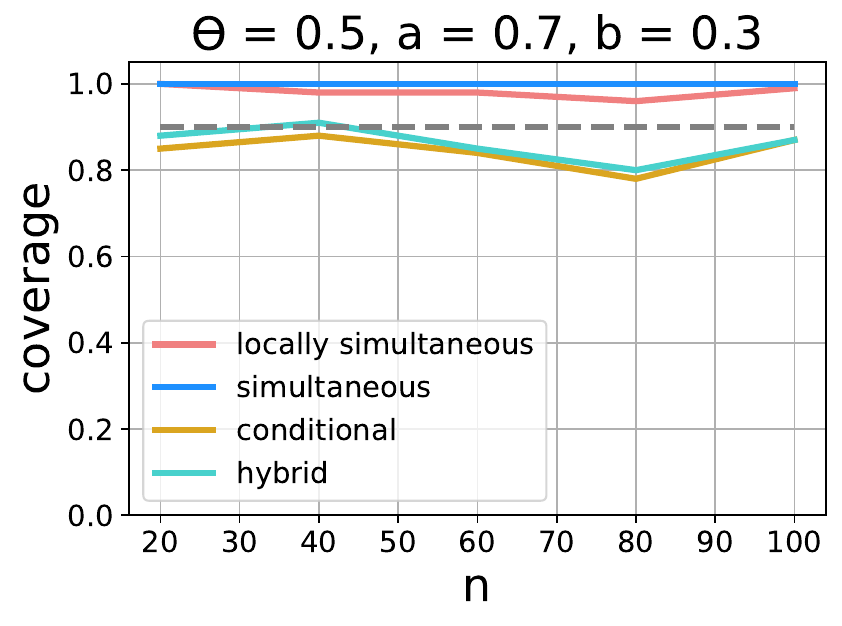}
    \includegraphics[width=0.24\textwidth]{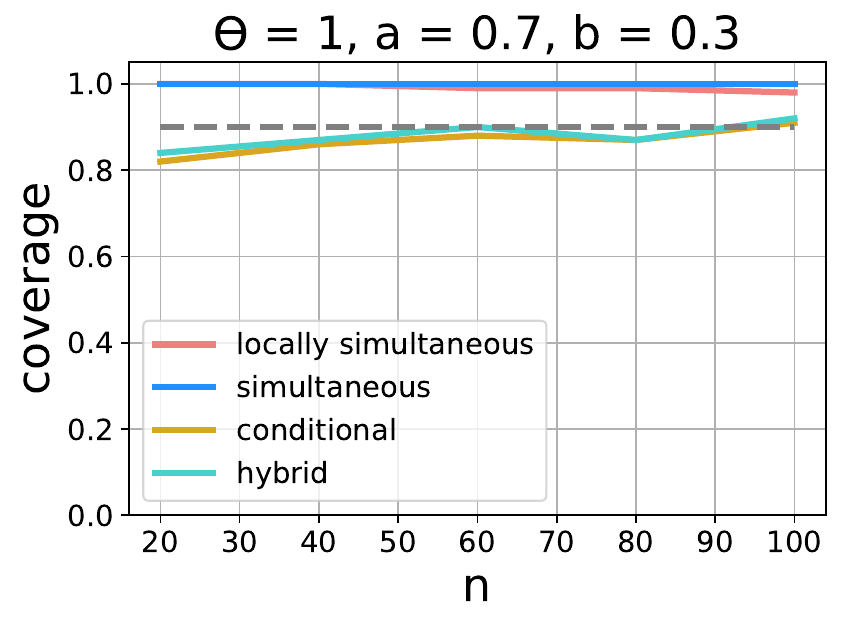}
    \includegraphics[width=0.24\textwidth]{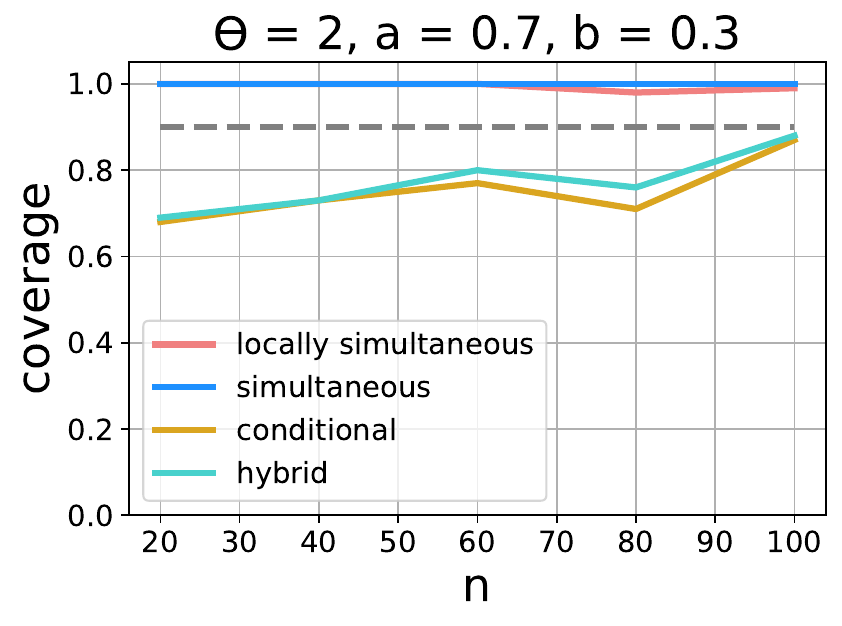}
    \includegraphics[width=0.24\textwidth]{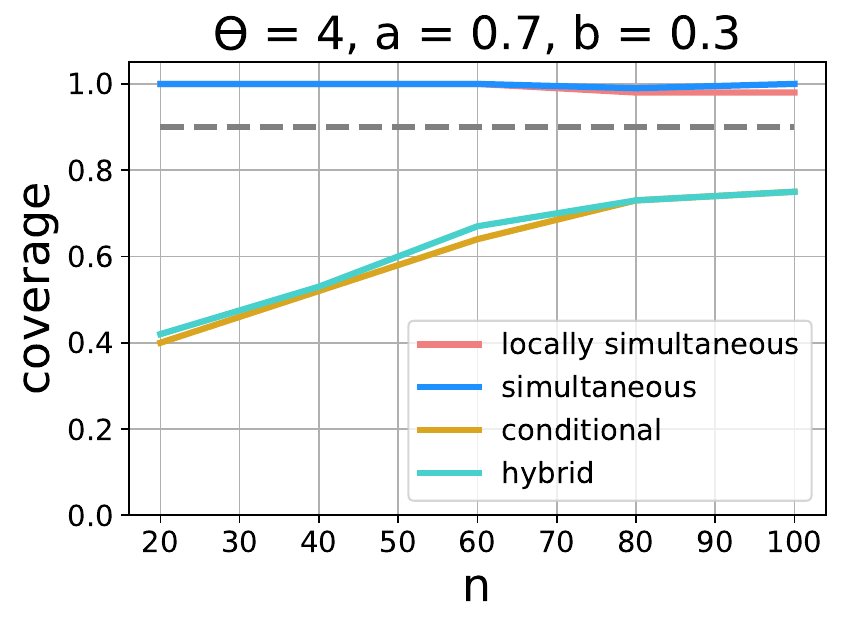}
    \includegraphics[width=0.24\textwidth]{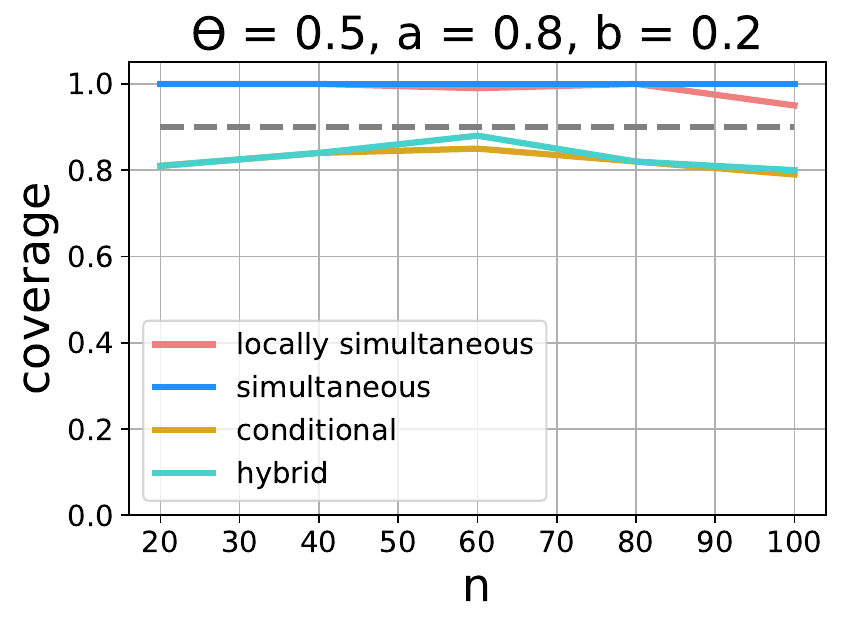}
    \includegraphics[width=0.24\textwidth]{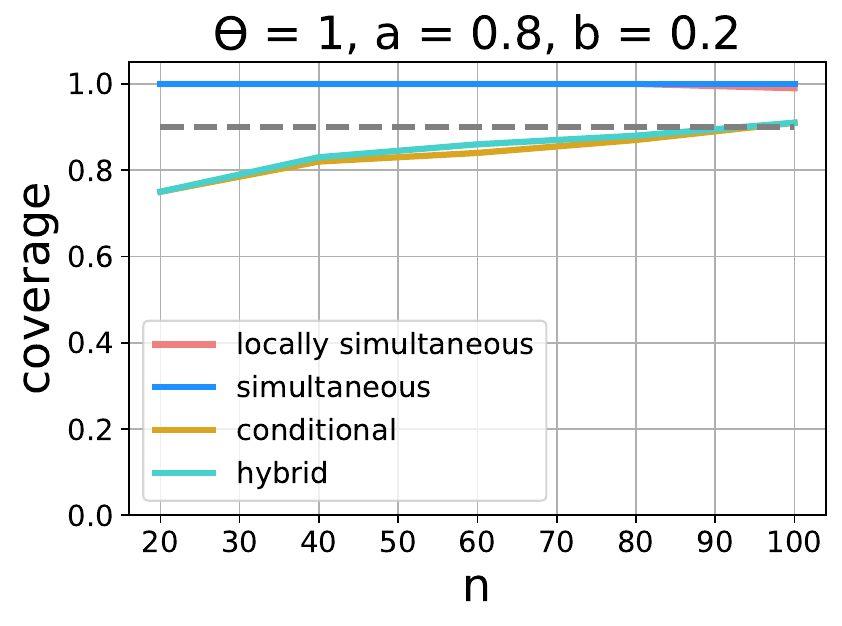}
    \includegraphics[width=0.24\textwidth]{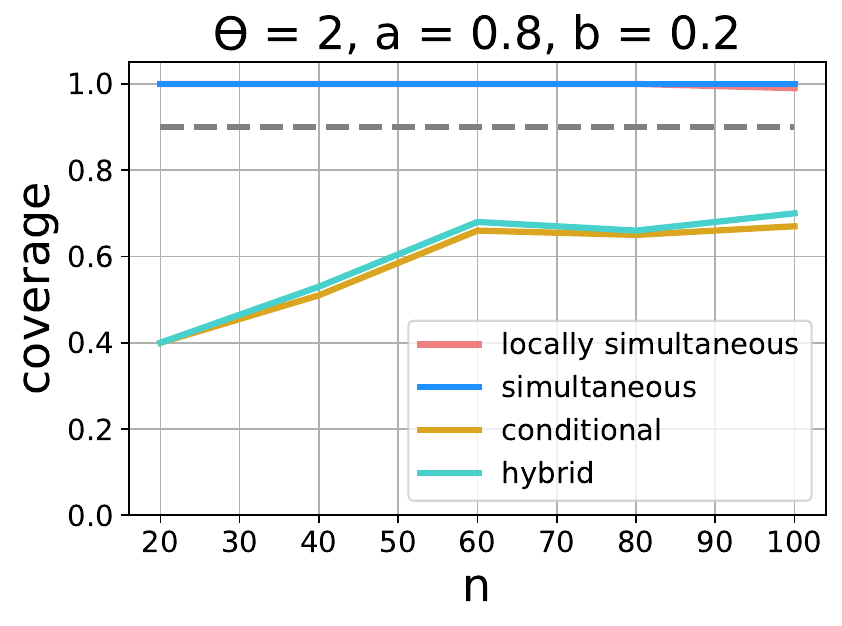}
    \includegraphics[width=0.24\textwidth]{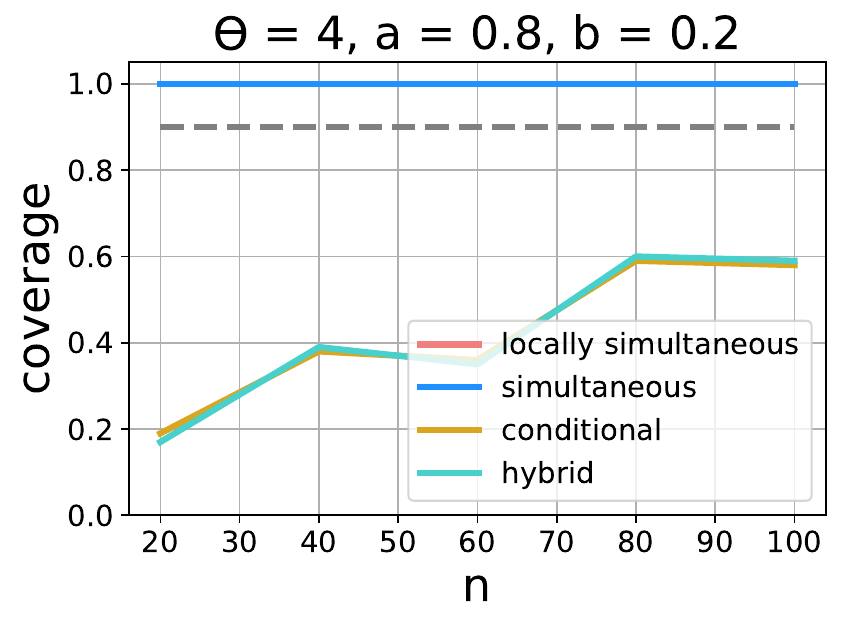}
    \includegraphics[width=0.24\textwidth]{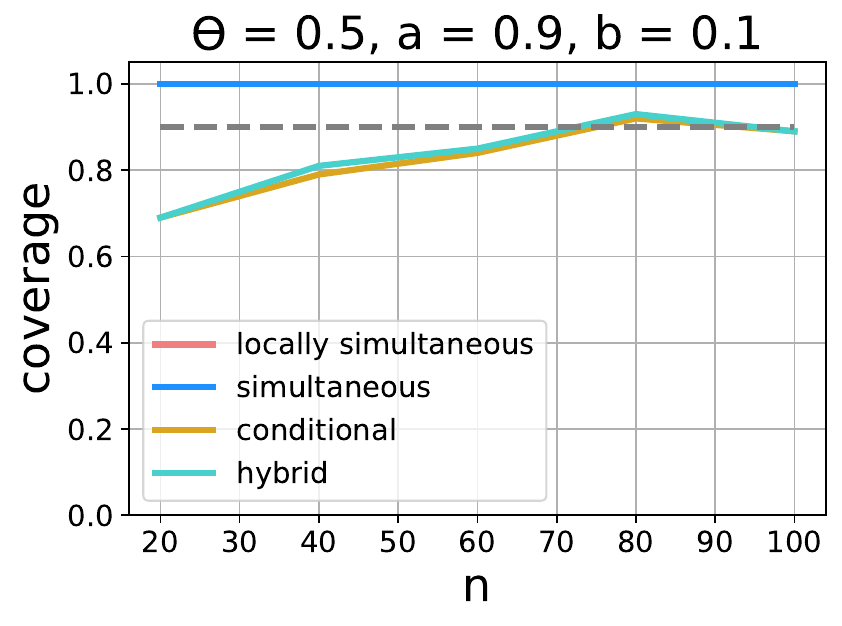}
    \includegraphics[width=0.24\textwidth]{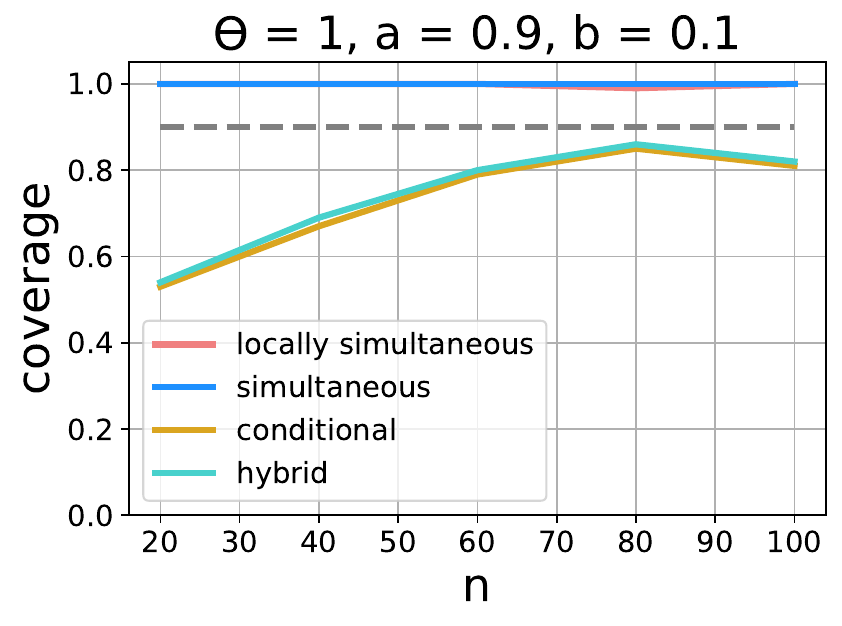}
    \includegraphics[width=0.24\textwidth]{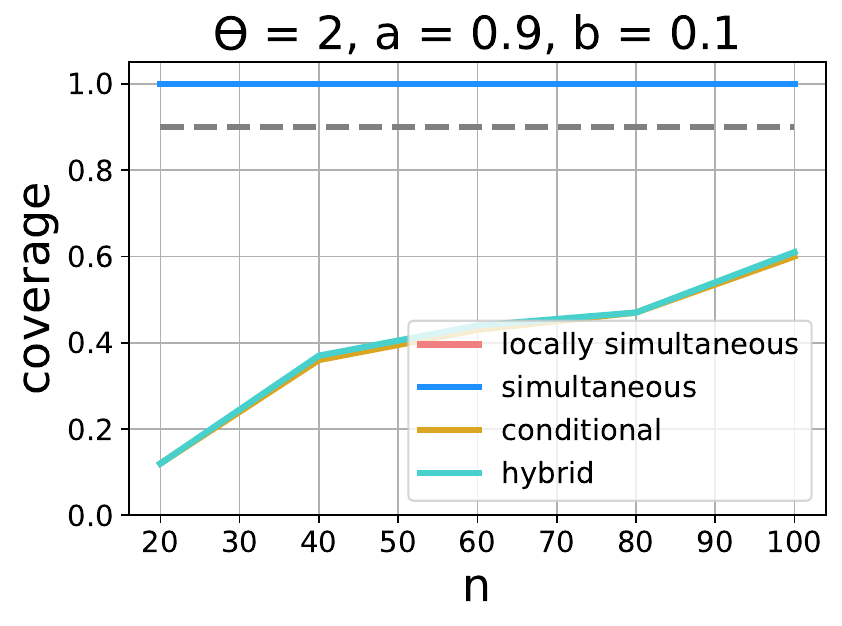}
    \includegraphics[width=0.24\textwidth]{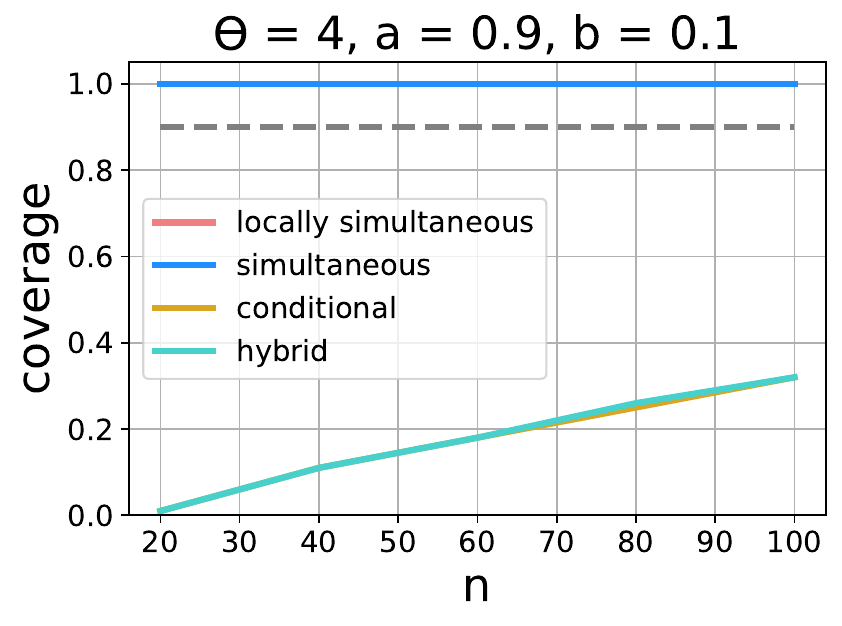}
    \caption{Coverage of locally simultaneous, fully simultaneous, conditional, and hybrid inference when the noise is sampled from $\mathrm{Beta}(a,b)$. The conditional and hybrid approaches use a normal approximation; the locally simultaneous and fully simultaneous approaches use nonparametric, finite-sample-valid confidence intervals due to Waudby-Smith and Ramdas~\cite{waudby2020estimating}. The target coverage is $0.9$, indicated by the dashed line.}
\label{fig:smooth_mu_np}
\end{figure}

\paragraph{Nonparametric case.} We emphasized that locally simultaneous inference is rigorously applicable in nonparametric settings, while conditional approaches are not. Still, it might seem like a reasonable heuristic to apply conditional inference after a normal approximation based on the CLT. We test this heuristic empirically, comparing to a nonparametric application of locally and fully simultaneous inference. We observe that the heuristic application of conditional methods can severely undercover the target.

We fix $C=20$, $m = 100$, and vary $\theta$ to obtain the mean vector $\mu$. Given $\mu$, we generate $n$ i.i.d. samples $y^{(1)},\dots,y^{(n)}$, where $y^{(j)} = \mu + \xi^{(j)}$ and $\xi^{(j)}$ has i.i.d. entries sampled from $\mathrm{Beta}(a,b)$. To apply the locally and fully simultaneous methods, we use the betting-based confidence intervals by Waudby-Smith and Ramdas~\cite{waudby2020estimating} (Theorem 3), together with a Bonferroni correction. To form the acceptance region of the locally simultaneous method, we use the Bentkus concentration inequality \cite{bentkus2004hoeffding}.
In Figure~\ref{fig:smooth_mu_np} we plot the coverage of all four approaches for varying $a,b$, and sample size $n$. We observe that, as $\theta$ grows, the conditional methods have diminishing coverage. This confirms the need for a more robust, nonparametrically applicable correction. In contrast, the two simultaneous methods have valid coverage and typically overcover, which is to be expected given the use of nonparametric concentration inequalities. In Figure \ref{fig:int_width_np} we plot the interval width implied by the four methods. The conditional methods yield much smaller intervals, but this comes at the cost of invalid coverage, as shown in Figure \ref{fig:smooth_mu_np}. The locally simultaneous intervals are consistently smaller than the fully simultaneous intervals, with the improvement being more pronounced when there are few plausible candidates, that is, when $\theta$ is small. Moreover, as $n$ grows, the locally simultaneous intervals gradually approach the conditional intervals; this makes sense seeing that the coverage of the conditional methods improves with $n$.

\begin{figure}[t]
    \centering
    \includegraphics[width=0.24\textwidth]{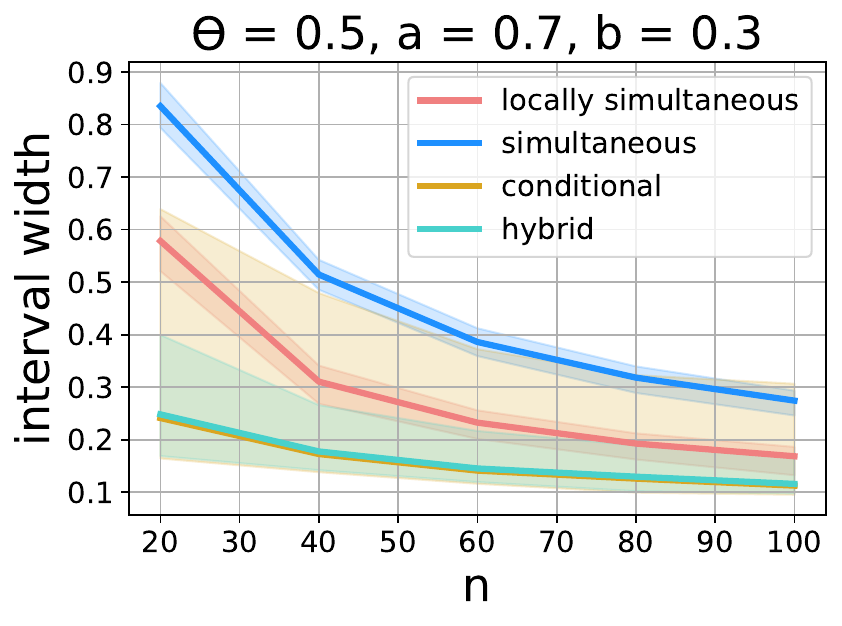}
    \includegraphics[width=0.24\textwidth]{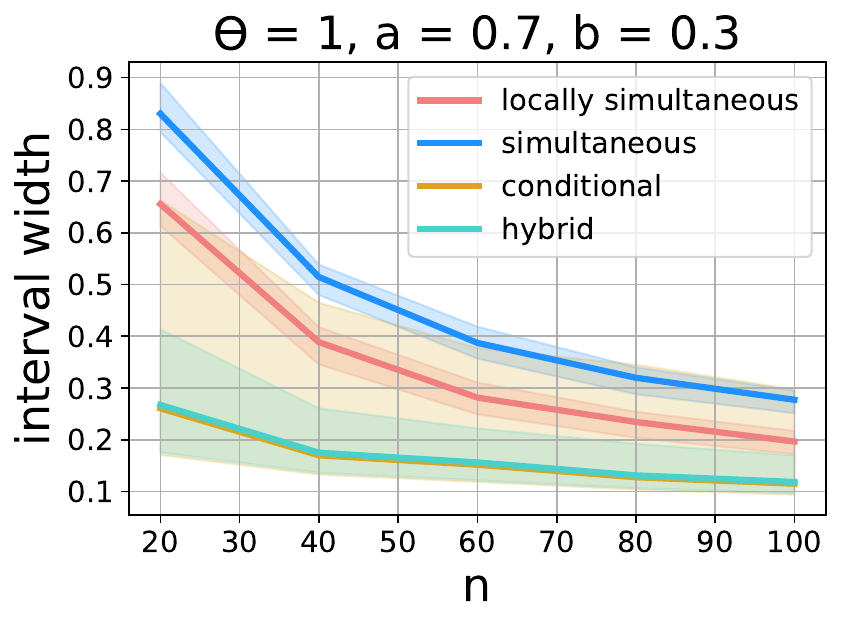}
    \includegraphics[width=0.24\textwidth]{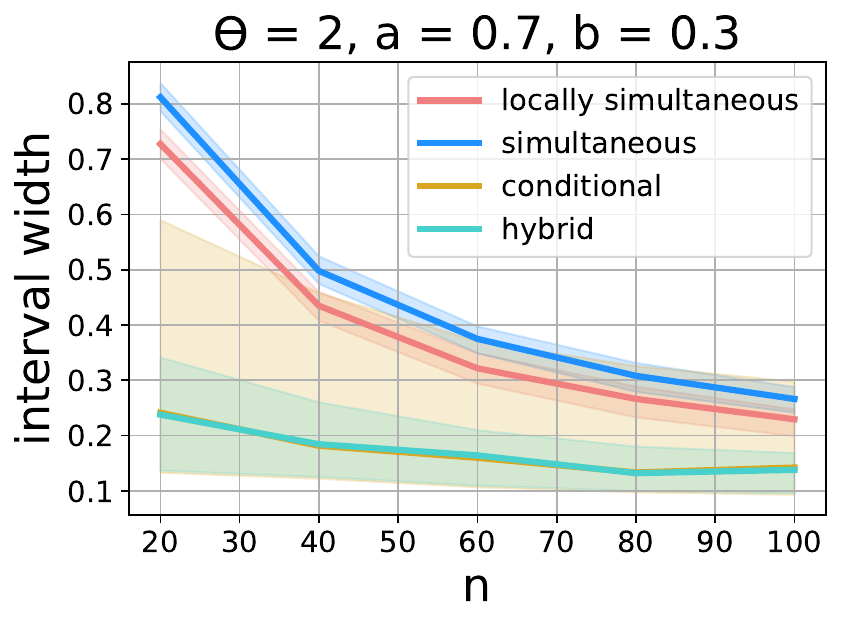}
    \includegraphics[width=0.24\textwidth]{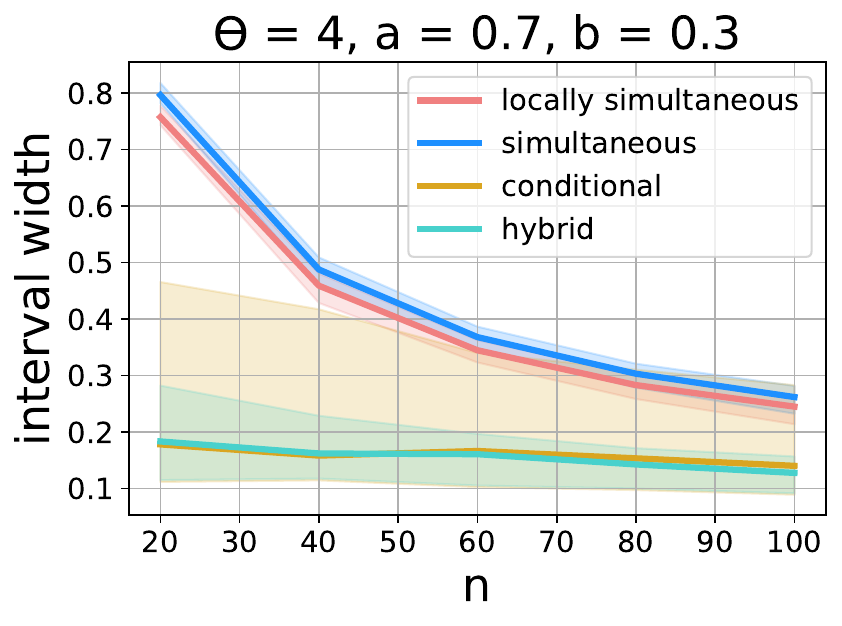}
    \includegraphics[width=0.24\textwidth]{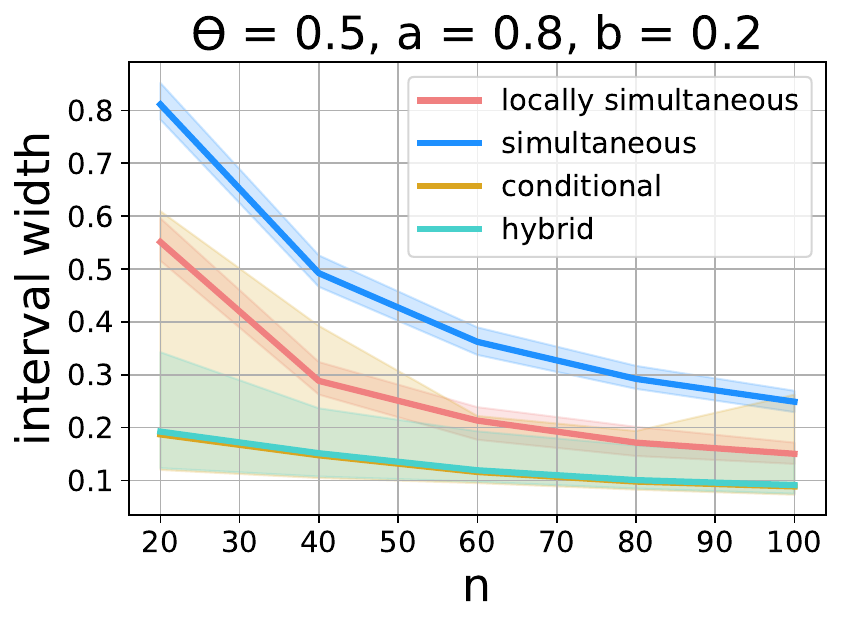}
    \includegraphics[width=0.24\textwidth]{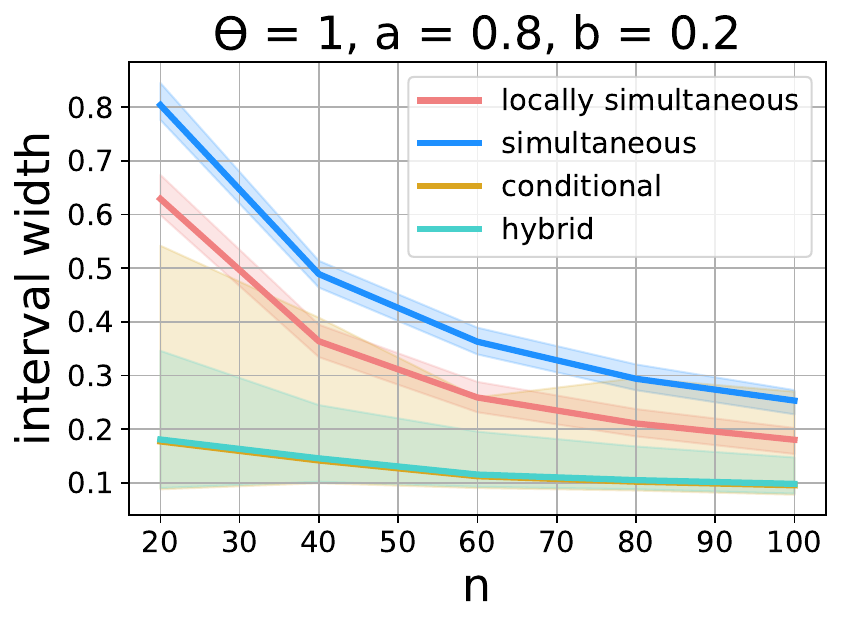}
    \includegraphics[width=0.24\textwidth]{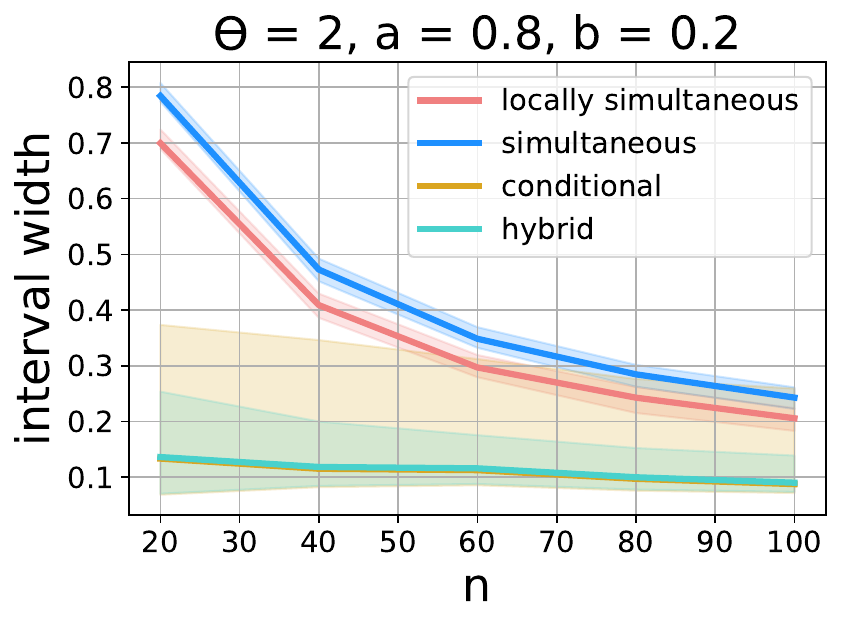}
    \includegraphics[width=0.24\textwidth]{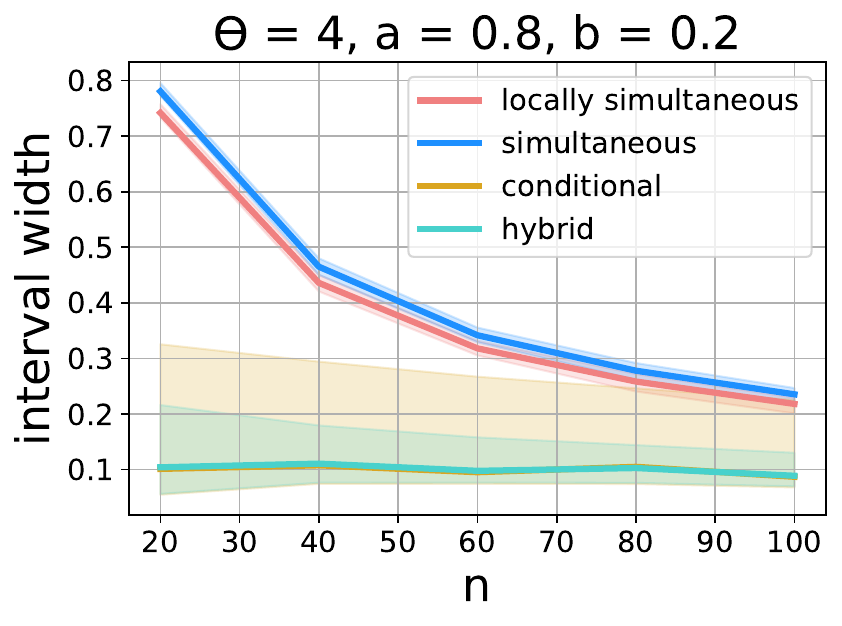}
    \includegraphics[width=0.24\textwidth]{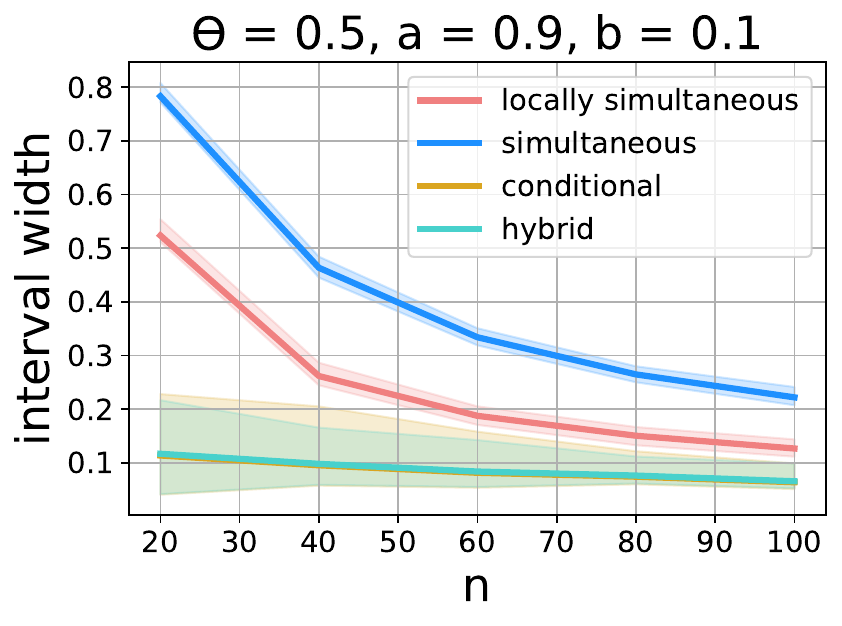}
    \includegraphics[width=0.24\textwidth]{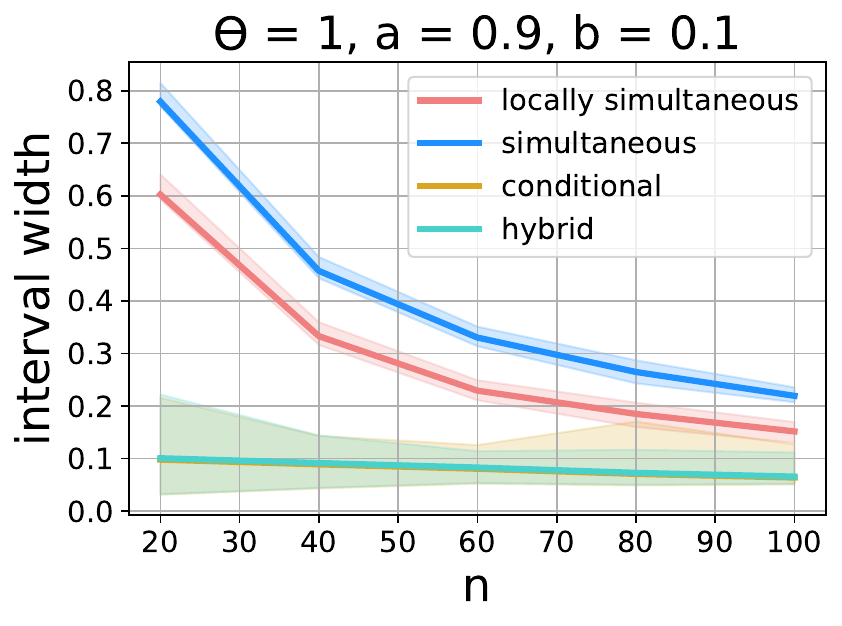}
    \includegraphics[width=0.24\textwidth]{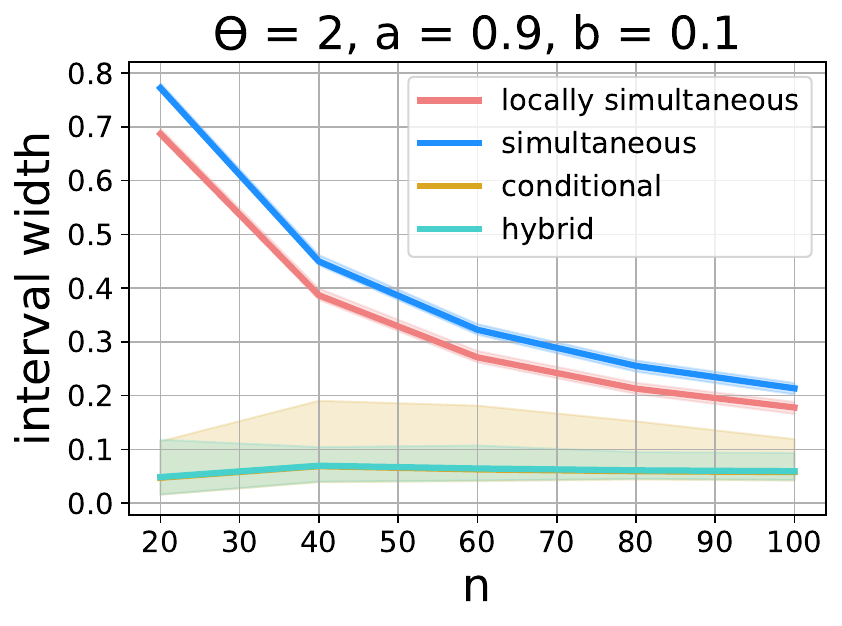}
    \includegraphics[width=0.24\textwidth]{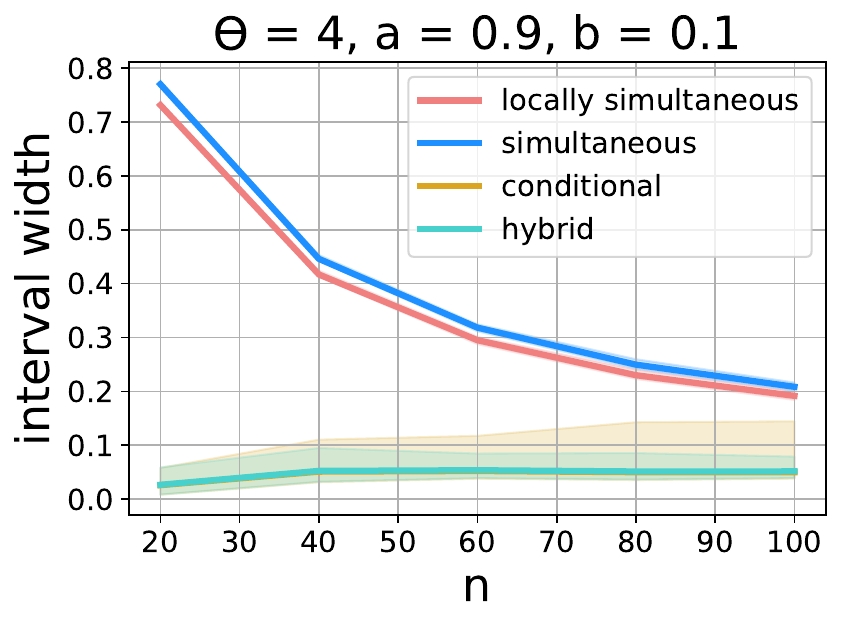}
    \caption{Interval width achieved by locally simultaneous, fully simultaneous, conditional, and hybrid inference when the noise is sampled from $\mathrm{Beta}(a,b)$. The conditional and hybrid approaches use a normal approximation; the locally simultaneous and fully simultaneous approaches use nonparametric, finite-sample-valid confidence intervals due to Waudby-Smith and Ramdas~\cite{waudby2020estimating}.}
\label{fig:int_width_np}
\end{figure}

\begin{figure}[t]
    \centering
    \includegraphics[width=0.24\textwidth]{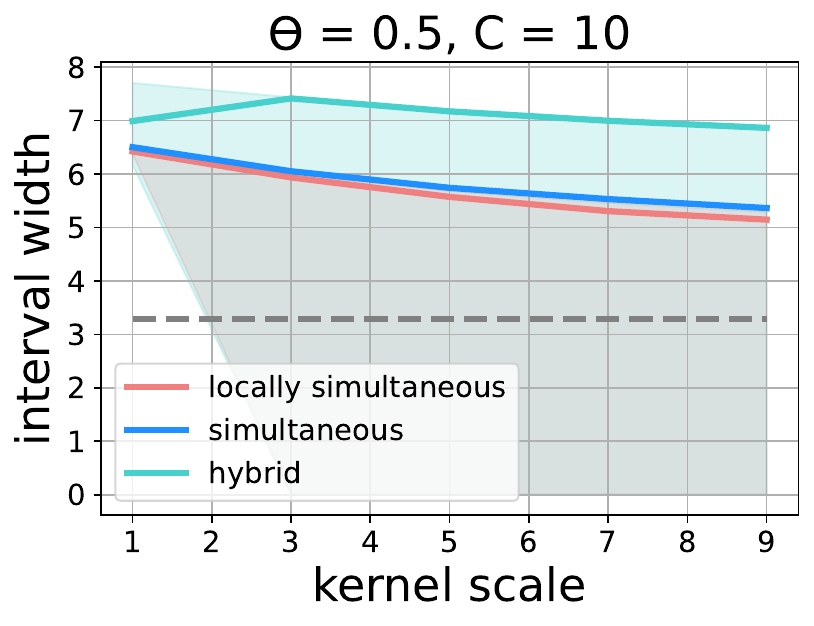}
    \includegraphics[width=0.24\textwidth]{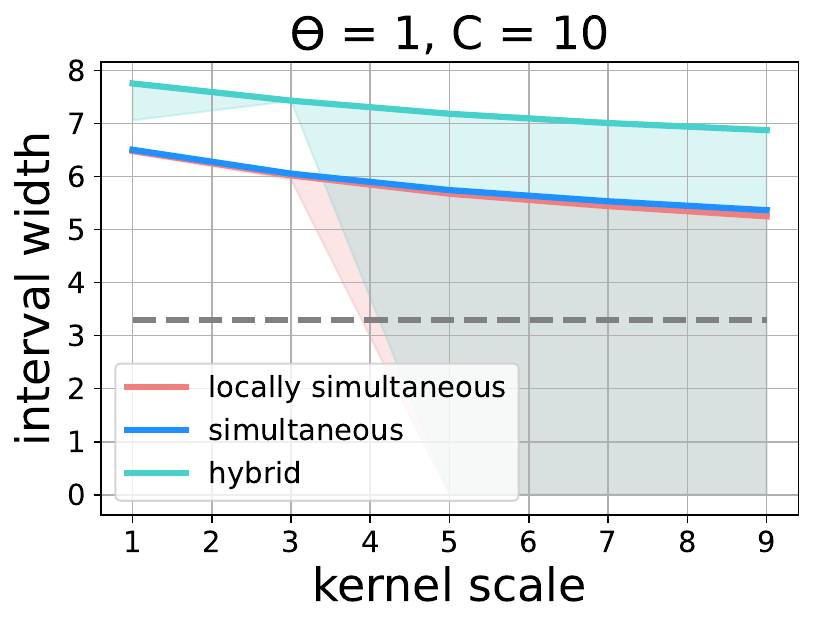}
    \includegraphics[width=0.24\textwidth]{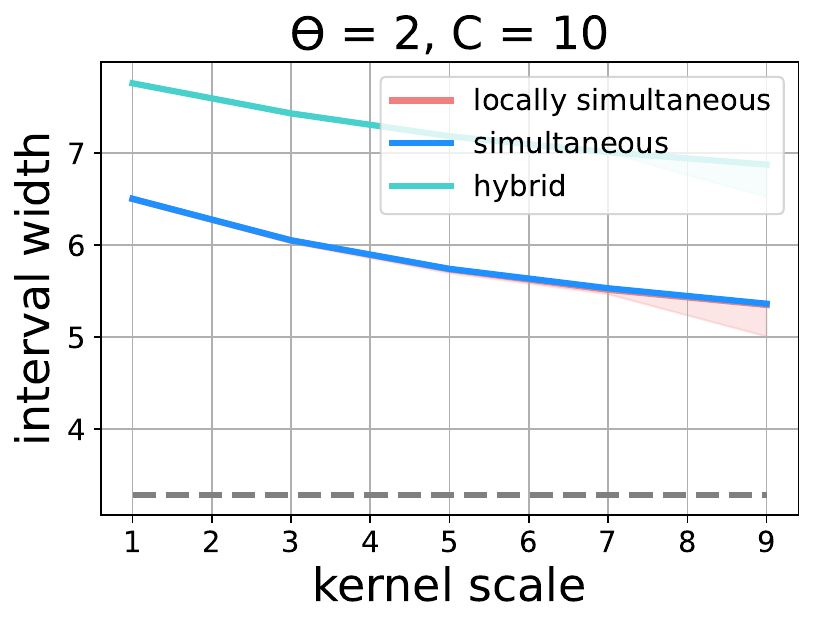}
    \includegraphics[width=0.24\textwidth]{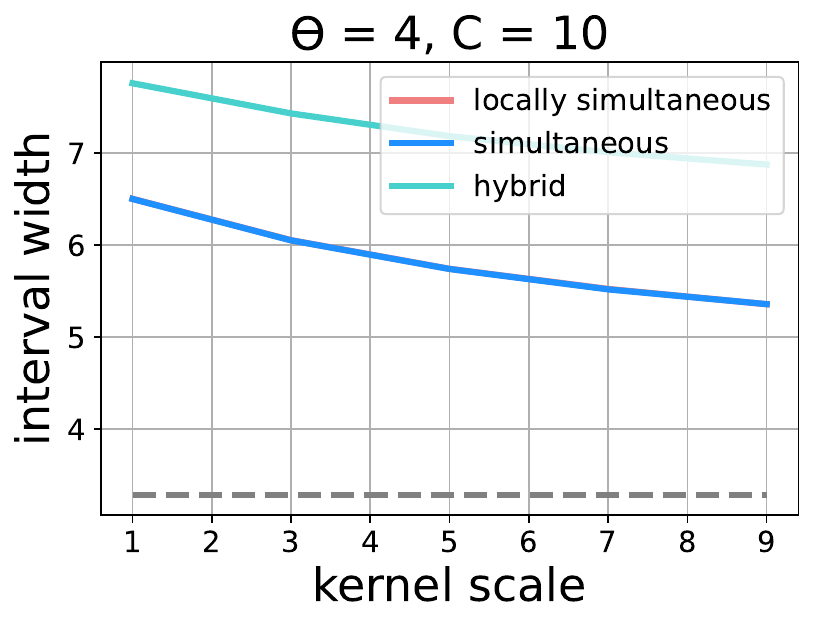}
    \includegraphics[width=0.24\textwidth]{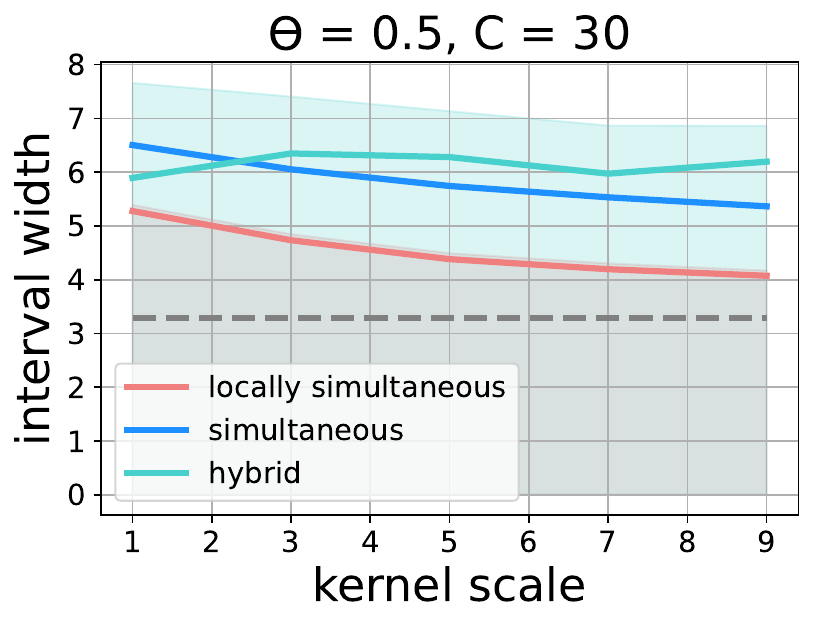}
    \includegraphics[width=0.24\textwidth]{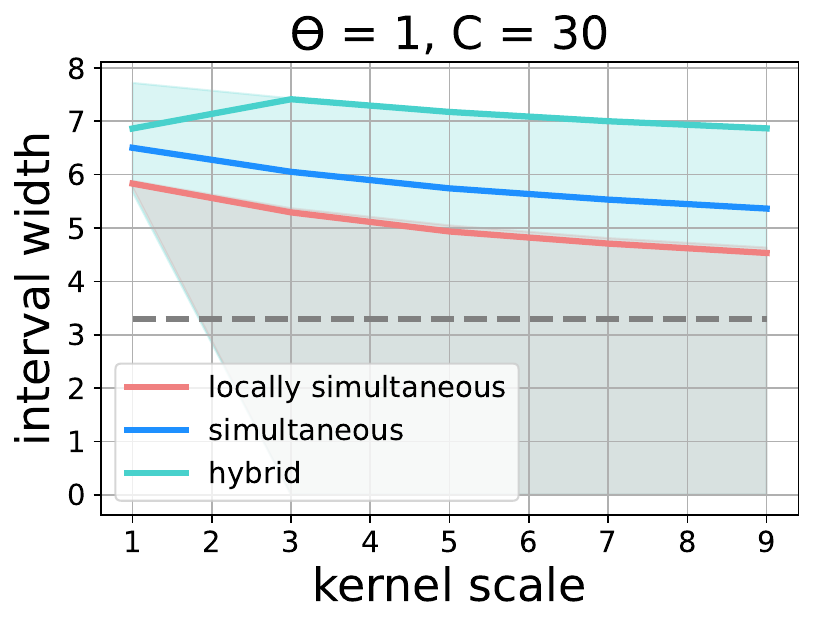}
    \includegraphics[width=0.24\textwidth]{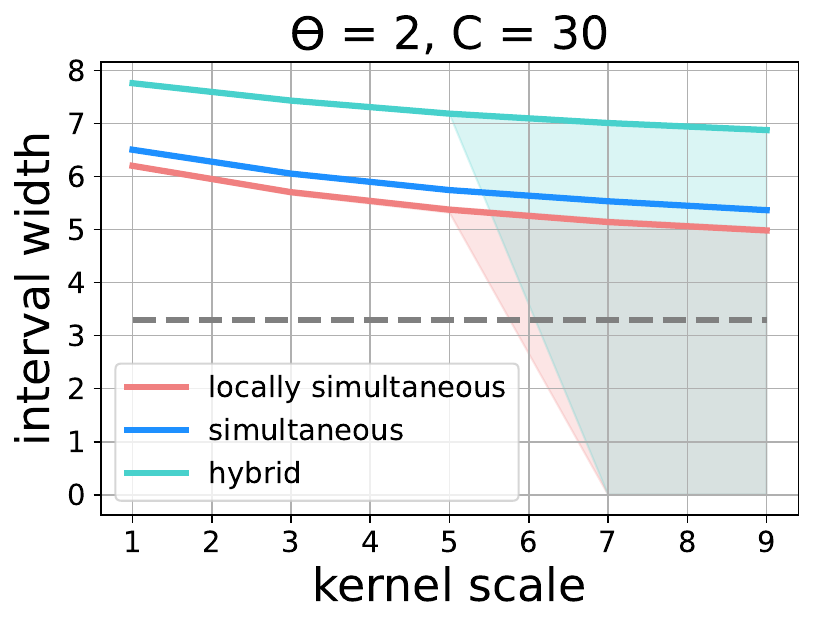}
    \includegraphics[width=0.24\textwidth]{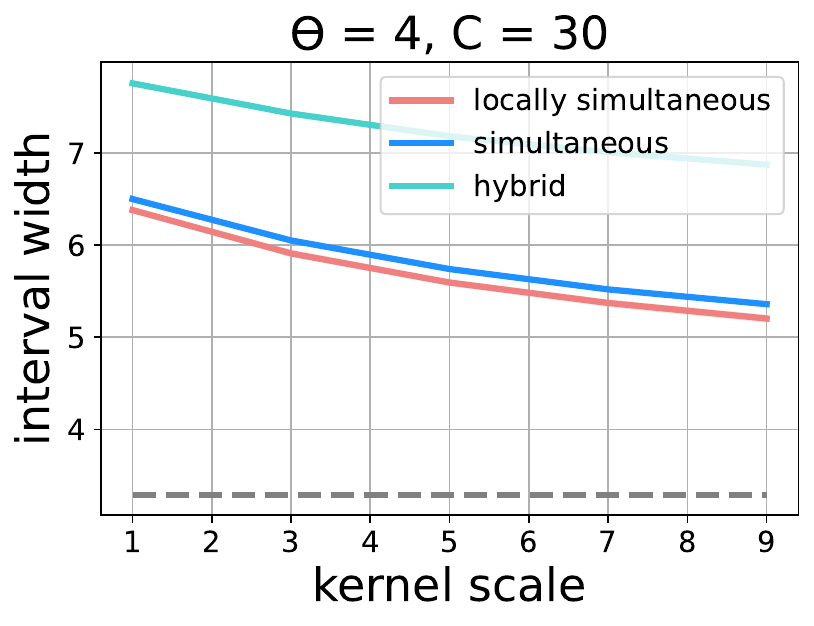}
       \includegraphics[width=0.24\textwidth]{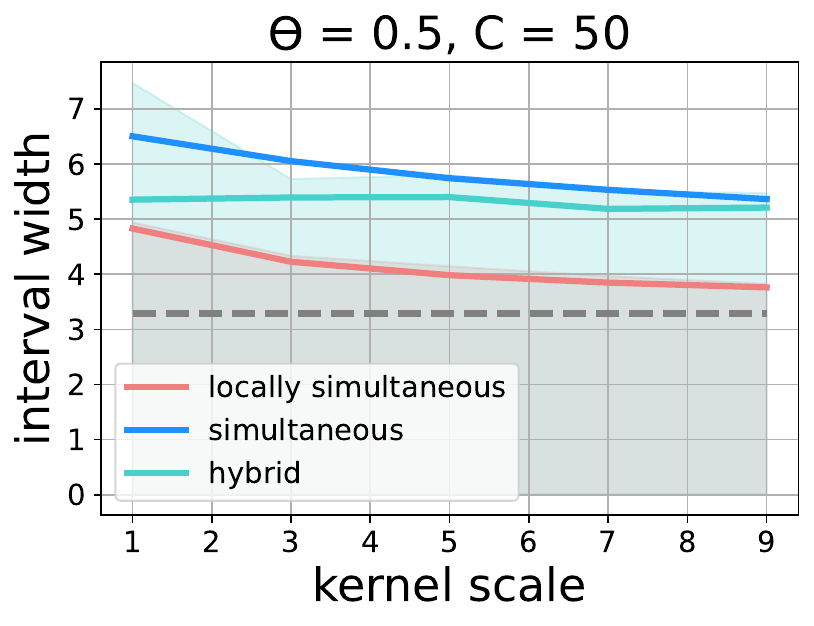}
    \includegraphics[width=0.24\textwidth]{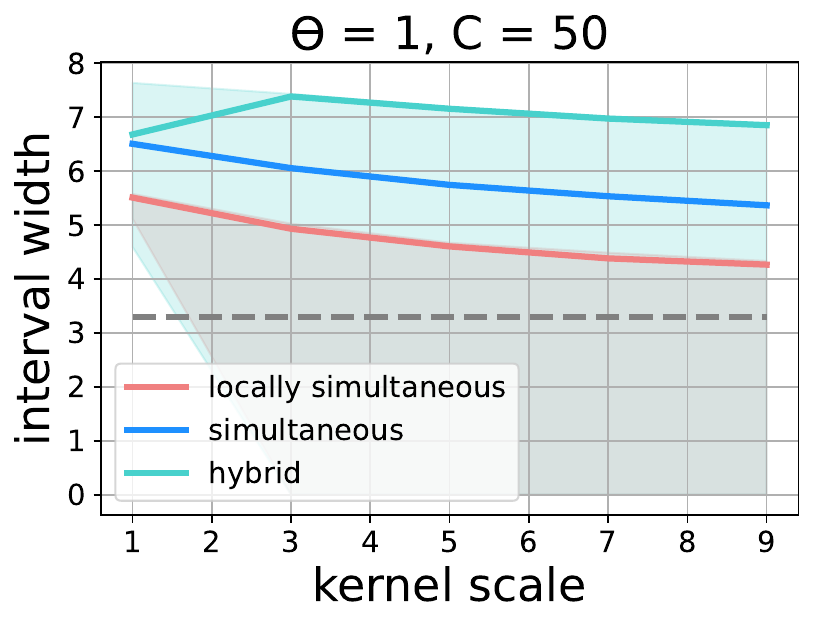}
    \includegraphics[width=0.24\textwidth]{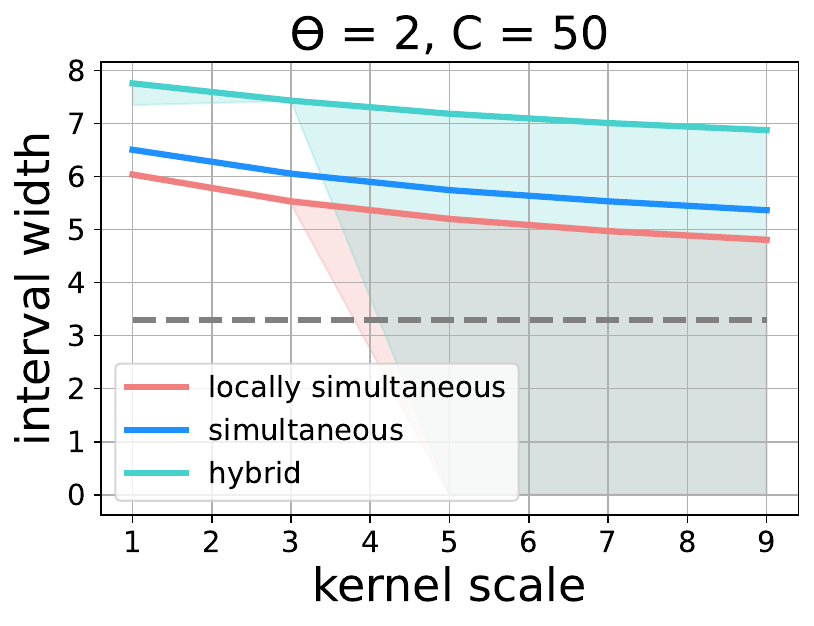}
    \includegraphics[width=0.24\textwidth]{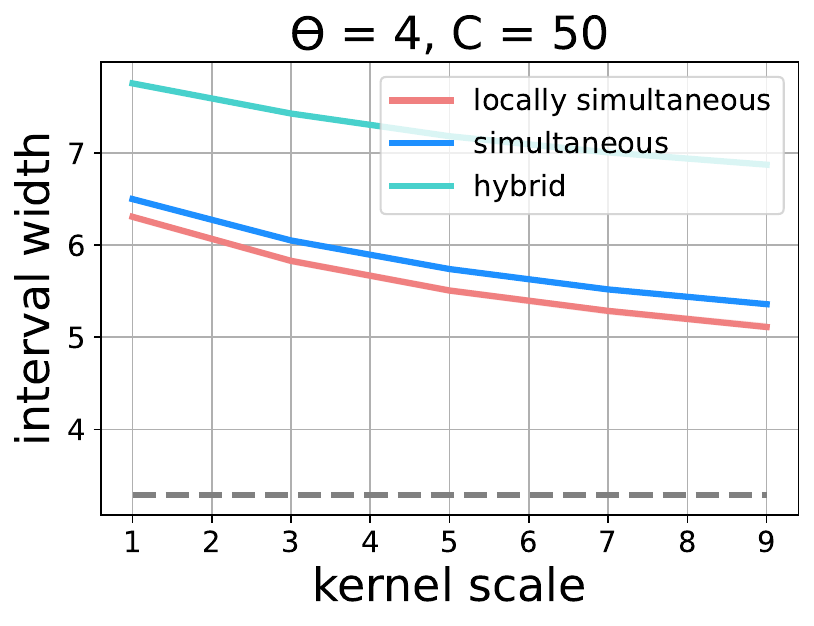}
    \caption{Interval width achieved by locally simultaneous, fully simultaneous, and hybrid inference in the file-drawer problem. Conditional inference achieves much wider intervals and is thus not included in the plots. The dashed line corresponds to nominal interval widths.}
\label{fig:inf_above_threshold}
\end{figure}

\subsection{File-drawer problem}
\label{sec:filedrawer_exps}

The next problem we consider is the file-drawer problem from Section \ref{sec:winner}. As alluded to earlier, the conditional and hybrid approaches provide inference for one real-valued parameter at a time and it is unclear how to generalize them to multi-dimensional problems without resorting to a Bonferroni correction. In contrast, locally simultaneous inference is able to adapt to the dependencies in the data.

To demonstrate this, we let $y = \mu + \xi$, where $\xi \sim \mathcal{N}(0,\Sigma)$ is a Gaussian noise process with the RBF kernel, $\Sigma_{ij} = \exp\left(-\frac{|i-j|^2}{2\phi^2}\right)$; $\phi$ is the key parameter that we vary. As $\phi$ gets larger, the errors become more dependent. We generate $\mu$ as in Section \ref{sec:iow_exp}, again varying $\theta$ and $C$.

First, we observe that the conditional approach is exceptionally fragile in this problem setting: its intervals are consistently much larger than the intervals of the other competitors, often even of a different order of magnitude. We include example plots in Appendix \ref{app:file_drawer_conditional}. For this reason, we omit the conditional approach from the comparison.
In Figure \ref{fig:inf_above_threshold} we plot the interval widths of locally simultaneous, simultaneous, and hybrid inference. We set $T=-1$ and vary the kernel scale $\phi$, as well as $\theta$ and $C$, which control the shape of $\mu$. We combine the hybrid method with a Bonferroni correction over the selected set. We observe that the simultaneous and locally simultaneous methods are indeed able to adapt to the kernel scale. Moreover, as in the previous problem setting, increasing $\theta$ makes the problem more challenging for the hybrid method, and as $C$ increases the problem becomes easier for locally simultaneous inference.

\begin{figure}[t]
    \centering
    \includegraphics[width=0.32\textwidth]{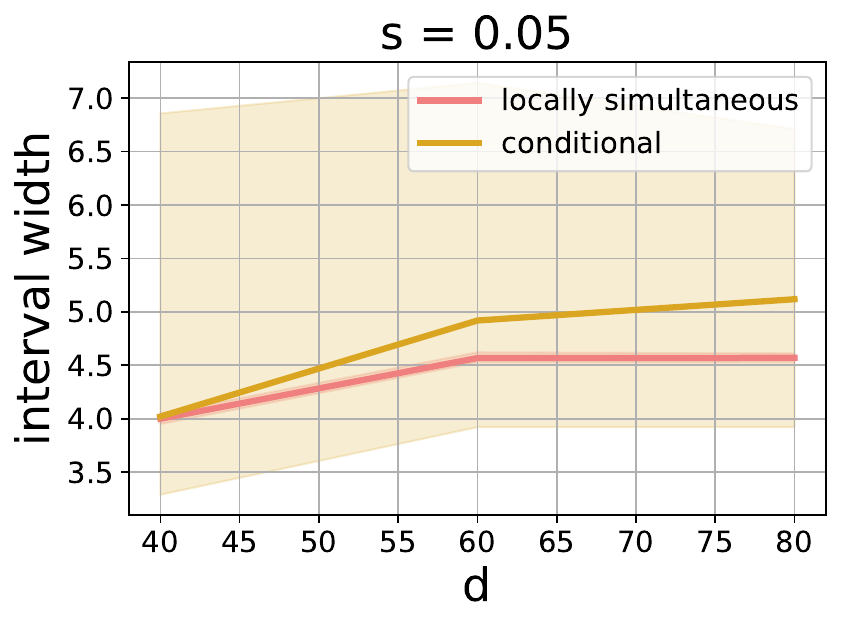}
    \includegraphics[width=0.32\textwidth]{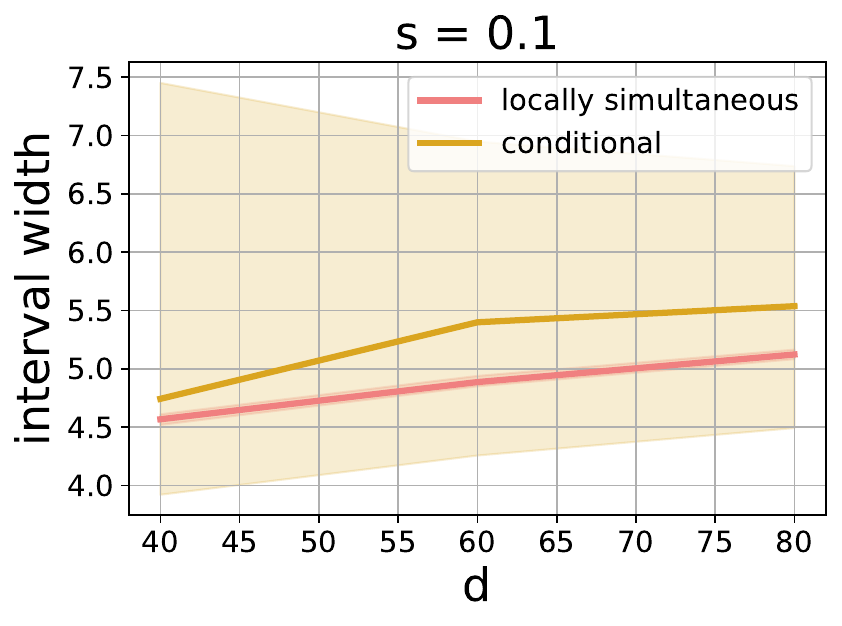}
    \includegraphics[width=0.32\textwidth]{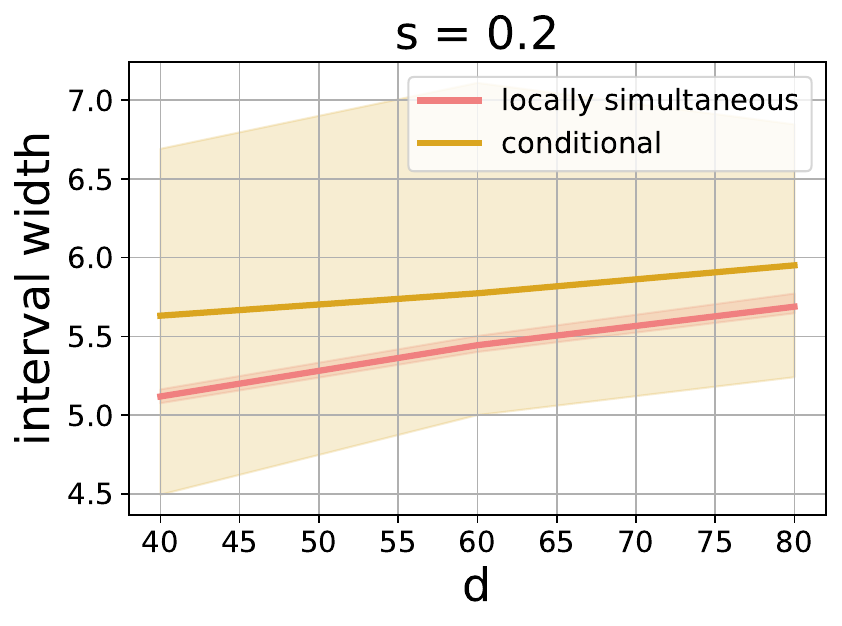}
    \caption{Interval width achieved by locally simultaneous and conditional inference in the problem of inference after selection via the LASSO, when the true underlying signal is $s$-sparse.}
\label{fig:lasso}

\includegraphics[width=0.32\textwidth]{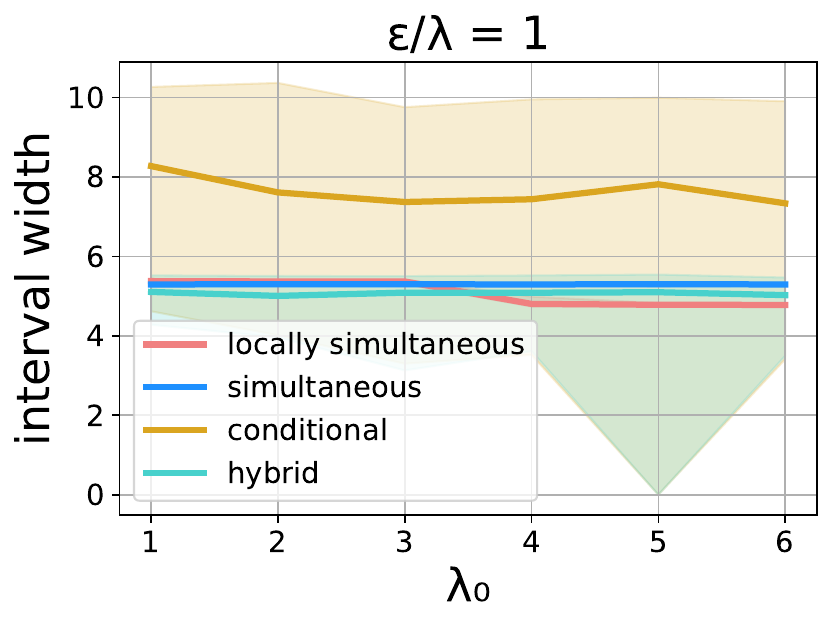}
    \includegraphics[width=0.32\textwidth]{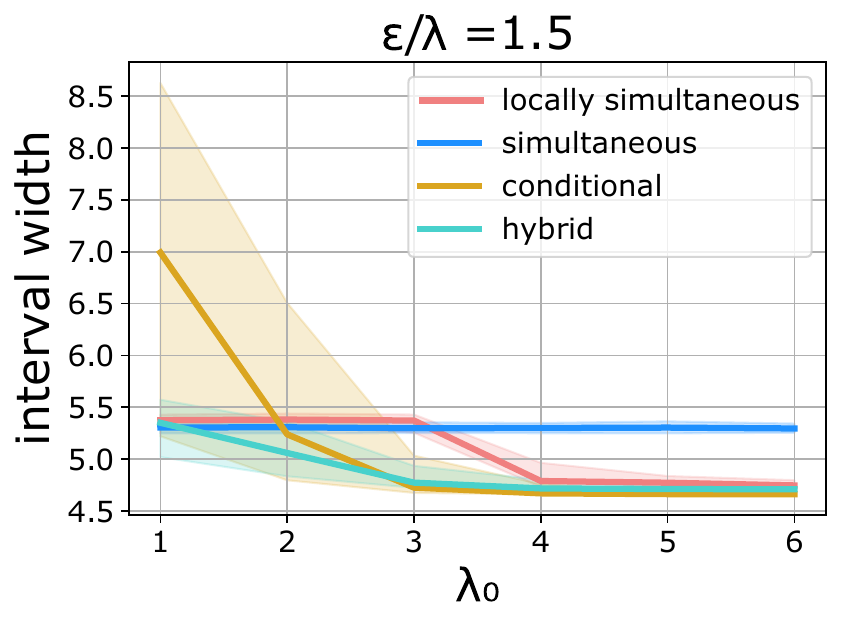}
    \includegraphics[width=0.32\textwidth]{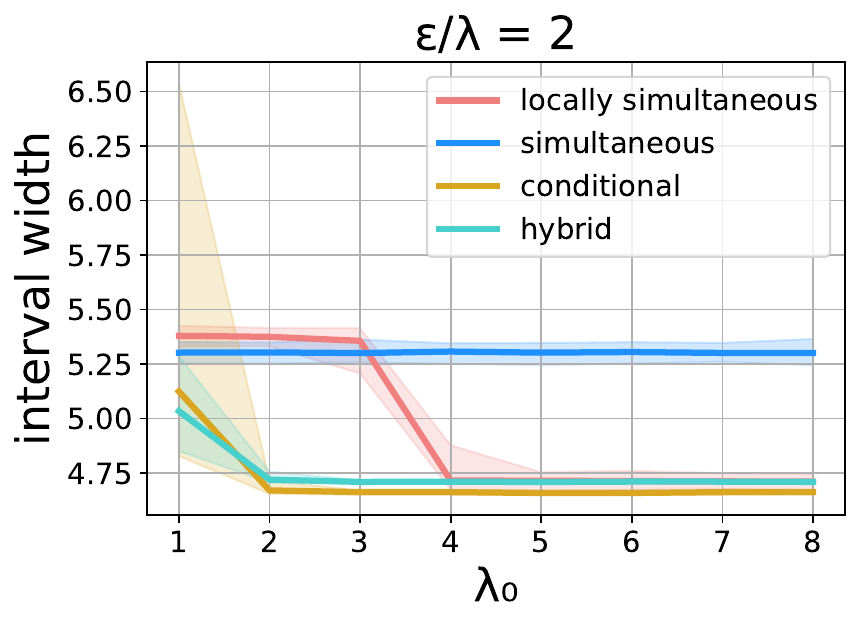}
    \caption{Interval width achieved by locally simultaneous, fully simultaneous, conditional, and hybrid inference in the problem of inference after selection via the LASSO, when we vary to ratio of signal strength to regularization.}
\label{fig:lasso2}
\end{figure}

\subsection{Inference after selection via the LASSO}

Next, we look at the problem of inference after model selection via the LASSO. As the baselines, we consider the simultaneous PoSI method of Berk et al.~\cite{berk2013valid}, the conditional method of Lee et al.~\cite{lee2016exact}, and hybrid inference for the LASSO by McCloskey~\cite{mccloskey2020hybrid}. The hyrbid method uses PoSI as a subroutine for narrowing down the conditional confidence intervals.

Already when the dimension $d$ is greater than $20$, the number of models admissible for selection exceeds $10^6$, making the fully simultaneous PoSI method of Berk et al.~\cite{berk2013valid} prohibitively computationally expensive. This in turn means that the hybrid method likewise becomes prohibitively expensive. Here we show that the set of plausible models $\widehat \M_\nu^+$ can be much smaller than the set of all subsets of $[d]$ when the true data-generating model is sparse, making locally simultaneous inference both powerful and computationally tractable.

In Figure \ref{fig:lasso}, we consider the following data-generating process. We generate the design matrix to have i.i.d. standard normal entries and normalize the columns to have norm $1$. We let $y = X\beta + \xi$, where $\xi\sim\mathcal N(0,I_n)$ and $\beta$ has $\lceil s\cdot d\rceil$ nonzero entries, where we vary the sparsity parameter $s$. Of the $\lceil s\cdot d\rceil$ nonzero entries, we take half of them to be ``weak'', specifically equal to $\lambda$, and half of them to be ``strong'', specifically equal to $2\lambda$. We let $\lambda$ have the usual scaling of $\sim \sqrt{2\log(e\cdot d)}$. In particular, we fix $\lambda = 6\sqrt{2\log(e\cdot d)}$ and $n=1000$. In this parameter regime, we observe that the plausible models are typically those models that always include the strong variables, never include the irrelevant variables, and contain an arbitrary subset of the weak variables. We only compare locally simultaneous inference to conditional inference, seeing that fully simultaneous inference is computationally challenging for the values of $d$ we consider. As before, we observe that the conditional approach exhibits high variability. Moreover, the median interval width implied by the locally simultaneous approach is noticeably smaller.

That being said, the locally simultaneous solution for the LASSO has computational disadvantages. Its complexity scales with the number of plausible model-sign pairs and there can be up to $3^d$ such corresponding pairs, which means that in the worst case the search for all plausible models can be fairly slow. A reasonable remedy is to introduce a parameter $P_{\max}$ such that, if the size of $\P_{\mathrm{todo}}$ in Algorithm \ref{alg:local_lasso} exceeds $P_{\max}$, the search for new model-sign pairs stops and the procedure simply runs the PoSI method of Berk et al. at error level $\alpha-\nu$. We implement this strategy in Figure \ref{fig:lasso2}. Specifically, we let $d = 10$ and generate $X$ and $y$ as before, only now $\beta_i = \epsilon$ for $i\in\{1,\dots,5\}$ and $\beta_i = 0$ for $i\in\{6,\dots,10\}$. We set $\lambda = \lambda_0 \sqrt{2\log(e\cdot d)}$ and vary $\lambda_0$ and $\epsilon/\lambda_0$. When $\epsilon/\lambda_0=1$, the non-nulls of $\beta$ are approximately at the threshold of being selected, and as $\epsilon/\lambda$ grows the true model becomes more obvious. To speed up locally simultaneous inference, we set $P_{\max}=2000$. As in the experiments on inference on the winner, we observe that locally simultaneous inference is preferred to conditional inference when the data is near the selection boundary, which happens when $\lambda_0$ or $\epsilon/\lambda$ is small. As the selection becomes more obvious, conditional inference becomes more powerful. The hybrid method is a favorable solution in this problem, especially seeing that $d$ is small.

\begin{figure}[t]
    \centering
    \includegraphics[height=0.21\textwidth]{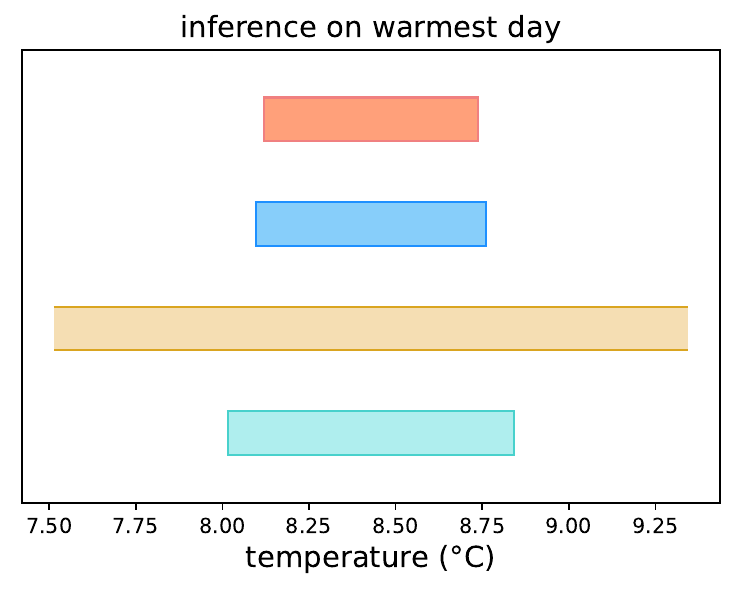}
    \includegraphics[height=0.21\textwidth]{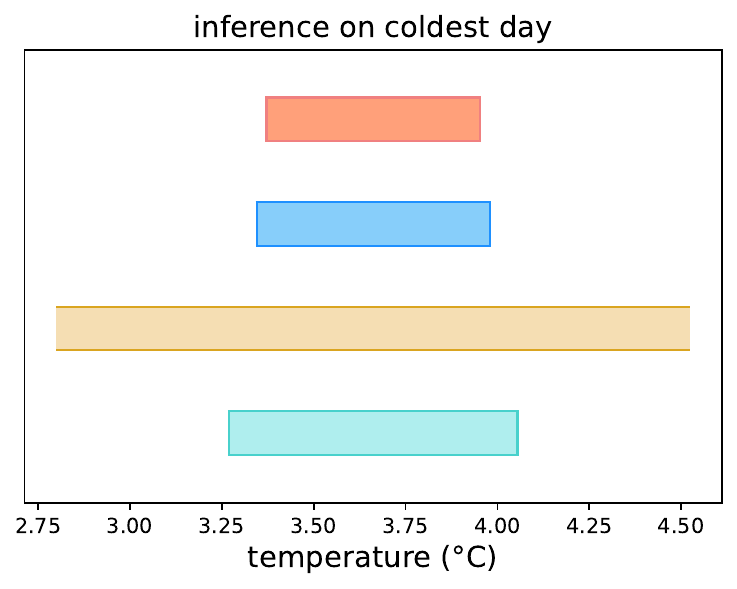}
    \includegraphics[height=0.21\textwidth]{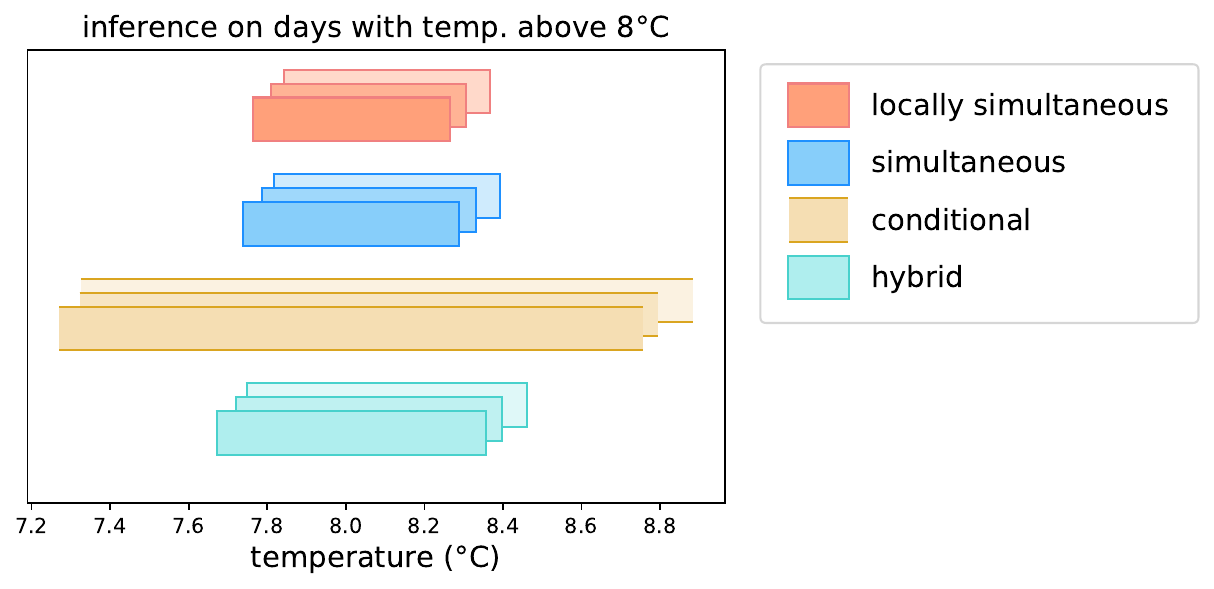}
    \caption{Intervals for the mean temperature constructed via locally simultaneous, fully simultaneous, conditional, and hybrid inference, for selections based on time. The conditional intervals were clipped for the purpose of the visualization.}
\label{fig:climate_time}

\includegraphics[height=0.21\textwidth]{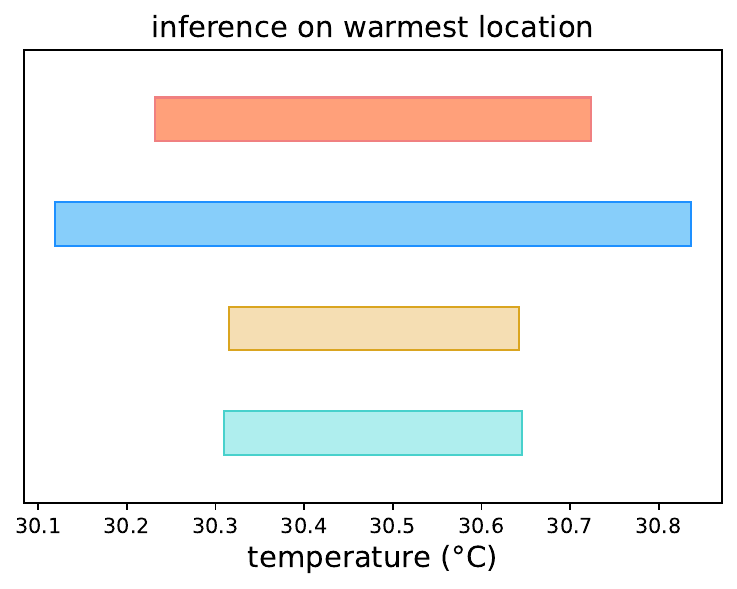}
    \includegraphics[height=0.21\textwidth]{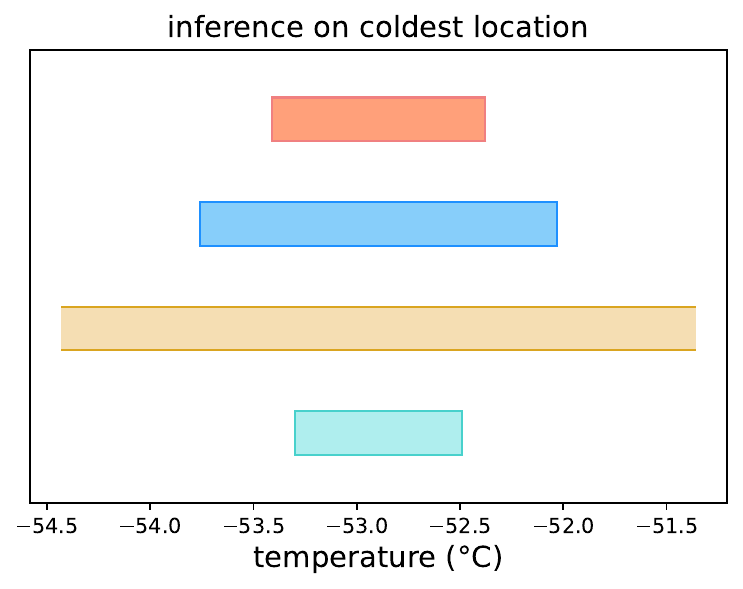}
    \includegraphics[height=0.21\textwidth]{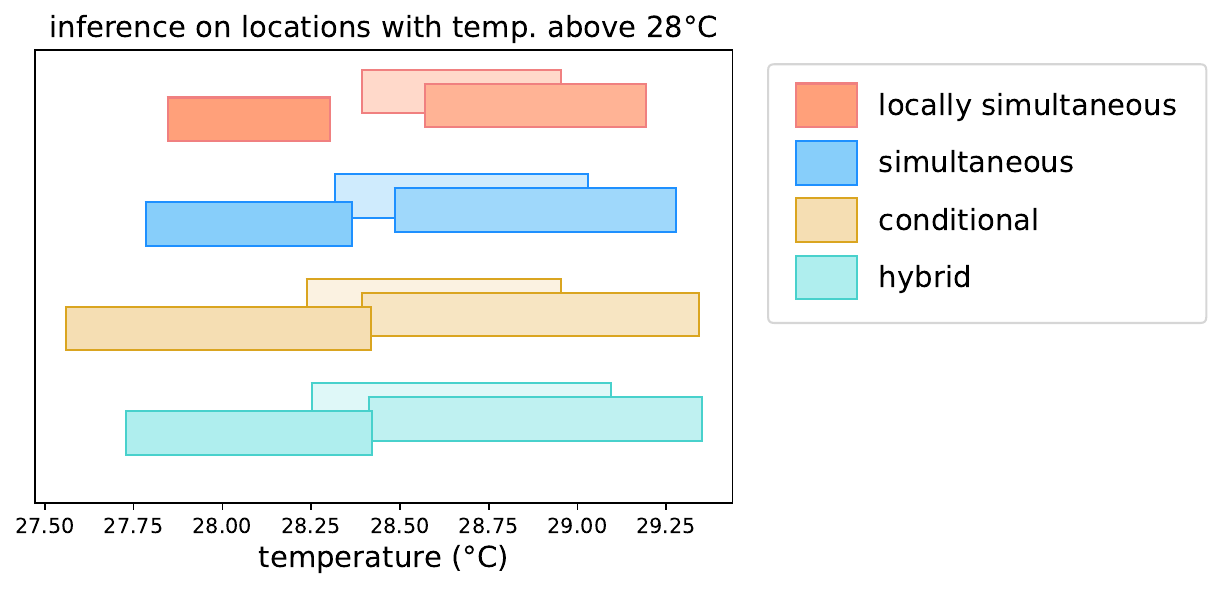}
    \caption{Intervals for the mean temperature constructed via locally simultaneous, fully simultaneous, conditional, and hybrid inference, for selections based on location. The conditional interval in the middle panel was clipped for the purpose of the visualization.}
\label{fig:climate_location}
\end{figure}

\subsection{Experiments on real climate data}

Finally, we conduct experiments on a real climate dataset \cite{rasp2020weatherbench}. The dataset contains hourly measurements of temperature from 1999 to 2018 across a discrete grid of locations on Earth. The grid is obtained by pairing 32 latitude coordinates with 64 longitude coordinates. We model the measurements from the 20 years as i.i.d. draws from an underlying distribution. We again compare locally simultaneous, fully simultaneous, conditional, and hybrid inference. To be able to apply the conditional and hybrid approaches, we model the draws as i.i.d. multivariate Gaussians. We use older data, from 1979 to 1998, to estimate the Gaussian covariance. We study two types of selection: based on time and based on location. 

In the first set of experiments we compute the average temperature on Earth (averaged over all locations on the grid) and look at the resulting time series. For each year we take one measurement per day, evaluated at noon, resulting in a series of 365 entries. We ask for inference on the warmest day, coldest day, and all days with temperature above $8\degree C$. In the second set of experiments we compute the average annual temperature and look at its distribution over the different recorded locations on Earth. Similarly to the first set of experiments, we ask for inference on the warmest location, coldest location, and all locations with temperature above $28 \degree C$. We plot the resulting intervals in Figure~\ref{fig:climate_time} and Figure~\ref{fig:climate_location}, respectively. For the two file-drawer problems, for interpretability we only visualize three intervals, corresponding to the first three days/locations where the temperature exceeds the critical threshold. In the first set of experiments, we observe that conditional inference leads to significantly wider intervals than any of the three alternatives; in the visualizations, the conditional intervals are clipped to fit within the plot. In the first problem of the second set of experiments the winner stands out and hence the conditional method outperforms the other approaches. The locally simultaneous approach gives narrower intervals than the fully simultaneous approach, as expected. The hybrid approach leads to smaller intervals than the locally simultaneous approach in some settings and wider in others. Critically, however, the hybrid approach is only applicable because we imposed an assumption of Gaussianity, while locally simultaneous inference would be applicable even nonparametrically. We note that the hybrid approach is technically applicable nonparametrically, but in the asymptotic regime. This follows by the central limit theorem; see \cite{andrews2019inference}. In this application we only have $20$ observations, so requiring finite-sample validity is a natural choice.

\section{Discussion}

We introduced a correction for selective inference called locally simultaneous inference, which requires taking a simultaneous correction only over selections that seem plausible in hindsight. As such, locally simultaneous inference is a refinement of standard simultaneous inference methods and, in fact, it strictly dominates simultaneous inference under a suitable choice of acceptance region $A_\nu(P)$. Of course, one limitation of locally simultaneous inference relative to classical simultaneous inference is that the former requires knowledge of the selection mechanism, while the latter does not. In this sense, the latter is ``more robust'' to latent data manipulation on the part of the researcher that a consumer of research is unable to see. Similarly, classical simultaneous inference applies to selection mechanisms that are discovered adaptively, while locally simultaneous inference requires specifying the selection strategy ahead of time.

\subsection{Tradeoffs with baselines}

We compared locally simultaneous inference to the conditional method and its hybrid refinement. Conditional and locally simultaneous inference are founded on different principles and neither dominates the other in general. However, there is some intuition about the comparison to be drawn from the experiments in Section \ref{sec:exps}. The conditional approach is very sensitive near the selection boundary: we observed that increasing the number of near-selected candidates drastically increased the intervals. This was demonstrated by varying the shape of the mean outcome near the maximizer, as determined by $\theta$. The hybrid method largely resolves this issue, but if there are many candidate selections, the simultaneous-inference step of hybrid inference may be very conservative. Locally simultaneous inference, on the other hand, is sensitive to a larger number of candidate selections than conditional methods, seeing that the plausible selections include all near-selected candidates and more, but the sensitivity is far more moderate. This less local but more moderate sensitivity was demonstrated by varying the signal range $C$; indeed, we observed that locally simultaneous intervals are more sensitive to increasing $C$ than increasing $\theta$. Therefore, we believe that locally simultaneous inference is competitive in cases where there are many possible selections that could be obtained by slightly perturbing the data. Another setting where we believe locally simultaneous inference is competitive is when the target of inference is multidimensional and there are strong dependencies between the coordinates of the target. Locally simultaneous inference is able to adapt to such dependencies successfully, as demonstrated in Section \ref{sec:filedrawer_exps}.

\subsection{Choice of acceptance region}

One question that remains open is optimal choices for the acceptance region $A_\nu(P)$ in the construction of the plausible set $\widehat\Gamma^+_\nu$. This set crucially determines both the computational tractability of finding the set of plausible selections and the statistical power of downstream inferences. In most of our examples we used the ``$\ell_\infty$-box'' approach due to its simplicity and tractability, however we believe that there are more powerful choices and we leave this direction for future work.

\subsection{Hypothesis testing} While our exposition primarily focused on forming confidence intervals, locally simultaneous inference has implications to multiple testing as well.

Consider the following concrete example. Suppose we want to test the null hypotheses $H_\gamma:\theta_\gamma\leq 0$, $\gamma\in[m]$, at family-wise error level $\alpha$. We have access to independent observations $y_\gamma\sim \mathcal N(\theta_\gamma,1)$, $\gamma\in[m]$.
One baseline is the \v{S}id\` ak correction, which forms a lower confidence bound at level $1-(1-\alpha)^{1/m}$ for all hypotheses and rejects those where zero is excluded. If, say, $m=10$ and $\alpha=0.1$, this means that those hypotheses for which $y_\gamma$ is greater than about $2.3$ would get rejected.
Locally simultaneous inference can be leveraged to focus power on the most promising hypotheses, even allowing for no multiplicity correction. In particular, if one applies the file-drawer technique from Section \ref{sec:winner} with $\nu=0.1\alpha$ and $T=\Phi^{-1}(\alpha-\nu) = 1.34$, then all observations that get selected will have a one-sided p-value at most $\alpha - \nu$. This is turn implies that if a single observation crosses the threshold and all other observations are sufficiently smaller, the hypothesis corresponding to the largest observation can automatically be rejected at threshold $T=1.34$. This is a substantial improvement compared to the \v{S}id\` ak threshold.
Exploring the implications of locally simultaneous inference to hypothesis testing is a valuable future direction.

\subsection{Choosing \texorpdfstring{$\nu$}{nu}} In our experiments we heuristically fixed $\nu = 0.1\alpha$ throughout. When utilizing the stronger result of Theorem~\ref{thm:refined} is challenging due to the requirement of performing fully simultaneous inference, as in the case of the LASSO, this heuristic seems reasonable as in the worst case the analyst will perform fully simultaneous inference at error level $0.9\alpha$. To give further insight into the consequences of different choices of $\nu$, in Appendix \ref{app:choosing_nu} we show simulation results for varying $\nu$. Generally speaking, when $\nu$ is small fewer candidate selections get filtered out but inference on the remaining candidates is more powerful. When $\nu$ is large there is a greater possibility of reducing the set of plausible candidates, but inference on the remaining candidates is less powerful.
The right choice of $\nu$ thus depends on the configuration of alternative selections, especially how close they are to the winning one.

\section*{Acknowledgements}

We thank Anastasios Angelopoulos for comments on a draft of this manuscript. We also thank the anonymous reviewers and associate editor for exceptionally detailed feedback that has helped improve the clarity and quality of our paper, as well as pointers to relevant prior work we had initially omitted.

\bibliography{refs}
\bibliographystyle{plainnat}

\appendix

\newpage

\section{Deferred proofs}

\subsection{Proof of Theorem \ref{thm:gauss_seq_winner}}

The proof essentially follows by applying Theorem \ref{thm:refined}. The bulk of the proof applies to both the problem of inference on the winner and the file-drawer problem. Towards the end we specialize the analysis to the individual problems.

We let
\[A_\nu(\mu) = \left\{y: \|y-\mu\|_\infty \leq q^{\nu}([m])\right\} .\]
The validity of this set follows directly by the definition of $q^{\nu}([m])$:
$$P_\mu\left\{\max_{i\in[m]}|y_i-\mu_i| \leq q^{\nu}([m])\right\} = P_0\left\{\max_{i\in[m]}|Z_i| \leq q^{\nu}([m]) \right\} \geq 1-\alpha.$$

Using the above choice of $A_\nu(\mu)$ we can write
$$\Gamma_\nu(\mu) = \cup_{y:\|y-\mu\|_\infty \leq q^{\nu}([m])} \widehat \Gamma(y),$$
and thus
\begin{align*}
\widehat\Gamma_\nu^+ &= \cup_{\mu:\|y-\mu\|_\infty \leq q^{\nu}([m]) } \Gamma_\nu(\mu)\\
&= \cup_{\mu:\|y-\mu\|_\infty \leq q^{\nu}([m]) } \cup_{y':\|y'-\mu\|_\infty \leq q^{\nu}([m])} \widehat\Gamma(y')\\
&= \cup_{y':\|y'-y\|_\infty \leq 2q^{\nu}([m]) } \widehat\Gamma(y').
\end{align*}

In words, the set $\widehat\Gamma_\nu^+$ is the set of all selections obtained by perturbing the entries in $y$ by at most $2q^{\nu}([m])$ in $\ell_\infty$-norm.

Now we specialize the analysis to the two selection problems.

In the problem of inference on the winner, the selected inferential target is indexed by $\widehat\Gamma(y) = \left\{\hat \gamma (y)\right\}$. The most favorable perturbation $y'$ for an index $j\in[m]$ to be selected is obtained by taking $y'_j = y_j + 2q^{\nu}([m])$ and $y'_k = y_k - 2q^{\nu}([m])$ for $k\neq j$; therefore, $\widehat \Gamma_\nu^+ = \{\gamma\in[m]: y_\gamma \geq y_{\hat \gamma} - 4q^\nu([m])\}$ is the set of plausible selections.

In the file-drawer problem, the selected targets are $\widehat\Gamma(y)  = \{\gamma\in[m]:y_\gamma \geq T\}$. The indices that could fall in this set given a $2q^{\nu}([m])$ perturbation around $y$ are $\widehat \Gamma_\nu^+ = \{\gamma\in[m]: y_\gamma \geq T - 2q^\nu([m])\}$.

Therefore, in both cases we have identified $\widehat\Gamma^+_\nu$. The final statement follows by applying Theorem \ref{thm:refined}.

\subsection{Proof of Theorem \ref{thm:nonparametric_winner}}

By following virtually the same argument as in Theorem \ref{thm:gauss_seq_winner}, we can conclude that $\widehat \Gamma_\nu^+$ is the set of all plausible selections, for both the problem of inference on the winner and the file-drawer problem. The final statement follows by applying Theorem \ref{thm:local_posi_main}, together with a Bonferroni correction over $\widehat \Gamma_\nu^+$.

\subsection{Proof of Corollary \ref{cor:linear_reg}}

We argue that $\widehat \V_\nu^+$ is the set of plausible targets $\widehat \Gamma_\nu^+$ when the acceptance region is chosen as
$$A_\nu(\mu) = \left\{y:\|X^\top y - X^\top \mu\|_\infty \leq q^{\nu}(\{X_j\}_{j=1}^d) \right\}.$$
After this is established, the result follows directly from Theorem \ref{thm:local_posi_main}.

First, the validity of $A_\nu(\mu)$ follows by the definition of $q^{\nu}$:
$$P_\mu\left\{\|X^\top y - X^\top \mu\|_\infty > q^{\nu}(\{X_j\}_{j=1}^d) \right\} = P_0\left\{\max_{j\in[d]} |X_j^\top Z| > q^{\nu}(\{X_j\}_{j=1}^d) \right\} \leq \nu,$$
where $Z\sim P_0$. Therefore, $P_\mu\{y\in A_\nu(\mu)\} \geq 1-\nu$.

We can write
$$\Gamma_\nu(\mu) = \cup_{y:\|X^\top y-X^\top \mu\|_\infty \leq q^{\nu}(\{X_j\}_{j=1}^d) } \widehat \Gamma(y),$$
and thus
\begin{align*}
\widehat\Gamma_\nu^+ &= \cup_{\mu:\|X^\top y-X^\top \mu\|_\infty \leq q^{\nu}(\{X_j\}_{j=1}^d) } \Gamma_\nu(\mu)\\
&= \cup_{\mu:\|X^\top y-X^\top \mu\|_\infty \leq q^{\nu}(\{X_j\}_{j=1}^d) } \cup_{y':\|X^\top y'-X^\top \mu\|_\infty \leq q^{\nu}(\{X_j\}_{j=1}^d) } \widehat\Gamma(y')\\
&= \cup_{y':\|X^\top y-X^\top y'\|_\infty \leq 2q^{\nu}(\{X_j\}_{j=1}^d) } \widehat\Gamma(y').
\end{align*}
Since $\widehat\Gamma(y) = \left\{(j,\hat M(y)) : j\in\hat M(y)\right\}$, we finally have
\begin{align*}
\widehat\Gamma_\nu^+ &= \cup_{y':\|X^\top y-X^\top y'\|_\infty \leq 2q^{\nu}(\{X_j\}_{j=1}^d) } \left\{(j,\hat M(y')) : j\in\hat M(y')\right\}\\
&= \left\{(j,M) : j\in M, M \in \widehat\M_\nu^+\right\}.
\end{align*}
Therefore, by Theorem \ref{thm:local_posi_main} it suffices to take a simultaneous correction over $\widehat \Gamma_\nu^+$. The set $\widehat\V^+$ is the set of contrasts that ensures simultaneously valid inference for $\{\theta_{j\cdot M}\}_{M\in\widehat\M_\nu^+}$ (see, e.g., Theorem 4.1 in~\cite{berk2013valid}).

\subsection{Proof of Lemma \ref{lemma:safe_Mc}}

First, we argue that the condition $|X_j^\top (y - X_M\beta_{(M,s)}(y))| < \lambda - s_\nu (1 + \|X_j^\top X_M(X_M^\top X_M)^{-1}\|_1)$ is equivalent to
$$\max_{y'\in \B^\infty_\nu} |X_j^\top (y' - X_M\beta_{(M,s)}(y'))| < \lambda.$$
This follows because
\begin{align*}
&\max_{y'\in \B^\infty_\nu} |X_j^\top (y' - X_M\beta_{(M,s)}(y'))|\\
&= \max_{y'\in \B^\infty_\nu} \left|X_j^\top (y - X_M\beta_{(M,s)}(y)) + X_j^\top \left(y' - y - X_M\beta_{(M,s)}(y') + X_M\beta_{(M,s)}(y)\right)\right|\\
&= \max_{y'\in \B^\infty_\nu} \left|X_j^\top (y - X_M\beta_{(M,s)}(y)) + X_j^\top \left(I - X_M(X_M^\top X_M)^{-1}X_M^\top\right)(y' - y)\right|.
\end{align*}
Now notice that by the definition of $\B^\infty_\nu$ we have
\begin{align*}
\max_{y'\in \B^\infty_\nu} X_j^\top (I - X_M(X_M^\top X_M)^{-1}X_M^\top)(y' - y) &= \max_{|z_j|\leq s_\nu, \|z_M\|_\infty \leq s_\nu} z_j - X_j^\top X_M(X_M^\top X_M)^{-1} z_M\\
&= s_\nu + s_\nu \|X_j^\top X_M(X_M^\top X_M)^{-1}\|_1,
\end{align*}
which follows by the duality between the $\ell_1$- and $\ell_\infty$-norms. Similarly we have
$$\min_{y'\in \B^\infty_\nu} X_j^\top (I - X_M(X_M^\top X_M)^{-1}X_M^\top)(y' - y) = -s_\nu - s_\nu \|X_j^\top X_M(X_M^\top X_M)^{-1}\|_1.$$
Thus, we see that
\begin{align*}
&\max_{y'\in \B^\infty_\nu} \left|X_j^\top (y - X_M\beta_{(M,s)}(y)) + X_j^\top (I - X_M(X_M^\top X_M)^{-1}X_M^\top)(y' - y)\right|\\
&\quad = \left|X_j^\top (y - X_M\beta_{(M,s)}(y))\right| + s_\nu\left(1 + \|X_j^\top X_M(X_M^\top X_M)^{-1}\|_1\right)
\end{align*}
Putting everything together, we have shown that
\begin{align*}
\max_{y'\in \B^\infty_\nu} |X_j^\top (y' - X_M\beta_{(M,s)}(y'))| = \left|X_j^\top (y - X_M\beta_{(M,s)}(y))\right| + s_\nu\left(1 + \|X_j^\top X_M(X_M^\top X_M)^{-1}\|_1\right).
\end{align*}
Therefore, the screening rule in Lemma \ref{lemma:safe_Mc} is equivalent to
\begin{equation}
\label{eq:safe_Mc}
\max_{y'\in \B^\infty_\nu} \left|X_j^\top (y' - X_M\beta_{(M,s)}(y'))\right| < \lambda.
\end{equation}

We argue that the condition in Eq.~\eqref{eq:safe_Mc} implies that there cannot exist a pair $(M',s')\in\B(M,s)$ such that $M' = M\cup \{j\}$. Indeed, if this were true, then there must exist a point $y'\in\B_\nu^\infty$ on the boundary between the two corresponding polyhedra. Given the polyhedral characterization of \citet{lee2016exact}, this point must satisfy
$$X_j^\top (I-X_M(X_M^\top X_M)^{-1}X_M^\top )y' = \lambda (1-X_j^\top(X_M^\top)^+s).$$
By rearranging, we see that this equality is equivalent to 
\begin{equation}
\label{eq:boundary_condition}
X_j^\top \left(y' - X_M(X_M^\top X_M)^{-1}(X_M^\top y' - \lambda s)\right) = \lambda.
\end{equation}
The left-hand side is equal to $X_j^\top (y' - X_M\beta_{(M,s)}(y'))$; therefore, condition \eqref{eq:boundary_condition} contradicts condition \eqref{eq:safe_Mc}, and thus we can conclude that $M\cup \{j\}$ cannot be the model of any neighboring model-sign pair $(M',s')\in\B(M,s)$.

\subsection{Proof of Lemma \ref{lemma:safe_M}}

The proof proceeds similarly to the proof of Lemma \ref{lemma:safe_Mc}. Fix $j\in M$. First, we argue that the condition $|\beta_{j\cdot (M,s)}(y)| > s_\nu \|e_{j\cdot (M,s)}^\top (X_M^\top X_M)^{-1}\|_1$ is equivalent to
\begin{equation}
\label{eq:safe_M_equiv}
\min_{y'\in\B^\infty_{\nu}} |\beta_{j\cdot (M,s)}(y')| >0.
\end{equation}
This follows by writing
\begin{align*}
\min_{y'\in\B^\infty_{\nu}} |\beta_{j\cdot (M,s)}(y')| &= \min_{y'\in\B^\infty_{\nu}} |\beta_{j\cdot (M,s)}(y) + \beta_{j\cdot (M,s)}(y') - \beta_{j\cdot (M,s)}(y)|\\
&= \min_{y'\in\B^\infty_{\nu}} |\beta_{j\cdot (M,s)}(y) + e_{j\cdot (M,s)}^\top (X_M^\top X_M)^{-1}X_M^\top (y' - y)|.
\end{align*}
By the definition of $\B^\infty_\nu$, we can write
\begin{align*}
\max_{y'\in\B^\infty_{\nu}} e_{j\cdot (M,s)}^\top (X_M^\top X_M)^{-1}X_M^\top (y' - y) = \max_{\|z_M\|_\infty \leq s_\nu} e_{j\cdot (M,s)}^\top (X_M^\top X_M)^{-1} z_M = s_\nu \|e_{j\cdot (M,s)}^\top (X_M^\top X_M)^{-1}\|_1,
\end{align*}
and similarly $\min_{y'\in\B^\infty_{\nu}} e_{j\cdot (M,s)}^\top (X_M^\top X_M)^{-1}X_M^\top (y' - y) = - s_\nu \| e_{j\cdot (M,s)}^\top(X_M^\top X_M)^{-1}\|_1$. Putting everything together, we see that $\min_{y'\in\B^\infty_{\nu}} |\beta_{j\cdot (M,s)}(y')| >0$ implies $s_\nu \| e_{j\cdot (M,s)}^\top(X_M^\top X_M)^{-1}\|_1 < |\beta_{j\cdot (M,s)}(y)|$, and vice versa.

Now we argue that condition \eqref{eq:safe_M_equiv} implies that variable $j$ cannot exit the model in any of the neighboring polyhedra within $\B^\infty_\nu$. If it can, then the \citet{lee2016exact} characterization implies that there exists a point $y'\in\B^\infty_\nu$ on the boundary between the respective polyhedra such that $e_{j\cdot(M,s)}^\top (X_M^\top X_M)^{-1} y' = \lambda e_{j\cdot(M,s)}^\top (X_M^\top X_M)^{-1}s$. By rearranging, we can rewrite this equality as $e_{j\cdot (M,s)}^\top \beta_{j\cdot (M,s)}(y') = 0$, which contradicts condition \eqref{eq:safe_M_equiv}.

\subsection{Proof of Theorem \ref{thm:risk_minimizer}}

Let
$$A_\nu(P) = \left\{\D: | R_n(f,\D) - R(f,P)| \leq \texttt{Gap}_n(\F) + \sqrt{\frac{2}{n} \log\left(\frac{1}{\nu}\right)},~\forall f\in\F \right\}.$$
By the bounded differences inequality, we know
$$P\{\D \in A_\nu(P)\} \geq 1-\nu.$$
To simplify notation, in the rest of the proof we let $G = \texttt{Gap}_n(\F) + \sqrt{\frac{2}{n} \log\left(\frac{1}{\nu}\right)}$. 

The selected target is $ \hat f \equiv \hat f(\D) = \argmin_{f\in\F}  R_n(f,\D)$ and thus $\Gamma_\nu(P) = \cup_{\D:\D\in A_\nu(P)} \hat f(\D)$. By a similar argument as in Theorem \ref{thm:gauss_seq_winner}, we have
\begin{align*}
    \widehat \Gamma_\nu^+ &= \cup_{P:\D\in A_\nu(P) } \Gamma_\nu(P)\\
    &= \cup_{P:\D\in A_\nu(P) } \cup_{\D':\D'\in A_\nu(P)} \hat f(\D')\\
    &= \cup_{P:\max_{f\in\F}| R_n(f,\D) - R(f,P)| \leq G} \cup_{\D':\max_{f\in\F}| R_n(f,\D') - R(f,P)| \leq G} \hat f(\D')\\
    &\subseteq \cup_{\D':\max_{f\in\F}| R_n(f,\D') -  R_n(f,\D)| \leq 2G} \hat f(\D').
\end{align*}
Therefore, we only need to consider empirical risk minimizers on datasets $\D'$ such that $\max_{f\in\F}| R_n(f,\D') -  R_n(f,\D)| \leq 2G$. The set of those risk minimizers is contained in the set of hypotheses with empirical risk within $4G$ of $\min_{f\in\F}  R_n(f,\D)$, i.e.
\begin{align*}
\widehat \Gamma_\nu^+ &\subseteq \{f\in\F:  R_n(f,\D) \leq  R_n(\hat f,\D) + 4G\}\\
&= \left\{f\in\F:  R_n(f,\D) \leq  R_n(\hat f,\D) + 4\texttt{Gap}_n(\F) + 4\sqrt{\frac{2}{n} \log\left(\frac{1}{\nu}\right)}\right\}.
\end{align*}
Finally, invoking the bounded differences inequality together with Theorem \ref{thm:local_posi_main} completes the proof.

\section{Deferred analyses and results}

\subsection{Example in Figure \ref{fig:intro_comparison}}
\label{app:figure1}

In Figure \ref{fig:intro_comparison}, we plot the quantiles of interval width computed over $100$ trials.

We use the polyhedral method of \citet{lee2016exact} to construct the conditional confidence intervals. The simultaneous confidence intervals are constructed as usual: if $q^\alpha_{\mathrm{sim}}$ is the $1-\alpha$ quantile of $\max_{i\in\{1,2\}}|Z_i|$, where $Z_1,Z_2\stackrel{\mathrm{i.i.d.}}{\sim} \mathcal{N}(0,1)$, then the selective intervals are $(y_{\hat \gamma} \pm q^\alpha_{\mathrm{sim}})$. In what follows we derive locally simultaneous intervals. We set $\alpha=0.1$ and $\nu = 0.05\alpha$.

Let $\Delta_\mu = \mu_2 - \mu_1$ and $A_\nu(\mu) = \{(y_1,y_2): |y_2 - y_1 - \Delta_\mu| \leq \sqrt{2}q^{\nu} \}$, where $q^{\nu}$ is the $1-\nu$ quantile of $|Z|\sim\mathcal N(0,1)$. First we show that $A_\nu(\mu)$ is a valid acceptance region:
\begin{align*}
    P_\mu\{y \in A_\nu(\mu)\} &= P_\mu\left\{|y_2 - y_1 - \Delta_\mu| \leq \sqrt{2}q^{\nu} \right\}\\
    &= P_0\left\{|Z_2 - Z_1| \leq \sqrt{2}q^{\nu}\right\}\\
    &= 1-\nu,
\end{align*}
where $Z_1,Z_2\stackrel{\mathrm{i.i.d.}}{\sim} \mathcal N(0,1)$. Therefore, $A_\nu(\mu)$ is valid. Now it remains to find $\widehat\Gamma^+_\nu$:
\begin{align*}
\widehat\Gamma_\nu^+ &= \cup_{\mu:|y_2-y_1 - \Delta_\mu|\leq \sqrt{2}q^{\nu}} \Gamma_\nu(\mu)\\
&= \cup_{\mu:|y_2-y_1 - \Delta_\mu|\leq \sqrt{2}q^{\nu}} \cup_{y':|y_2'-y_1' - \Delta_\mu|\leq \sqrt{2}q^{\nu}} \hat \gamma(y')\\
&= \cup_{y':|y_2'-y_1' - y_2 - y_1|\leq 2\sqrt{2}q^{\nu}} \hat \gamma(y').
\end{align*}
If $|y_2 - y_1|\leq 2\sqrt{2}q^{\nu}$, then $\widehat\Gamma^+_\nu = \{1,2\}$; otherwise, the selection is ``obvious'' and $\widehat\Gamma_\nu^+ = \{\hat \gamma\}$. Finally, we construct the locally simultaneous intervals as $(y_{\hat \gamma} \pm \min\{q_{\mathrm{sim}}^{\alpha-\nu}(|\widehat\Gamma^+_\nu|), q_{\mathrm{sim}}^{\alpha} \})$, where $q_{\mathrm{sim}}^{\alpha-\nu}(k)$ is the $1-(\alpha-\nu)$ quantile of $\max_{i\in [k]} |Z_i|$, $Z_i\stackrel{\mathrm{i.i.d.}}{\sim}\mathcal N(0,1)$. The validity of this choice follows by Lemma \ref{lemma:local_posi_tighter}. In particular, $\tilde C_{\gamma\cdot\Gamma'} = (y_{\gamma} \pm \min\{q_{\mathrm{sim}}^{\alpha-\nu}(|\Gamma'|), q_{\mathrm{sim}}^{\alpha} \})$ satisfies Eq.~\eqref{eq:C-til} by the same argument as in Theorem \ref{thm:refined}, by observing that $\{y\in A_\nu(\mu)\}$ is implied by $\|y-\mu\|_\infty \leq q^\alpha_{\mathrm{sim}}$ for the considered values of $\alpha$ and $\nu$.



\subsection{Example \ref{ex:chi}}
\label{app:chi_section}

We use a slight refinement of Lemma \ref{lemma:local_posi_tighter} that can sometimes lead to a tighter correction at no cost under additional assumptions on the interval construction. As in Theorem \ref{thm:refined}, we assume that the confidence intervals are centered around an estimator $\hat\theta_\gamma$ and have width proportional to a standard error term $\hat \sigma_\gamma$: intervals taken simultaneously over a set of targets $\Gamma$ are of the form $C_{\gamma\cdot \Gamma} = \left(\hat\theta_\gamma \pm q_{\Gamma}^\alpha \hat \sigma_\gamma\right)$, for some appropriately chosen width $q_{\Gamma}^\alpha$.

We note that, for almost all choices of $A_\nu(P)$ used in the paper, the refined correction just reduces to the correction of Theorem \ref{thm:local_posi_main}, although it will prove useful in the context of Example \ref{ex:chi}.

\begin{proposition}
\label{prop:local_posi_symmetric}
Fix $\alpha\in(0,1)$ and $\nu,\tau$ such that $\nu + \tau\in(0,\alpha)$. Let $\hat q$ be any value such that
$$P\left\{\hat q \geq q_{\Gamma_\nu(P)}^{(\alpha-\nu - \tau)}, y \in A_\nu(P)\right\} \geq 1-\tau.$$
For example, $\hat q = \sup_{P'\in B_\nu(y)} q_{\Gamma_\nu(P')}^{(\alpha-\nu)}$ satisfies the above with $\tau = 0$.
Then, it holds that
\[P\left\{\theta_\gamma \in \left(\hat \theta_\gamma \pm \hat q \hat\sigma_\gamma \right), \forall \gamma \in \widehat\Gamma\right\} \geq 1-\alpha.\]
\end{proposition}

\begin{proof}

Following the proof of Lemma \ref{lemma:local_posi_tighter} and specializing it to the symmetric-interval setting gives
\begin{equation}
\label{eq:miscoverage_symmetric}
P\left\{\exists \gamma \in \widehat\Gamma : \theta_\gamma \not\in \left(\hat\theta_\gamma \pm \hat q \hat\sigma_\gamma \right) \right\}  \leq P\left\{\exists \gamma \in \widehat\Gamma : \theta_\gamma \not\in \left(\hat\theta_\gamma \pm \hat q \hat\sigma_\gamma \right), y\in A_\nu(P) \right\} + \nu.
\end{equation}
On the event $\{y\in A_\nu(P)\} = \{P\in B_\nu(y)\}$, we have $\Gamma_\nu(P)\supseteq \widehat\Gamma$. Therefore,
\begin{align*}
&P\left\{\exists \gamma \in \widehat\Gamma : \theta_\gamma \not\in \left(\hat\theta_\gamma \pm \hat q \hat\sigma_\gamma \right), y\in A_\nu(P) \right\}\\ &\leq P\left\{\exists \gamma \in \widehat\Gamma : \theta_\gamma \not\in \left(\hat\theta_\gamma \pm q_{\Gamma_\nu(P)}^{(\alpha-\nu - \tau)}\hat\sigma_\gamma \right), \hat q \geq q_{\Gamma_\nu(P)}^{(\alpha-\nu - \tau)}, y\in A_\nu(P) \right\} + \tau\\
& \leq P\left\{\exists \gamma \in \Gamma_\nu(P) : \theta_\gamma \not\in \left(\hat\theta_\gamma \pm q_{\Gamma_\nu(P)}^{(\alpha-\nu - \tau)}\hat\sigma_\gamma \right), y\in A_\nu(P) \right\} + \tau\\
& \leq P\left\{\exists \gamma \in \Gamma_\nu(P) : \theta_\gamma \not\in \left(\hat\theta_\gamma \pm q_{\Gamma_\nu(P)}^{(\alpha-\nu - \tau)}\hat\sigma_\gamma \right) \right\} + \tau.
\end{align*}
By the simultaneous validity of the confidence intervals, the final probability is at most $\alpha-\nu$. Going back to Eq.~\eqref{eq:miscoverage_symmetric}, we thus get
$P\left\{\exists \gamma \in \widehat\Gamma : \theta_\gamma \not\in \left(\hat\theta_\gamma \pm \hat q \hat\sigma_\gamma \right) \right\} \leq \alpha$,
as desired.
\end{proof}

We now derive $\delta_d(\|y\|)$ used in Example \ref{ex:chi}. We observe $y\sim P_\mu = \mathcal{N}(\mu,I_d)$ and want to do inference on $\theta_{\hat\gamma}(\mu) = \hat\gamma^\top \mu$, where $\hat\gamma = \frac{y}{\|y\|}$. Here and in the rest of this section $\|\cdot\|$ denotes the $\ell_2$-norm.

Define
$$q_\angle^{\alpha}(\delta) = \inf\left\{q: P_0\left\{\sup_{\gamma\in \mathcal{S}^{d-1}: \angle (e_1,\gamma) \leq \delta} |\gamma^\top Z|\geq q \right\} \leq \alpha\right\},$$
where $Z\sim \cN(0,I_d)$ and $e_1$ is a canonical vector. In words, $q_\angle^{(\alpha)}(\delta)$ is the required simultaneous correction for all contrasts within angle $\delta$ of $e_1$. Note that by the rotational invariance of $Z$, $e_1$ is taken without loss of generality; any fixed vector would give the same $q^\alpha_\angle(\delta)$. We let $F_{\mathrm{nc}\chi^2_d(c)}(\cdot)$ and $F_{\mathrm{nct}_{d-1}(c)}(\cdot)$ denote the CDFs of the noncentral $\chi_d^2$-distribution with noncentrality parameter $c$ and the noncentral $t$-distribution with $d-1$ degrees of freedom and noncentrality parameter $c$, respectively.

\begin{proposition}
\label{prop:chi-squared}
Fix $\alpha\in(0,1)$ and $\nu\in(0,\alpha)$. Then,
$$P_\mu\left\{\hat\gamma^\top \mu \in \left(\|y\| \pm q^{(\alpha-\nu)}_\angle\left(\delta_d(\|y\|)\right)\right)\right\} \geq 1- \alpha,$$
where
\begin{align*}
\delta_d(\|y\|) &= 2\arccos s_{\nu/2}\left( \mu_{\nu/2}^{\mathrm{LB}}\right);\\
\mu^{\mathrm{LB}}_\tau &= \sqrt{\sup\left\{c: F_{\mathrm{nc}\chi^2_d(c)}(\|y\|^2) \geq 1-\tau \right\}};\\
s_{\tau}(c) &= \sqrt{\frac{1}{(d-1)/(F_{\mathrm{nct}_{d-1}(c)}^{-1}(\tau))^2 + 1}} \cdot \mathrm{sgn}\left(F_{\mathrm{nct}_{d-1}(c)}^{-1}(\tau)\right),
\end{align*}
for all $\tau\in(0,1)$.
\end{proposition}

\begin{proof}
We can write
\begin{equation}
\label{eq:yprodmu}
\frac{y^\top \mu}{\|y\| \|\mu\|} = \frac{y^\top \frac{\mu}{\|\mu\|}}{\sqrt{(y^\top \frac{\mu}{\|\mu\|})^2 + (d-1)S^2 }},
\end{equation}
where we denote $(d-1)S^2 = \|y\|^2 - (y^\top \frac{\mu}{\|\mu\|})^2$. By a standard argument,
we have $(d-1) S^2 \sim \chi^2_{d-1}$ and $y^\top \frac{\mu}{\|\mu\|}\perp S^2$. If we therefore let $T = \frac{y^\top \frac{\mu}{\|\mu\|}}{S}$, we have that $T$ follows a noncentral $t$-distribution with $d-1$ degrees of freedom and noncentrality parameter $\|\mu\|$.

Going back to Eq.~\eqref{eq:yprodmu}, we have shown the equality:
\begin{equation}
\label{eq:yprodmu_T}
\frac{y^\top \mu}{\|y\| \|\mu\|} = \frac{T}{\sqrt{T^2 + d-1}}.
\end{equation}

We now use the distribution of $T$ to find a high-probability lower bound on $\frac{y^\top \mu}{\|y\| \|\mu\|}$. 

Let
\begin{equation}
\label{eq:snu_mu_def}
s_{\tau}(\|\mu\|) = \sqrt{\frac{1}{(d-1)/q^2_{\mathrm{nct}_{d-1}(\|\mu\|)}(\tau) + 1}} \cdot \mathrm{sgn}\left(q_{\mathrm{nct}_{d-1}(\|\mu\|)}(\tau)\right),
\end{equation}
where we let $q_{\mathrm{nct}_{d-1}(\|\mu\|)}(\tau) = F_{\mathrm{nct}_{d-1}(\|\mu\|)}^{-1}(\tau)$.

We argue that
\begin{equation}
\label{eq:yprodmu_snu}
    P_\mu\left\{\frac{y^\top \mu}{\|y\| \|\mu\|} \leq s_\tau(\|\mu\|) \right\} \leq \tau.
\end{equation}
First note that by Eq.~\eqref{eq:yprodmu_T} we have
\begin{align*}
    P_\mu\left\{\frac{y^\top \mu}{\|y\| \|\mu\|} \leq s_\tau(\|\mu\|) \right\} &= P_\mu\left\{\frac{T}{\sqrt{T^2 + d - 1} } \leq s_\tau(\|\mu\|) \right\}.
\end{align*}
Now we split the analysis into two cases. First consider the case when $q_{\mathrm{nct}_{d-1}(\|\mu\|)}(\tau) <0$; then we have
\begin{align*}
    P_\mu\left\{\frac{T}{\sqrt{T^2 + d - 1} } \leq s_\tau(\|\mu\|) \right\}
    &= P_\mu\left\{\frac{T^2}{T^2 + d - 1 } \geq s_\tau^2(\|\mu\|), T < 0 \right\} = P_\mu\left\{T^2 \geq \frac{s_\tau^2(\|\mu\|)(d-1)}{1 - s_\tau^2(\|\mu\|)}, T < 0 \right\} \\
    &= P_\mu\left\{T \leq - \frac{|s_\tau(\|\mu\|)|\sqrt{d-1}}{\sqrt{1 - s_\tau^2(\|\mu\|)}} \right\} = P_\mu\left\{T \leq - |q_{\mathrm{nct}_{d-1}(\|\mu\|)}(\tau)| \right\}\\
    &= P_\mu\left\{T \leq q_{\mathrm{nct}_{d-1}(\|\mu\|)}(\tau) \right\} = \tau.
\end{align*}
When $q_{\mathrm{nct}_{d-1}(\|\mu\|)}(\tau) \geq 0$, we have
\begin{align*}
    P_\mu\left\{\frac{T}{\sqrt{T^2 + d - 1} } \leq s_\tau(\|\mu\|) \right\}
    &= P_\mu\left\{\frac{T^2}{T^2 + d - 1 } \leq s_\tau^2(\|\mu\|) \right\} + P_\mu\left\{\frac{T^2}{T^2 + d - 1 } > s_\tau^2(\|\mu\|), T \leq 0 \right\}\\
    &= P_\mu\left\{T^2 \leq \frac{s_\tau^2(\|\mu\|)(d-1)}{1 - s_\tau^2(\|\mu\|)} \right\} + P_\mu\left\{T^2 > \frac{s_\tau^2(\|\mu\|)(d-1)}{1 - s_\tau^2(\|\mu\|)}, T \leq 0 \right\}\\
    &= P_\mu\left\{T \leq \frac{s_\tau(\|\mu\|)\sqrt{d-1}}{\sqrt{1 - s_\tau^2(\|\mu\|)}} \right\} = P_\mu\left\{T \leq q_{\mathrm{nct}_{d-1}(\|\mu\|)}(\tau) \right\} = \tau.
\end{align*}
Putting everything together, we have proved \eqref{eq:yprodmu_snu}. Therefore, for any $\tau$, it is valid to take the acceptance region $A_\tau(\mu)$ to be:
$$A_\tau(\mu) = \left\{y: \frac{y^\top \mu}{\|y\| \|\mu\|} \geq s_{\tau}(\| \mu\|) \right\} = \left\{y: \angle(y,\mu)\leq \arccos s_{\tau}(\|\mu\|) \right\},$$
where $s_\tau(\|\mu\|)$ is given in Eq.~\eqref{eq:snu_mu_def}. The set of plausible targets under $\mu$ is
$$\Gamma_\tau(\mu) = \{\gamma\in\mathcal{S}^{d-1}: \angle(\gamma,\mu)\leq \arccos s_{\tau}(\|\mu\|) \}.$$
Notice that $q^\alpha_{\Gamma_{\nu/2}(\mu)} = q_\angle^{\alpha}(\arccos s_{\tau}(\|\mu\|))$ is a valid simultaneous correction over $\Gamma_\tau(\mu)$; formally, for any $\mu$ we have
$$P_{\mu}\left\{\sup_{\gamma\in \Gamma_\tau(\mu) }|(y-\mu)^\top \gamma| \geq  q_\angle^{\alpha}(\arccos s_{\tau}(\|\mu\|)) \right\} \leq \alpha.$$
As per Proposition \ref{prop:local_posi_symmetric}, we want a high-probability upper bound on $q_\angle^{\alpha}(\arccos s_{\tau}(\|\mu\|))$. We obtain one by estimating a lower bound on $\|\mu\|$, which in turns lower bounds $s_{\tau}(\|\mu\|)$. We use $\nu/2$ of the error budget for this task. Notice that $\|y\|^2$ follows a noncentral $\chi_d^2$-distribution with noncentrality parameter $\|\mu\|^2$; let
$$\mu_{\tau}^\mathrm{LB} = \sqrt{\sup\left\{c: F_{\mathrm{nc}\chi^2_d(c)}(\|y\|^2) \geq 1-\tau \right\}};$$
by construction we then know that $P_\mu\{\|\mu\| > \mu_{\tau}^\mathrm{LB}\} \geq 1-\tau$.


Therefore, if we take $\hat q = q_\angle^{(\alpha-\nu)}\left(\arccos s_{\nu/2}(\mu_{\nu/2}^\mathrm{LB})\right)$, we have
\begin{align*}
P_\mu\left\{\hat q \leq q_{\Gamma_{\nu/2}(\mu)}^{(\alpha-\nu)}\right\} &= P_\mu\left\{q_\angle^{(\alpha-\nu)}\left(\arccos s_{\nu/2}(\mu_{\nu/2}^\mathrm{LB})\right) \leq q_\angle^{(\alpha-\nu)}\left(\arccos s_{\nu/2}(\|\mu\|)\right)\right\} = \frac{\nu}{2}.
\end{align*}
Invoking Proposition \ref{prop:local_posi_symmetric} completes the proof.
\end{proof}

\subsection{Example \ref{ex:conditional}}
\label{app:example_conditional}

We can write $y = \hat\gamma \cdot R$, where $\hat\gamma = \frac{y}{\|y\|_2}$ and $R = \|y\|_2$. Therefore, the conditional distribution of $y$ given $\{\hat \gamma = \gamma\}$ is equivalent to the one-dimensional distribution of $R$ given $\{\hat \gamma = \gamma\}$. By a change-of-variable argument, one can relate the conditional density of $R$ to the conditional density of $y$:
$$p_R(r) \propto p_y(r \gamma) r^{d-1},$$
where the term $r^{d-1}$ comes from the Jacobian of the transformation. This term follows from the standard argument that derives the $\chi$-distribution from the normal distribution. We use the $\propto$ notation to indicate that we are ignoring the normalizing constant.

Plugging in $p_y(r\gamma)$, we get
$$p_R(r) \propto \exp\left(-\frac{1}{2} (r\gamma - \mu)^\top \Sigma^{-1} (r\gamma - \mu)\right) r^{d-1} \propto \exp\left( - \frac{1}{2} r^2 \gamma^\top \Sigma^{-1}\gamma + r \mu^\top \Sigma^{-1} \gamma \right) r^{d-1}.$$
Therefore, the distribution $p_R(r)$ depends on $\mu$ only through $\mu^\top \Sigma^{-1}\gamma$.

\subsection{Locally simultaneous inference for the LASSO: details}

Before stating the subroutines of Algorithm \ref{alg:local_lasso} for safe and exact screening, we introduce the necessary notation.

We denote by $A_{0}^{+j}(M,s)$ the row in $A_{0}^+(M,s)$ corresponding to the variable $j\in M^c$. We adopt analogous definitions for $A_{0}^{-j}(M,s)$ and $A_{1}^j(M,s)$ (where in the latter case we consider $j\in M$). For $j\in M^c$, we will use $P^{\setminus \{+j\}}(M,s)$ (resp. $P^{\setminus \{-j\}}(M,s)$) to denote the polyhedron $P(M,s)$ with constraint $(A_{0}^{+j}(M,s), b_{0}^{+j}(M,s))$ (resp. $(A_{0}^{-j} (M,s), b_{0}^{-j}(M,s))$) removed. We similarly use $P^{\setminus \{j\}}(M,s)$ for $j\in M$.

We use $(M,s)^{-j}$ to denote the model-sign pair obtained by removing variable $j\in M$ and the corresponding sign. Similarly, we use $(M,s)^{+(j,+1)}$ to denote the model-sign pair obtained by adding variable $j\in M^c$ to $M$ with a positive sign. We use $(M,s)^{+(j,-1)}$ analogously, only the corresponding sign is negative.

\begin{algorithm}[H]
\SetAlgoLined
\SetKwInOut{Input}{input}
\Input{design matrix $X$, outcome vector $y$, current model-sign pair $(M,s)$}
\textbf{output: } safely screened variables $\mathcal{I}_{\mathrm{safe}}(M,s)$\newline
Compute extrapolated solution at $y$, $\beta_{(M,s)}(y) = (X_M^\top X_M)^{-1}(X_M^\top y - \lambda s)$\newline
$\mathcal{I}^+_{\mathrm{safe}}(M,s) \leftarrow \{j\in M: | \beta_{j\cdot (M,s)}(y)| > s_\nu \|e_{j\cdot (M,s)}^\top (X_M^\top X_M)^{-1}\|_1 \}$\newline
$\mathcal{I}^-_{\mathrm{safe}}(M,s) \leftarrow \{j\in M^c: |X_j^\top (y - X_M \beta_{(M,s)}(y))| < \lambda - s_\nu (1 + \|X_j^\top X_M(X_M^\top X_M)^{-1}\|_1)\}$\newline
Return $\mathcal{I}_{\mathrm{safe}}(M,s) \leftarrow \mathcal{I}^+_{\mathrm{safe}}(M,s) \cup \mathcal{I}^{-}_{\mathrm{safe}}(M,s)$
\caption{SafeScreening}
\label{alg:safe_screening}
\end{algorithm}

\vspace{-10pt}

\begin{algorithm}[H]
\SetAlgoLined
\SetKwInOut{Input}{input}
\Input{design matrix $X$, outcome vector $y$, current model-sign pair $(M,s)$, (optionally) safely screened variables $\mathcal{I}_{\mathrm{safe}}(M,s)$}
\textbf{output: } neighboring model-sign pairs $\B(M,s)$\newline
Initialize $\B(M,s) \leftarrow \emptyset$\newline
For all $j\in M^c \setminus \mathcal{I}_{\mathrm{safe}}(M,s)$, compute LASSO constraint $(A_0^{+j}(M,s), b_0^{+j}(M,s)), (A_0^{-j} (M,s), b_0^{-j} (M,s))$\newline
For all $j\in M \setminus \mathcal{I}_{\mathrm{safe}}(M,s)$, compute LASSO constraint $(A_1^{j} (M,s), b_1^{j} (M,s))$\newline
\For{$j\in [d]\setminus \mathcal{I}_{\mathrm{safe}}(M,s)$}
{
\uIf{$j\in M$}
{
\ Solve LP: $\texttt{Val} = \max_z z^\top A_1^{j}(M,s) \text{ s.t. }z\in P^{\setminus \{j\}}(M,s)$\newline
If $\texttt{Val} > b_1^{j} (M,s)$, add $(M,s)^{-j}$ to $\mathcal{B}(M,s)$}
\uElseIf{$j\in M^c$}{
\ Solve LP: $\texttt{Val} = \max_z z^\top A_0^{+j} (M,s) \text{ s.t. }z\in P^{\setminus \{+j\}}(M,s)$\newline
If $\texttt{Val} > b_0^{+j} (M,s)$, add $(M,s)^{+(j,+1)}$ to $\B(M,s)$\newline
Solve LP: $\texttt{Val} = \max_z z^\top A_0^{-j} (M,s) \text{ s.t. }z\in P^{\setminus \{-j\}}(M,s)$\newline
If $\texttt{Val} > b_0^{-j} (M,s)$, add $(M,s)^{+(j,-1)}$ to $\B(M,s)$\newline
}
}
Return $\B(M,s)$
\caption{ExactScreening}
\label{alg:exact_screening}
\end{algorithm}

\subsection{File-drawer problem with conditional inference}
\label{app:file_drawer_conditional}

We include example plots of conditional confidence interval widths in the file-drawer problem. Inference is extremely numerically unstable and leads to significantly wider intervals compared to the competitors.

\begin{figure}[h!]
    \centering
    \includegraphics[width=0.28\textwidth]{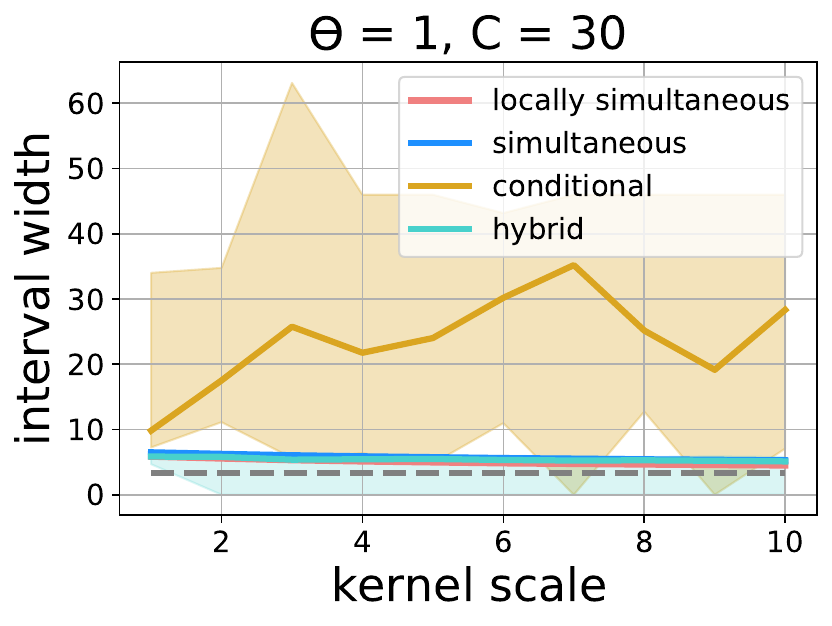}
    \includegraphics[width=0.28\textwidth]{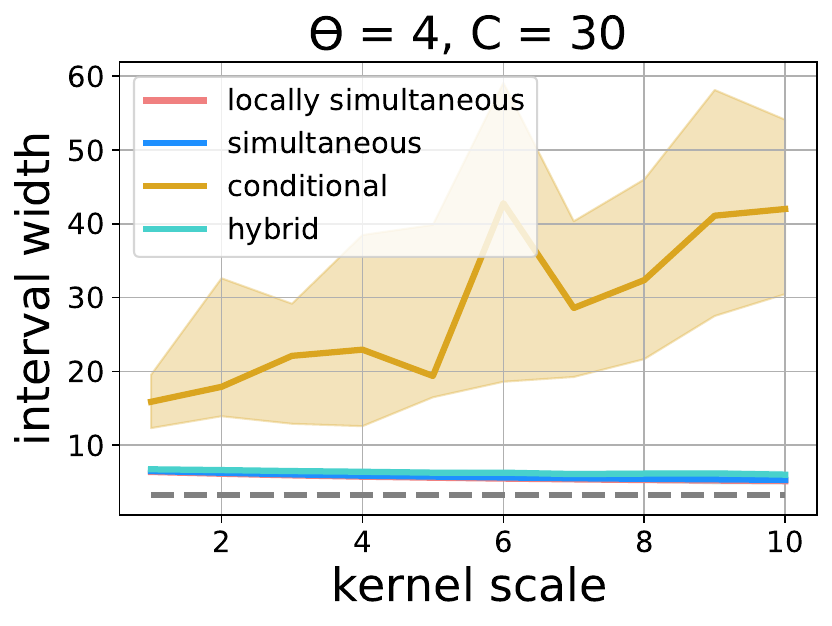}
    \caption{Interval width achieved by locally simultaneous, fully simultaneous, conditional, and hybrid inference in the file-drawer problem.}
\label{fig:filedrawer_conditional}
\end{figure}

\subsection{Effect of the choice of $\nu$}
\label{app:choosing_nu}

To illustrate the consequences of different choices of $\nu$, we consider the same experimental setup as in the parametric case in Section \ref{sec:iow_exp}, this time varying $\nu\in\{0.01\alpha, 0.1\alpha, 0.2\alpha, 0.3\alpha\}$. We fix $\alpha=0.1$ as before. We also fix $m=100$ and vary $\theta$ and $C$.

See Figure \ref{fig:varying_nu} for the comparison, averaged over $100$ trials. We observe that  the relative performance of different choices of $\nu$ depends on both $\theta$ and $C$. For example, for $C=20$ we see that locally simultaneous inference with the smallest value of $\nu$ outperforms the competitors when $\theta$ is large, while for small $\theta$ it is the least favorable of the four options. Generally speaking, when $\nu$ is small fewer candidate selections get filtered out but inference on the remaining candidates is more powerful. When $\nu$ is large there is a greater possibility of reducing the set of plausible candidates, but inference on the remaining candidates is less powerful. In general, the right choice of $\nu$ will depend on how ``densely distributed'' the candidates are. Changing $\nu$ induces relatively little change in the margin $4q^\nu([m])$ in Theorem \ref{thm:gauss_seq_winner}. This difference is only consequential when $\theta$ is small, as in that case there is a sufficiently large $\nu$ that allows filtering out irrelevant alternatives. When $\theta$ is large, most candidates are close to being the winner, meaning filtering out candidates is hard for any $\nu$, and thus it is best to reserve as much error budget as possible for inference.

\begin{figure}[t]
    \centering
    \includegraphics[width=0.32\textwidth]{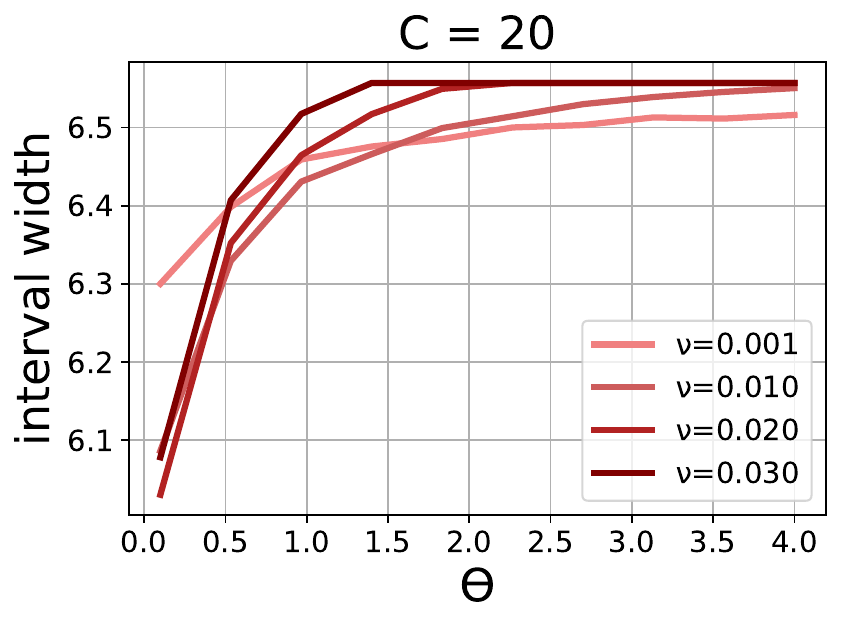}
    \includegraphics[width=0.32\textwidth]{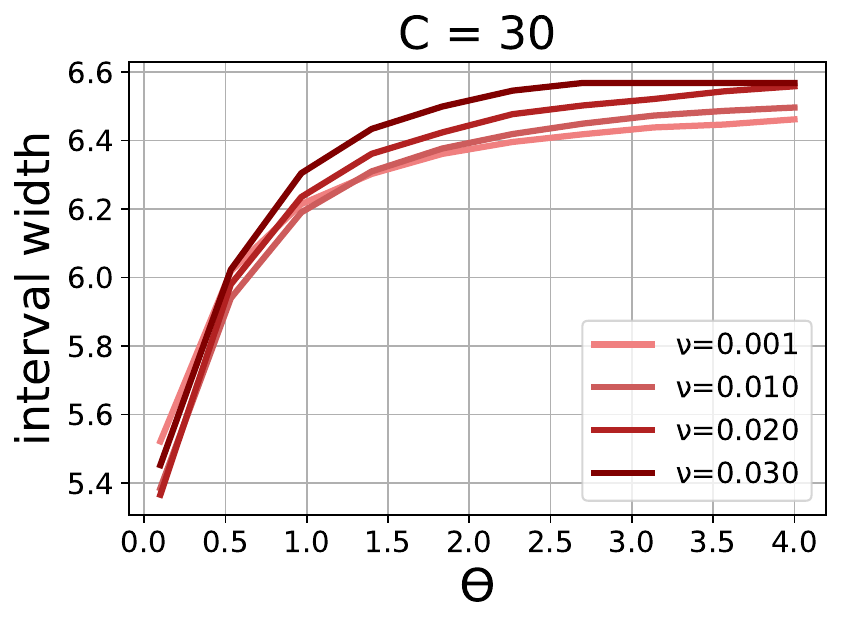}
    \includegraphics[width=0.32\textwidth]{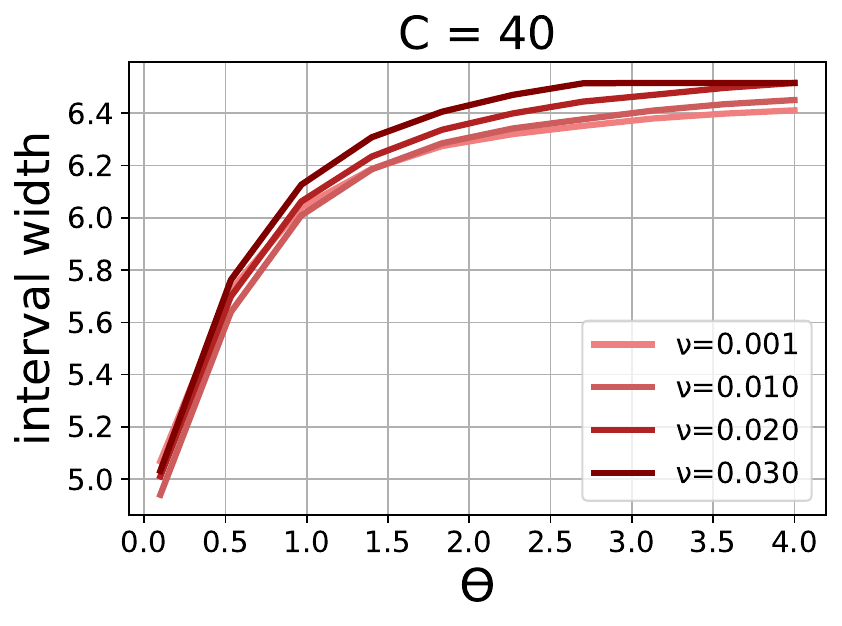}
    \caption{Interval width achieved by locally simultaneous inference for a varying choice of $\nu$.}
\label{fig:varying_nu}
\end{figure}

\end{document}